\newif \ifdoublecolumn
\doublecolumntrue

\ifdoublecolumn
\documentclass[conference,letter]{IEEEtran}

\else
\documentclass[lettersize,onecolumn]{IEEEtran}
\usepackage{algorithmic}
\usepackage{algorithm}
\usepackage{array}
\usepackage[caption=false,font=normalsize,labelfont=sf,textfont=sf]{subfig}
\usepackage{textcomp}
\usepackage{stfloats}
\usepackage{url}
\usepackage{verbatim}
\usepackage{graphicx}
\usepackage{cite}
\fi

\usepackage{mathdots}
\usepackage[utf8]{inputenc} 
\usepackage[T1]{fontenc}
\usepackage{url}
\usepackage{cite}             
\usepackage{placeins}
\usepackage[cmex10]{amsmath}

\newif\ifshortversion
\shortversiontrue  

\DeclareMathOperator*{\essinf}{ess\,inf}

\usepackage{mleftright}       
\mleftright                   

\usepackage{graphicx}         
\usepackage{booktabs}         

\usepackage{soul}
\usepackage{enumitem}
\usepackage{acronym}
\usepackage{xcolor}

\newcommand{\myVec}[1]{{\boldsymbol{#1}}}
\newcommand{\myfun}[1]{{\mathsf{#1}}}

\newcommand{\E}[1]{\mathbb{E}\left[{#1}\right]}	
\newcommand{\Prob}[1]{\mathbb{P}\left[{#1}\right]}	
\newcommand{\aver}[2]{\frac{#1}{n}\sum_{i=1}^n {#2}}	
\newcommand{\product}[1]{\prod_{i=1}^n{#1}}
\newcommand{\var}[1]{\mathrm{Var}[#1]}	
\newcommand{\mySet}[1]{{\mathcal{#1}}}

\acrodef{iid}[i.i.d.]{independent and identically distributed}
\acrodef{rdf}[RDF]{rate-distortion function}
\acrodef{aep}[AEP]{asymptotic equipartition property}

\usepackage[font=small]{caption}
\usepackage{bbm}

\allowdisplaybreaks[2]
\graphicspath{{figures/}}

\usepackage{amsthm}
\usepackage{amsmath}
\usepackage{amsfonts}
\usepackage{cancel}
\usepackage[
  colorlinks=true,
  citecolor=green,
  linkcolor=blue,
  urlcolor=blue
]{hyperref}

\theoremstyle{plain}
\newtheorem{theorem}{Theorem}
\newtheorem{lemma}{Lemma}
\newtheorem{definition}{Definition}
\newtheorem{corollary}{Corollary}

\usepackage{subcaption}

\usepackage{cleveref}
\crefname{equation}{Eq}{Eqs} 
\crefrangelabelformat{equation}{(#3#1#4-#5#2#6)}

\usepackage{amssymb}

\usepackage{float}

\usepackage[font=small]{caption}
\IEEEoverridecommandlockouts

\begin{document}

\title{Dispersion of Gaussian Sources with Memory \thanks{
E. Tasci and V. Kostina are with the Department of Electrical Engineering, California Institute of Technology, Pasadena, CA 91125, USA. E-mails: \textit{\{etasci, vkostina\}@caltech.edu.} This work was presented in part at ISIT 2026~\cite{tasci2026_ISIT}.}}

\author{%
  \IEEEauthorblockN{Eyyüp Taşçı, Victoria Kostina}
}
\maketitle 

\begin{abstract}
We consider finite blocklength lossy compression of information sources with independent but non-identically distributed components. Crucially, Gaussian sources with memory can be cast in this form. We show that under the constraint of exceeding distortion \(d\) with probability at most \(\epsilon\), the minimum achievable rate at blocklength \(n\) satisfies \(R(n, d, \epsilon)=\mathbb{R}_n(d)+\sqrt{\frac{\mathbb{V}_n(d)}{n}}Q^{-1}(\epsilon)+O \left(\frac{\log n}{n}\right)\), where \(Q^{-1}(\cdot)\) is the inverse \(Q\)-function, and $\mathbb{R}_n(d)$ and $\mathbb{V}_n(d)$ are the $n$th-order informational rate-distortion function and source dispersion, respectively, computed using the $n$-letter joint distribution and distortion measure. This generalizes the existing dispersion result for abstract i.i.d. sources. Our analysis introduces three technical novelties. First, we introduce a point-mass product proxy measure that enables typicality analysis by generalizing the empirical distribution beyond the i.i.d. setting. Second, we relax the conventional uncoupled distortion moment assumption to a coupled one, enabling the analysis of nonstationary Gaussian sources. Third, using Lidskii's theorem to approximate the eigenvalues of the $n$th-letter inverse covariance matrix, we establish an $O(1/n)$ convergence rate for $\mathbb{R}_n(d)$ and $\mathbb{V}_n(d)$ toward their limits for Gaussian autoregressive sources, including those with infinite memory. These results sharpen and extend the only known dispersion result for a source with memory, the scalar Gauss-Markov source, to general Gaussian autoregressive models, providing a first unified treatment of stationary and nonstationary sources.
\end{abstract}


\section{Introduction}
Rate-distortion theory studies the minimum coding rate required for data compression when imperfect reproduction within a distortion level $d$ is permitted. The fundamental limits of this trade-off are characterized by \textit{source coding theorems}, which quantify the operational rate-distortion trade-off in terms of informational quantities.

In his seminal paper, Shannon~\cite{shannon59_RD} proved the first source coding theorem in rate-distortion theory, establishing that the minimum rate required to reproduce a discrete memoryless source with average distortion at most $d$ is equal to the informational \ac{rdf} $\mathbb{R}(d)$, which is given by a single-letter mutual-information minimization. Shannon also sketched ideas for generalizing his results to continuous-amplitude sources with memory, but did not prove a coding theorem for that setting. Subsequent work addressed this gap and generalized Shannon's source coding theorem to broader classes of stationary sources. Goblick~\cite{Goblick69_StrongMix} proved a source coding theorem for continuous-amplitude strongly mixing sources, and Berger~\cite{berger68_AbstractAlphabet} treated abstract-alphabet block-ergodic sources. Gallager~\cite[Ch. 9]{gallager68_book} established a coding theorem for finite-alphabet ergodic sources and showed that the asymptotic \ac{rdf} is characterized as the limit of a sequence of $n$-letter mutual-information minimization problems. Berger~\cite[Ch. 7]{berger71_RDBook} later extended this characterization to sources with abstract alphabets. Gray and Davisson~\cite{Gray74_NoErgodic} then showed that the ergodicity assumption can be removed, proving a source coding theorem under stationarity alone.

While proving source-coding theorems of increasing generality was a challenging endeavour, evaluating the \ac{rdf} in closed form presented a challenge of its own. Since the \ac{rdf} is defined through a variational problem whose complexity grows exponentially with the blocklength, explicit formulas are known only in special cases. For a stationary Gaussian processes under quadratic distortion, Kolmogorov~\cite{kolmogorov56_RD} characterized the \ac{rdf} in terms of the source power spectral density by decorrelating the source through a unitary transformation, resulting in the classical \textit{reverse water-filling} formula. Gallager~\cite{gallager68_book} later gave the corresponding source-coding theorem for a stationary Gaussian random processes, establishing the operational meaning of this \ac{rdf}. These works were later extended beyond stationary sources. Berger~\cite{berger70_InformationRatesWiener} treated the Wiener process under quadratic distortion, proving a source-coding theorem and deriving its \ac{rdf}. Generalizing this result, Gray~\cite{gray70_Autoregressive} gave a reverse water-filling expression for the \ac{rdf} of possibly nonstationary Gaussian autoregressive sources and proved a corresponding source-coding theorem. In the same paper, Gray also derived a lower bound for the binary symmetric Markov source under bit-error distortion, which becomes tight in the low-distortion regime, and later extended his results to finite-state, finite-alphabet Markov sources~\cite{gray71_FiniteStateMarkov}. A detailed survey of the development of rate-distortion theory can be found in~\cite{berger98_Survey}.



The classical results above characterize optimal performance as the blocklength $n$ grows without bound, but they do not directly quantify the operational limits at finite blocklength, where practical compression schemes operate. This motivates finite blocklength analyses that quantify the convergence rate of the optimal coding rate to the \ac{rdf} as $n$ increases. Sakrison~\cite{Sakrison68_FiniteGauss} and Wyner~\cite{Wyner68_FiniteGauss}, in separate contributions in the same year, investigated the rate at which the performance of a code can approach the \ac{rdf} as a function of the blocklength for an \ac{iid} Gaussian source under mean-squared error distortion. Zhang, Yang, and Wei~\cite{Yang97_redundancy_finite} studied the redundancy of the minimum rate achieved at finite blocklength $n$ relative to the asymptotically optimal rate for fixed-rate compression of a finite-alphabet \ac{iid} source under an average distortion constraint. Zhang and Yang~\cite{yang99_RedundancyRD} later generalized this result to abstract sources. Kontoyiannis~\cite{kontoyiannis00_PointwiseRedundancy} established refined asymptotics for the pointwise redundancy of abstract sources under variable-length compression in terms of distortion $d$-ball probabilities. For fixed-rate compression of finite-alphabet \ac{iid} sources, second-order asymptotics were derived by Ingber and Kochman~\cite{ingber11_Dispersion}, and for abstract \ac{iid} sources by Kostina and Verdú~\cite{kostina12_FixedLength}.  Second-order asymptotic analyses show that the minimum achievable rate $R(n,d,\epsilon)$ at blocklength $n$ under the operational constraint of exceeding distortion $d$ with probability at most $\epsilon$ admits the following expansion, referred to as the Gaussian approximation~\cite{kostina12_FixedLength, ingber11_Dispersion}:
\begin{equation}
\label{eq:rnd_approx}
    R(n,d,\epsilon)
    \approx \mathbb{R}(d)
    + \sqrt{\frac{\mathbb{V}(d)}{n}}\,Q^{-1}(\epsilon),
\end{equation}
where the first term is the informational \ac{rdf} $\mathbb{R}(d)$, and the second term quantifies the convergence rate to $\mathbb{R}(d)$ through the source characteristic termed source dispersion $\mathbb{V}(d)$, and $Q^{-1}(\cdot)$ denotes the inverse $Q$-function. For sources with memory, the finite blocklength analysis remains limited to the work of Tian and Kostina on stationary and nonstationary Gauss-Markov sources~\cite{tian19_Stationary, tian21_NonStationary}. They provided a water-filling expression for the source dispersion, which parallels that of the \ac{rdf}. More recently, $d$-tilted information for binary Markov sources  and its properties is studied in~\cite{krishnamachari2026}, however it doesn't identify the operational dispersion. See~\cite{motani23_monograph} for a recent monograph on finite blocklength results in lossy compression.

In this paper, we extend the second-order asymptotic analysis of~\cite{kostina12_FixedLength} beyond \ac{iid} sources, and the Gaussian process results of~\cite{tian19_Stationary, tian21_NonStationary} beyond the scalar Gauss-Markov source. We consider abstract sources whose components are independent but not necessarily identically distributed, and, like most prior work in rate-distortion theory, we consider separable distortion measures. We show that the minimum achievable rate admits a Gaussian approximation
\begin{equation}
\label{eq:rnd}
    R(n,d,\epsilon)
    = \mathbb{R}_n(d)
    + \sqrt{\frac{\mathbb{V}_n(d)}{n}}\,Q^{-1}(\epsilon)
    + O\left(\frac{\log n}{n}\right),
\end{equation}
where $\mathbb{R}_n(d)$ and $\mathbb{V}_n(d)$ are the $n$th-order informational \ac{rdf} and source dispersion, respectively, both associated with the $n$-letter mutual information minimization problem. For sources with \ac{iid} components,~\eqref{eq:rnd} reduces to the Gaussian approximation \eqref{eq:rnd_approx} in~\cite{ingber11_Dispersion, kostina12_FixedLength}, except in the zero-dispersion regime, which requires separate treatment. Furthermore, if the $n$th-order \ac{rdf} and the source dispersion converge to their limiting functions sufficiently fast, our analysis yields a Gaussian approximation in terms of these limiting functions.  

Our framework covers Gaussian sources under squared error distortion since any Gaussian source can be decorrelated into independent components by an orthogonal transformation that preserves Euclidean distances. In our preliminary work~\cite{tasci2026_ISIT}, we characterized the second-order term for stationary, finite-memory Gaussian autoregressive sources. In this work, we extend this analysis to a unified framework encompassing both stationary and nonstationary Gaussian autoregressive sources, including those with infinite memory. In particular, for this broader class of Gaussian autoregressive sources, we provide a reverse water-filling representation of the second-order term. This extends the dispersion results for the Gauss-Markov source established in~\cite{tian19_Stationary, tian21_NonStationary}, while also sharpening the remainder term.

A key technical novelty lies in the achievability proof, where we construct a typical set using a point-mass product proxy measure. This choice preserves the additive structure of the distortion and information densities and enables typicality analysis via Berry--Esseen bounds. For \ac{iid} sources, previous works~\cite{kostina12_FixedLength, yang99_RedundancyRD, kontoyiannis00_PointwiseRedundancy} use the empirical measure of the observed sequence to construct a typical set. However, an empirical proxy measure does not facilitate typicality arguments when the source has non-identical components. Tian and Kostina~\cite{tian19_Stationary,tian21_NonStationary} instead use a specially designed maximum likelihood estimator for the gain parameter of the Gauss-Markov source to construct typical sets. Because that approach relies on source-specific estimation, it does not easily generalize to other Gaussian sources with memory or beyond Gaussians.

Another technical difficulty arises in analyzing nonstationary sources, whose components can have variances growing without bound. This causes the conventional uncoupled distortion moment assumption used in previous works~\cite{kostina12_FixedLength, yang99_RedundancyRD} to fail. We overcome this limitation by introducing a weaker coupled moment assumption that allows us to analyze nonstationary Gaussian autoregressive sources. Finally, identifying the operational dispersion requires controlling how fast the $n$th-order \ac{rdf} and informational dispersion converge to their limiting counterparts. For Gaussian autoregressive sources, including those with infinite memory, we establish the required convergence bounds by using Lidskii's theorem and a circulant approximation of the $n$th-letter inverse covariance matrix. This yields $O(1/n)$ convergence rates for $\mathbb{R}_n(d)$ and $\mathbb{V}_n(d)$ toward their limiting values.

The paper is organized as follows. Section~\ref{sec:problem_setup} introduces the operational and information-theoretic definitions. Section~\ref{sec:main_results} states our main Gaussian approximation theorem, applicable to abstract sources with independent but not necessarily identical components along with corollaries that specialize it to three distinct cases: \ac{iid} sources, Gaussian sources with memory, and Gaussian autoregressive sources. We also discuss how the analysis differs between stationary and nonstationary Gaussian autoregressive sources, as well as between sources with finite and infinite memory. The proof of our main theorem is split into converse and achievability parts, presented in Sections~\ref{sec:converse} and~\ref{sec:achievability}. 

Throughout the paper, we use boldface lower-case (upper-case) letters for vectors (random vectors) of length $n$. The $i$th element of a vector $\myVec{X}$ is denoted by $X_i$. For a random variable $X$, we use $\E{X}$ and $\var{X}$ to denote its mean and variance, respectively. We use standard $O(\cdot)$ and $o(\cdot)$ notations to describe the asymptotic growth rates of functions.
Namely, let $f(n)$ and $g(n)$ be two functions of $n$, then $f(n) = O(g(n))$ if and only if there exists a positive real number $M$ and $n_0 \in \mathbb{N}$ such that $|f(n)| \leq M |g(n)|$ for any $ n \geq n_0$. Similarly, $f(n) = o(g(n))$ if and only if for any positive real number $M$, there exists $n_0 \in \mathbb{N}$ such that $|f(n)| \leq M |g(n)|$ for any $ n \geq n_0$. All exponents and logarithms are to base $e$.

\section{Problem Setup}
\label{sec:problem_setup}

\subsection{Operational Definitions}
\label{subsec:operational_defs}
We study sources with independent but not necessarily identical components. Let $\{X_i\}_{i=1}^\infty$ be independent random variables with $X_i \sim P_{X_i}$, each taking values in the alphabet $\mySet{X}_i$. For each blocklength $n\in\mathbb{N}$, define the source block $\myVec{X} \triangleq (X_1,\dots,X_n)\in \boldsymbol{\mySet{X}}$, where the product alphabet is $\boldsymbol{\mySet{X}} \triangleq \product{\mySet{X}_i}$ and the joint distribution factorizes as $P_{\myVec{X}} = \product{P_{X_i}}$. The reproduction sequence takes values in the product alphabet $\boldsymbol{\mySet{Y}} \triangleq \product{\mySet{Y}_i}$. The discrepancy between the source and its reproduction is measured via single-letter distortion functions $\mathsf{d}_i(\cdot, \cdot)\colon\mySet{X}_i \times \mySet{Y}_i \rightarrow [0, +\infty)$, which induce the separable block distortion $\mathsf{d}(\myVec{x}, \myVec{y}) = \aver{1}{\mathsf{d}_i(x_i, y_i)}$.

\begin{definition}
    For a fixed blocklength $n \in \mathbb{N}$, distortion threshold $d  \geq  0$, and excess distortion probability $\epsilon \in [0,1)$, an $(n, M, d, \epsilon)$ code consists of an encoder $\myfun{f} \colon \boldsymbol{\mySet{X}} \rightarrow \{1, \ldots, M \}$ and a decoder $\myfun{g} \colon \{1, \ldots, M \} \rightarrow \boldsymbol{\mySet{Y}}$ such that 
    \begin{equation}
        \Prob{\mathsf{d}(\myVec{X}, \myfun{g}(\myfun{f}(\myVec{X}))) > d } \leq \epsilon.
    \end{equation}
    The rate associated with an $(n, M, d, \epsilon)$ code is $R \triangleq \frac{\log M}{n}.$
\end{definition}
 The minimum achievable code size for given blocklength $n$, distortion $d$, and excess distortion probability $\epsilon \in [0,1)$ is
\begin{equation}
M^\star(n, d, \epsilon) \triangleq \min \{ M \colon\exists \text{ an } (n, M, d, \epsilon) \text{ code} \} ,
\end{equation}
and the corresponding minimum source coding rate is
\begin{equation}
\label{eq:rnd_oper}
R(n, d, \epsilon) \triangleq \frac{\log M^\star(n, d, \epsilon)}{n}.
\end{equation}

\subsection{Informational Definitions}
\label{subsec:informational_defs}
For a given source distribution and distortion measure, the $n$th-order informational \ac{rdf} is
\begin{equation}
\label{eq:rd}
    \mathbb{R}_n(d) \triangleq \inf_{\substack{P_{\myVec{Y}|\myVec{X}} \colon \boldsymbol{\mySet{X}} \mapsto \boldsymbol{\mySet{Y}} \\
    \mathbb{E}[\mathsf{d}(\myVec{X}, \myVec{Y})] \leq d}} \enspace \frac{1}{n}  I(\myVec{X}; \myVec{Y}).
\end{equation}
Here, $I(\myVec{X}; \myVec{Y})$ denotes the mutual information, and the infimum is over all conditional distributions $P_{\myVec{Y}|\myVec{X}}$ satisfying the distortion constraint. We define
\begin{equation}
\label{eq:dmin} 
d_{\min} \triangleq \inf \{d \colon \limsup_{n\rightarrow\infty} \mathbb{R}_n(d) < \infty\}.
\end{equation}
Throughout the paper, we impose the following restrictions:
\begin{enumerate}[label=(\Roman*)]
    \item \label{res:unique}
    The infimum in \eqref{eq:rd} is achieved by a unique conditional distribution $P_{\myVec{Y^\ast}|\myVec{X}}$.
    
    \item \label{res:min_d}
    The distortion level satisfies $d > d_{\min}$.
\end{enumerate}
We refer to $\myVec{Y^\ast}$ as the \ac{rdf}-achieving reproduction random vector. Under restrictions~\ref{res:unique}-\ref{res:min_d},
differentiation of $\mathbb{R}_n(d)$ with respect to $d$ is justified~\cite[Eq.~16]{kostina12_FixedLength}, and we define 
\begin{equation}
\label{eq:lambda}
    \lambda_n^\ast \triangleq -\mathbb{R}_n'(d).
\end{equation}
The associated tilted information density introduced in~\cite[Def.~6]{kostina12_FixedLength} is defined for $d>d_{\min}$ and any $\myVec{x}\in\boldsymbol{\mySet{X}}$ as
\begin{equation}
\label{eq:d_tilted}
\jmath(\myVec{x}, d) \triangleq -\lambda_n^\ast n d 
- \log \mathbb{E} \left[ \exp \left( -\lambda_n^\ast n \mathsf{d}(\myVec{x}, \myVec{Y^\ast})\right) \right],
\end{equation}
where the expectation is with respect to the \ac{rdf}-achieving reproduction distribution $P_{\myVec{Y^\ast}}$. The $d$-tilted information is the key quantity appearing in the nonasymptotic bounds in~\cite{kostina12_FixedLength}, which are also the starting point of our analysis. Since the joint distribution $P_{\myVec{X}}P_{\myVec{Y^\ast}|\myVec{X}}$ factorizes~\cite[Thm. 6.1]{Polyanskiy_Wu_2025} and the distortion is separable, the $d$-tilted information single-letterizes as 
\begin{equation}
\label{eq:tilted_single_letter}
\jmath(\myVec{x},d)=\sum_{i=1}^n \jmath(x_i,d_i), 
\enspace \text{where} \enspace d_i \triangleq \mathbb{E}[\mathsf{d}_i(X_i,Y_i^\ast)].
\end{equation}
The first moment of the $d$-tilted information satisfies the following identity~\cite{csiszar74_ExtremumInfoTheory}:
\begin{equation}
\label{eq:csiszar}
    \mathbb{R}_n(d) = \frac{1}{n}\E{\jmath(\myVec{X}, d)}.
\end{equation}
For finite-alphabet sources,~\eqref{eq:csiszar} follows from the Karush--Kuhn--Tucker (KKT) conditions~\cite[Eq.~(2.5.12)--(2.5.16)]{berger71_RDBook}. Csisz\'ar~\cite{csiszar74_ExtremumInfoTheory} later extended this identity to abstract alphabets.

Because the first moment of the $d$-tilted information determines the \ac{rdf} through~\eqref{eq:csiszar}, the fluctuations of $\jmath(\myVec{X}, d)$ around its mean play a central role in nonasymptotic analysis. These fluctuations are quantified by the $n$th-order informational dispersion, defined as
\begin{equation}
\label{eq:dispersion}
\mathbb{V}_n(d) \triangleq \frac{1}{n}\var{\jmath(\myVec{X}, d)}.
\end{equation}

We also define the limiting informational \ac{rdf} and source dispersion as 
\begin{align}
    \mathbb{R}(d) &\triangleq \lim_{n\rightarrow\infty} \mathbb{R}_n(d), \label{eq:rd_info} \\
    \mathbb{V}(d) &\triangleq \lim_{n\rightarrow\infty} \mathbb{V}_n(d), \label{eq:disp_info}
\end{align}
whenever these limits exist. For Gaussian sources under quadratic distortion, the limiting \ac{rdf} admits a reverse water-filling characterization \cite{kolmogorov56_RD, gray70_Autoregressive, Hashimoto80_RDAutoregressive, Gray08_NoteOnAutoregressive}. In this paper, we derive a reverse water-filling expression for the source dispersion of stationary Gaussian autoregressive sources, paralleling the corresponding \ac{rdf} characterization.

Next, we introduce a related variational problem first introduced by Blahut~\cite{blahut72_ComputationCapacityRD}, in which the reproduction marginal is fixed in advance. This decouples the choice of the output distribution from the choice of the kernel in~\eqref{eq:rd}. Specifically, for a given source distribution $P_{\myVec{X}}$ and a fixed reproduction distribution $P_{\myVec{Y}}$, define
\begin{equation}
\label{eq:crem}
    \mathbb{R}_n(\myVec{X}, \myVec{Y}, d) \triangleq \inf_{\substack{P_{\myVec{Z}|\myVec{X}} \colon \boldsymbol{\mySet{X}} \mapsto \boldsymbol{\mySet{Y}} \\
    \mathbb{E}[\mathsf{d}(\myVec{X}, \myVec{Z})] \leq d}} \enspace \frac{1}{n}  D(P_{\myVec{Z}|\myVec{X}} \| P_\myVec{Y} | P_\myVec{X}).
\end{equation}

Here, $D(P_{\myVec{Z}|\myVec{X}} \| P_\myVec{Y} | P_\myVec{X})$ denotes the conditional relative entropy. The optimization in~\eqref{eq:crem} upper bounds that in~\eqref{eq:rd}, with equality if and only if $P_\myVec{Y} = P_\myVec{Y^\ast}$. Because the reproduction distribution $P_{\myVec{Y}}$ is fixed in advance, $\mathbb{R}_n(\myVec{X}, \myVec{Y}, d)$ is often more convenient to work with than $\mathbb{R}_n(d)$, hence, it has appeared extensively in the nonasymptotic analysis of lossy compression \cite{kontoyiannis00_PointwiseRedundancy, kostina12_FixedLength, tian21_NonStationary, tian19_Stationary, yang99_RedundancyRD}. Operationally, it corresponds to lossy compression of a source distributed according to $P_{\myVec{X}}$ using a random codebook whose codewords are drawn according to $P_{\myVec{Y}}$ \cite{Dembo02_ConditionalRD}.

To characterize the optimizer of~\eqref{eq:crem}, define
\begin{align}
    d_{\min | \myVec{X}, \myVec{Y}} &\triangleq \inf \{d \colon \mathbb{R}_n(\myVec{X}, \myVec{Y}, d) < \infty\}.
\end{align}
For any $d >  d_{\min | \myVec{X}, \myVec{Y}}$, the infimum in~\eqref{eq:crem} is attained by a tilted kernel $P_{\myVec{Z}^\ast|\myVec{X}}$ given by~\cite{csiszar74_ExtremumInfoTheory} 
\begin{equation}
\label{eq:optimum_tilt}
    dP_{\myVec{Z}^\ast|\myVec{X} = \myVec{x}}(\myVec{y}) = \frac{dP_{\myVec{Y}}(\myVec{y}) \exp(-\lambda_{\myVec{X}, \myVec{Y}}^\ast n\mathsf{d}(\myVec{x}, \myVec{y}) )}{\mathbb{E}\left[\exp(-\lambda_{\myVec{X}, \myVec{Y}}^\ast n\mathsf{d}(\myVec{x}, \myVec{Y}) )\right]},
\end{equation}
where $\lambda_{\myVec{X}, \myVec{Y}}^\ast \triangleq -\mathbb{R}'_n(\myVec{X}, \myVec{Y}, d)$ and the expectation is with respect to the reproduction marginal. We now introduce the tilted information corresponding to \eqref{eq:crem} analogously to~\eqref{eq:d_tilted}: for $\myVec{x}\in\boldsymbol{\mySet{X}}$ and $\lambda \ge 0$, we define~\cite{kostina12_FixedLength}
\begin{equation}
\label{eq:generalized_tilted_info}
J_\myVec{Y}(\myVec{x}, \lambda)
\triangleq
- \log \mathbb{E} \left[ \exp \left( -\lambda n \mathsf{d}(\myVec{x}, \myVec{Y}) \right) \right].
\end{equation}
The quantity in~\eqref{eq:generalized_tilted_info} is called generalized $d$-tilted information. It satisfies
\begin{equation}
    \mathbb{R}_n(\myVec{X}, \myVec{Y}, d) = \frac{1}{n}\E{J_\myVec{Y}(\myVec{X}, \lambda_{\myVec{X}, \myVec{Y}}^\ast)} - \lambda_{\myVec{X}, \myVec{Y}}^\ast d.
\end{equation}
In particular, choosing $P_\myVec{Y} = P_\myVec{Y^\ast}$ in~\eqref{eq:crem} yields 
\begin{equation}
    \jmath(\myVec{x}, d)  =  - \lambda_n^\ast n d + J_\myVec{Y^\ast}(\myVec{x}, \lambda_n^\ast).
\end{equation}

For subsequent analysis, we also introduce the following tilted distortion moments. For each $i\leq n$, integers $j, k\geq 1$, and $\lambda \geq 0$ define
\begin{align}
\mathsf{\bar d}_{Y_i,j}(x_i, {\lambda}) &\triangleq \frac{\mathbb{E}\left[\mathsf{d}^j_i(x_i, Y_i)\exp(-\lambda \mathsf{d}_i(x_i, Y_i) )\right]}{\mathbb{E}\left[\exp(-\lambda \mathsf{d}_i(x_i, Y_i) )\right]}, \label{eq:dbar} \\
\mathsf{\bar d}_{\myVec{Y},j, k}(\myVec{x}, {\lambda}) &\triangleq \aver{1}{\mathsf{\bar d}_{Y_i,j}^k(x_i, {\lambda}) }.\label{eq:dbar_block}
\end{align}

In particular, for $j= k =1$, taking the expectation of~\eqref{eq:dbar_block} with respect to $P_{\myVec{X}}$ yields the distortion at the point on the generalized rate-distortion curve with slope $\lambda$.

\subsection{Operational Fundamental Limits}
Source coding theorems~\cite{shannon59_RD, berger68_AbstractAlphabet, gallager68_book, Goblick69_StrongMix, berger71_RDBook, Gray74_NoErgodic} characterize the asymptotic relationship between the
operational \ac{rdf} and its
informational counterpart. In particular, for sources for which the source coding theorem is established, for every $\epsilon \in (0,1)$ and $d > d_{\min}$, we have~\cite[Eq. 40]{kostina12_FixedLength}
\begin{equation}
    \lim_{n\rightarrow\infty} R(n,d,\epsilon) = \mathbb{R}(d),
\end{equation}
where $R(n,d,\epsilon)$ is defined in~\eqref{eq:rnd_oper}, and $\mathbb{R}(d)$ is defined in~\eqref{eq:rd_info}. However, this first-order characterization does not quantify the speed of convergence.

Kostina and Verdú~\cite{kostina12_FixedLength} provide second-order asymptotics for \ac{iid} sources by deriving the Gaussian approximation in~\eqref{eq:rnd_approx}, which shows that the second-order term scales as $\frac{1}{\sqrt{n}}$. The optimal constant governing this second-order behavior is captured by a source characteristic called the operational source dispersion. Formally, the operational dispersion is defined as
\begin{equation}
\label{eq:disp_oper}
    V(d) \triangleq \lim_{\epsilon \rightarrow 0} \lim_{n \rightarrow \infty} n \left( \frac{R(n,d, \epsilon) - \mathbb{R}(d)}{Q^{-1}(\epsilon)}\right)^2.
\end{equation}
when the limit exists. Hence, the Gaussian approximation~\eqref{eq:rnd_approx} identifies the operational dispersion $V(d)$ defined in~\eqref{eq:disp_oper} as the informational dispersion $\mathbb{V}(d)$ in~\eqref{eq:disp_info} for \ac{iid} sources. While the class of sources for which the informational and operational dispersions coincide is not fully understood, we show that they coincide for Gaussian autoregressive sources.

\section{Main Results}
\label{sec:main_results}

We first present a general Gaussian approximation theorem for independent, but not necessarily identically distributed sources over abstract alphabets. We consider sources satisfying the following regularity assumptions:

\begin{enumerate}[label=(\roman*)]
\item
\label{ass:bdd_disp}
The dispersion is uniformly bounded away from zero for sufficiently large $n\in \mathbb{N}$:
\begin{equation}
0<\liminf_{n\rightarrow\infty}\mathbb{V}_n(d)
\end{equation}

\item
\label{ass:up_bdd_disp}
There exists a constant $K_{0} < \infty$ such that the average third absolute central moment of $d$-tilted information is uniformly bounded:
\begin{equation}
    \frac{1}{n}\sum_{i=1}^{n}
    \mathbb{E}\left[\left|\jmath(X_i,d_i)-\mathbb{E}\left[\jmath(X_i,d_i)\right]\right|^3\right] < K_{0}
\end{equation}

\item
\label{ass:uni_moments}
There exist constants $\kappa_{1} > 0$ and $K_{1}<\infty$ such that, for all sufficiently large $n\in\mathbb{N}$ we have $\kappa_{1} \leq \lambda_n^\ast$ and
\begin{equation}
\label{eq:distortion_upper_bound}
\mathbb{E}\left[
    \mathsf{\bar d}_{\myVec{Y}^\ast,6,1}(\myVec{X},\lambda_n^\ast - \kappa_{1}) \right] \le K_{1},
\end{equation}
where the expectation is with respect to the source distribution $P_{\myVec{X}}$. 

\item
\label{ass:uni_curv}
There exists a constant $K_2 < \infty$ such that for all sufficiently large $n \in \mathbb{N}$
\begin{equation}
\label{eq:ass_bdd_curv_rd}
    \mathbb{R}_n''(\myVec{X}, \myVec{Y}^\ast, d) < K_2
\end{equation}
where $(\cdot)'$ denotes differentiation with respect to $d$.
\end{enumerate}

The main result of this paper is that, under these assumptions, the minimum achievable rate admits the following Gaussian approximation.

\begin{theorem}[Gaussian approximation]
\label{thm:gaussian_approximation}
Consider a source $\{X_i\}_{i=1}^\infty$ with independent but not necessarily identically distributed components under a separable distortion measure. Suppose that assumptions~\ref{ass:bdd_disp}--
\ref{ass:uni_curv} 
hold. Then for every $\epsilon\in(0,1)$, the minimum achievable rate satisfies
\begin{equation}
\label{eq:thm_rnd}
R(n,d,\epsilon) = \mathbb{R}_n(d) + \sqrt{\frac{\mathbb{V}_n(d)}{n}}Q^{-1}(\epsilon)
+ O\left(\frac{\log n}{n}\right),
\end{equation}
where $\mathbb{R}_n(d)$ is the informational $n$th-order \ac{rdf} defined in~\eqref{eq:rd} and $\mathbb{V}_n(d)$ is the source dispersion in~\eqref{eq:dispersion}.
\end{theorem}
The proof of Theorem~\ref{thm:gaussian_approximation} consists of a converse and an achievability part, presented in Sections~\ref{sec:converse} and~\ref{sec:achievability}, respectively. Although the $n$th-order \ac{rdf} and the source dispersion in~\eqref{eq:thm_rnd} are defined using only the first $n$ components of the source, the remainder term relies on regularity assumptions~\ref{ass:bdd_disp}--
\ref{ass:uni_curv} 
imposed on the entire process. This is necessary because Theorem~\ref{thm:gaussian_approximation} is an asymptotic statement, and the proof of the theorem requires uniform bounds over blocklengths $n$.

\textit{Remarks on assumptions:} A standard assumption in the literature~\cite{kostina12_FixedLength,yang99_RedundancyRD}, also used in our preliminary work~\cite{tasci2026_ISIT}, requires a finite distortion moment under the product distribution $P_{\myVec{X}}\times P_{\myVec{Y}^\ast}$:
\begin{align}
\label{eq:uncoupled_moment}
\mathbb{E}\left[
    \mathsf{\bar d}_{\myVec{Y}^\ast,j,1}(\myVec{X},0)
\right]
&=
\frac{1}{n}\sum_{i=1}^n
\mathbb{E}
\left[\mathsf{d}_i^j(X_i,Y_i^\ast)\right]
< K_{1}'
\end{align}
for a sufficiently large integer $j$. \footnote{For example, $j=12$ in~\cite{yang99_RedundancyRD}, and $j=9$
in~\cite{kostina12_FixedLength}. In our preliminary work~\cite{tasci2026_ISIT}, we relaxed the requirement to $j=6$.} Although equation~\eqref{eq:uncoupled_moment} is easier to verify than assumption~\ref{ass:uni_moments}, it is also more restrictive as shown in Appendix~\ref{app:discussion_assumption}. In particular, it fails for nonstationary Gaussian sources. We therefore adopt the weaker condition stated in assumption~\ref{ass:uni_moments}. An alternative, stronger but often easier to verify condition than assumption~\ref{ass:uni_moments} is the following:
\begin{enumerate}[label=(\roman**), start = 3]
\item \label{ass:uni_moments*} There exist constants $\kappa_{1}' > 0$ and $K_{1}'<\infty$ such that, for all sufficiently large $n\in\mathbb{N}$ we have $\kappa_{1}' \leq \lambda_n^\ast$ and
\begin{equation}
\frac{1}{n}\sum_{i=1}^n \mathbb{E}\left[\exp(\kappa_{1}' \mathsf{d}_i(X_i,Y_i^\ast))\right] \le K_{1}',
\end{equation}
where the expectation is with respect to the joint optimal distribution $P_{X_i, Y_i^\ast}$
\end{enumerate}
We show in Appendix~\ref{app:discussion_assumption}
that condition~\ref{ass:uni_moments*} implies assumption~\ref{ass:uni_moments}. Moreover, the assumptions~\ref{ass:bdd_disp}-\ref{ass:up_bdd_disp} and~\ref{ass:uni_curv} are specific to this paper and they control the asymptotic behavior of the source. Regarding assumption~\ref{ass:uni_curv}, one should note that although $\mathbb{R}_n(d) = \mathbb{R}_n(\myVec{X}, \myVec{Y}^\ast, d)$ and $\mathbb{R}'_n(d) = \mathbb{R}'_n(\myVec{X}, \myVec{Y}^\ast, d)$, the equality $\mathbb{R}''_n(d) = \mathbb{R}''_n(\myVec{X}, \myVec{Y}^\ast, d)$ does not necessarily hold, and hence the bound can not be checked through the curvature of the \ac{rdf}.

The rest of this section derives three consequences of Theorem~\ref{thm:gaussian_approximation}, yielding Gaussian approximation results for \ac{iid} sources over abstract alphabets, Gaussian sources with memory, and Gaussian autoregressive sources.

\subsection{Independent and Identically Distributed Sources}

In the \ac{iid} setting with nonzero dispersion, Theorem~\ref{thm:gaussian_approximation} recovers the result of~\cite[Thm.~12]{kostina12_FixedLength} under a weaker set of assumptions
\begin{corollary}
\label{corr:iid}
    Consider an \ac{iid} source $\{X_i\}_{i=1}^\infty$ under distortion measure $\mathsf{d}(\myVec{x}, \myVec{y}) = \aver{1}{\mathsf{d}(x_i, y_i)}$ which is separable with identical single-letter distortion functions. Suppose that assumptions~\ref{ass:bdd_disp}--
    \ref{ass:uni_curv}
    hold. Then for every $\epsilon\in(0,1)$, the minimum achievable rate satisfies
    \begin{equation}
    R(n,d,\epsilon) = \mathbb{R}(d) + \sqrt{\frac{\mathbb{V}(d)}{n}}Q^{-1}(\epsilon)
    + O\left(\frac{\log n}{n}\right).
    \end{equation}
\end{corollary}
\begin{proof}
The result follows by applying Theorem~\ref{thm:gaussian_approximation}, and noting that for an \ac{iid} source and a separable distortion measure, the $n$-letter \ac{rdf} and source dispersion reduce to $\mathbb{R}_n(d) = \mathbb{R}(d)$ and $\mathbb{V}_n(d) = \mathbb{V}(d)$ for all $n \in \mathbb{N}$. 
\end{proof}

Comparing Corollary~\ref{corr:iid} with~\cite[Thm.~12]{kostina12_FixedLength}, we recover the same first and second-order terms, and our analysis yields the same bounds on the remainder, although we do not state them explicitly. The main difference lies in the regularity assumptions. In particular, \cite[Thm.~12]{kostina12_FixedLength} imposes the stronger uncoupled ninth-order moment condition~\eqref{eq:uncoupled_moment}.

Moreover, the upper distortion threshold $d_{\max}$ in~\cite[Thm.~12]{kostina12_FixedLength} can be generalized to our setting by defining
\begin{equation}
\bar d_n \triangleq \inf_{\myVec{y}\in\boldsymbol{\mySet{Y}}} \E{\mathsf{d}(\myVec{X}, \myVec{y})},
\qquad
d_{\max}\triangleq \liminf_{n\rightarrow\infty} \bar d_n.
\end{equation}
We show in Appendix~\ref{app:discussion_assumption} that condition~\eqref{eq:uncoupled_moment} implies assumption~\ref{ass:up_bdd_disp} and, for $d < d_{\max}$, it also implies assumptions~\ref{ass:uni_moments}-\ref{ass:uni_curv}, even for more general sources with independent but not necessarily identically distributed components.

We also relax the uncoupled ninth-order moment requirement in~\cite[Thm.~12]{kostina12_FixedLength} to a sixth-order coupled moment assumption in~\ref{ass:uni_moments}. We achieve the reduction in the order by applying Chebyshev's inequality, rather than the Berry--Esseen Theorem, to establish concentration in the typical-set analysis in achievability. Finally, the assumption~\ref{ass:bdd_disp} excludes the zero-dispersion regime, which is handled separately in~\cite{kostina12_FixedLength} both in the achievability and converse proofs.

\subsection{Gaussian Sources with Memory}
\label{subsec:decor}
To specialize Theorem~\ref{thm:gaussian_approximation} to Gaussian sources with memory, we first recall the standard decorrelation technique, which reveals a useful equivalence between lossy compression of Gaussian sources with memory and that of independent but not necessarily identically distributed Gaussian sources.

Fix $n\in\mathbb{N}$ and let $\myVec{X} \sim \mathcal{N}(0, \mathsf{\Sigma}_n)$ where $\mathsf{\Sigma}_n \in \mathbb{R}^{n\times n}$ is the covariance matrix with eigendecomposition 
\begin{equation}
    \mathsf{\Sigma}_n = \mathsf{Q}_n^{\top}\mathsf{\Lambda}_n \mathsf{Q}_n.
\end{equation} Here, $\mathsf{Q}_n \in \mathbb{R}^{n\times n}$ is an orthogonal matrix and $\mathsf{\Lambda}_n = \mathrm{diag}(\sigma_{n,1}^2, \ldots, \sigma_{n,n}^2)$. Let 
\begin{equation}
    \label{eq:decor}
    \myVec{U} \triangleq \mathsf{Q}_n \myVec{X}
\end{equation}
be the decorrelated process. Since $\mathsf{\Lambda}_n$ is diagonal, $\myVec{U} \sim \mathcal{N}(0, \mathsf{\Lambda}_n)$ has independent Gaussian components. Furthermore, because orthogonal transformations preserve the squared error, the distortion is invariant under $\mathsf{Q}_n$. Thus, any $(n, M, d, \epsilon)$ code for $\myVec{U}$ can be transformed into a code for $\myVec{X}$ with the same parameters, and this equivalence also
extends to the $d$-tilted information: for every
$\myVec{x}\in\mathbb{R}^n$ and $\myVec{u}=\mathsf{Q}_n\myVec{x}$,
\begin{equation}
\label{eq:equivalence}
    \jmath(\myVec{u}, d) = \jmath(\myVec{x}, d).
\end{equation}
Applying Theorem~\ref{thm:gaussian_approximation} to this decorrelated representation yields the following corollary, which generalizes \cite[Thm.~12]{kostina12_FixedLength} to a class of Gaussian sources with memory.

\begin{corollary} 
\label{corr:gaussian}
Consider a zero-mean Gaussian process $\{X_i\}_{i=1}^\infty$ under mean-squared error distortion. Suppose that assumptions~\ref{ass:bdd_disp}--
\ref{ass:uni_curv}  
hold. Let $\{\sigma_{n,i}^2\}_{i=1}^n$ denote the eigenvalues of the $n$-letter covariance matrix of the source and let $\theta^\ast_n > 0$ be the water level chosen to satisfy
\begin{equation}
\label{eq:corr_distortion}
d = \frac{1}{n} \sum_{i=1}^n \min(\theta^\ast_n, \sigma_{n,i}^2).
\end{equation}
Then for every $\epsilon\in(0,1)$, the minimum achievable rate satisfies
\begin{equation}
    \label{eq:corr_rnd}
    R(n,d,\epsilon) = \mathbb{R}_n(d) + \sqrt{\frac{\mathbb{V}_n(d)}{n}} Q^{-1}(\epsilon) + O\left(\frac{\log n}{n}\right),
\end{equation}
where the \ac{rdf} $\mathbb{R}_n(d)$ and the source dispersion $\mathbb{V}_n(d)$ are given by reverse water-filling 
\begin{align}
    \mathbb{R}_n(d) &= \frac{1}{n} \sum_{i=1}^n \frac{1}{2} \log \max \left( 1, \frac{\sigma_{n,i}^2}{\theta^\ast_n} \right), \label{eq:corr_rd} \\
    \mathbb{V}_n(d) &= \frac{1}{n} \sum_{i=1}^n \frac{1}{2} \min \left( 1, \left( \frac{\sigma_{n,i}^2}{\theta^\ast_n} \right)^2 \right). \label{eq:corr_disp}
\end{align} 
\end{corollary}
\begin{proof}
By the decorrelation equivalence established in the paragraph preceding the corollary, it suffices to prove the claim for the vector $\myVec{U}$ defined in~\eqref{eq:decor}, where $\myVec{U}\sim\mathcal{N}(0,\mathsf{\Lambda}_n)$ has independent components with variances $\sigma_{n,1}^2,\ldots,\sigma_{n,n}^2$.

Let $\myVec{V}^\ast$ be the \ac{rdf}-achieving reproduction random vector for the decorrelated source $\myVec{U}$. It is well-known~\cite[Thm. 10.3.3]{cover2006_book} that $\myVec{V^\ast}$ has also independent coordinates and satisfies $V_i^\ast \sim \mathcal{N}(0, \max(0, \sigma_{n,i}^2 - \theta_n^\ast))$, so that the distortion allocated to the $i$th component is
\begin{equation}
    d_i = \min(\theta^\ast_n, \sigma_{n,i}^2).
\end{equation}

It remains to identify the $n$th-order \ac{rdf} and the source dispersion. The $d$-tilted information for each component can be computed explicitly from the distribution of $\myVec{V}^{\ast}$, as characterized in
\cite[Lem.~7]{tian19_Stationary}
\begin{equation}
\label{eq:tilted_gauss}
    \jmath(u_i, d_i) = \frac{\min(\theta^\ast_n, \sigma_{n,i}^2)}{2 \theta^\ast_n} \left( \frac{u_i^2}{\sigma_{n,i}^2} - 1 \right) + \frac{1}{2} \log \frac{\max\left( \theta^\ast_n, \sigma_{n,i}^2 \right)}{\theta^\ast_n}.
\end{equation}
Substituting this into the single-letterized form in~\eqref{eq:tilted_single_letter} with $x_i$ replaced by $u_i$ and taking the expectation and variance yields~\eqref{eq:corr_rd} and \eqref{eq:corr_disp}, respectively. Applying Theorem~\ref{thm:gaussian_approximation} to $\myVec{U}$ yields the expansion in \eqref{eq:corr_rnd}. 
\end{proof}

Corollary~\ref{corr:gaussian} yields a Gaussian approximation expressed directly in terms of the eigenvalues of the $n$-letter covariance matrix. From an operational source coding viewpoint, the result is analogous in spirit to Gray~\cite[Thm.~2]{gray70_Autoregressive}: the corollary should not be interpreted as saying that one can code a long sequence by dividing it into blocks of length $n$ and using the same code on every block, since subsequent blocks need not have the same distribution. Rather it is strictly "one-shot", meaning that while the corollary is about the compression of a length-$n$ source vector, the statement is an asymptotic statement, and the remainder term relies on the regularity assumptions~\ref{ass:bdd_disp}--
\ref{ass:uni_curv} 
imposed on the entire process. Hence, the corollary is a statement about the entire Gaussian process $\{X_i\}_{i=1}^\infty$, not just about the $n$-letter source vector. 
For a given $n$-letter source, if it can be extended to a process $\{X_i\}_{i=1}^\infty$ in such a way that the assumptions are satisfied, then the corollary applies to the extended source process.
\ifdoublecolumn
\begin{figure*}[!t]
    \centering

    \subfloat[$S(\omega)$ in~\eqref{eq:spectrum}.]{%
        \includegraphics[width=0.47\textwidth]{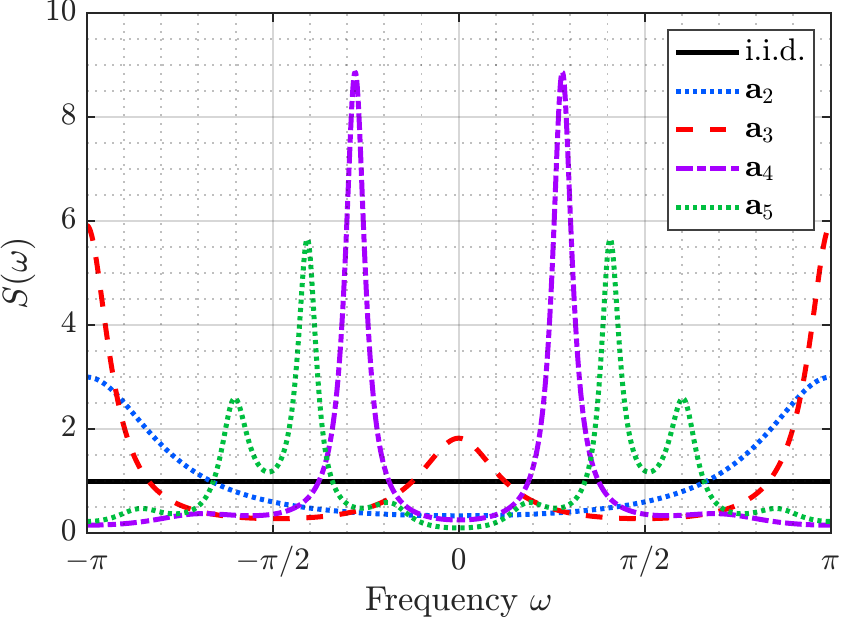}
        \label{fig:spectrum}
    }
    \hfill
    \subfloat[Rate-distortion $\mathbb{R}(d)$ in~\eqref{eq:GAR_RD}.]{%
        \includegraphics[width=0.47\textwidth]{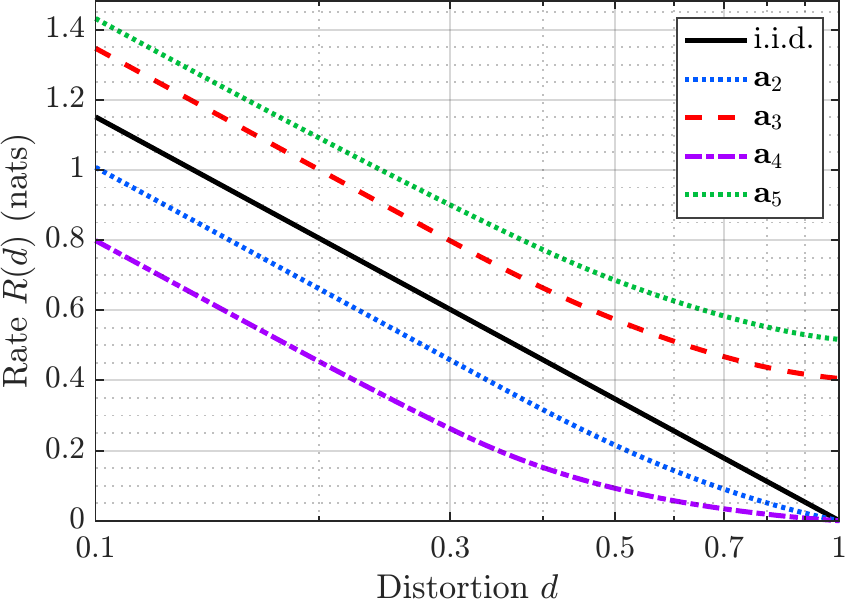}
        \label{fig:rate_distortion}
    }

    \par\medskip

    \subfloat[Dispersion-distortion $\mathbb{V}(d)$ in~\eqref{eq:GAR_VD}.]{%
        \includegraphics[width=0.47\textwidth]{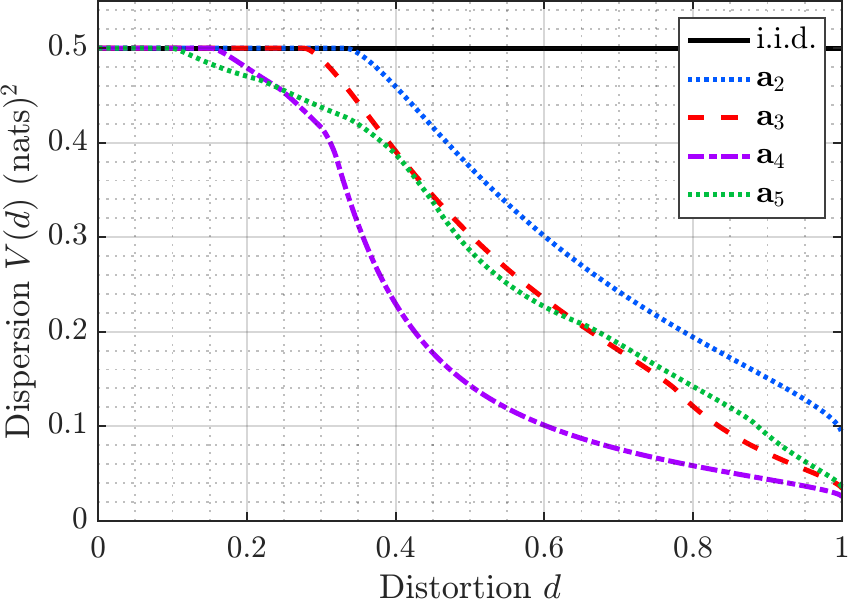}
        \label{fig:dispersion}
    }
    \hfill
    \subfloat[Required blocklength in~\eqref{eq:required_blocklength} to sustain $R = 0.75$ (nats), provided that excess-distortion probability is bounded by $\epsilon$. Color identifies the GAR source, while line style identifies $\epsilon$.]{%
        \includegraphics[width=0.47\textwidth]{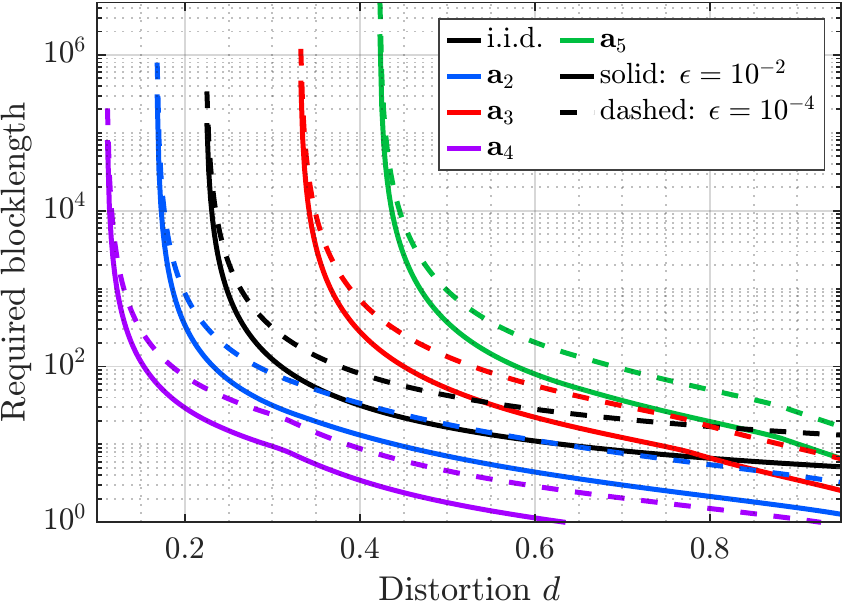}
        \label{fig:gar_required_blocklength}
    }

    \caption{ 
    Comparison of spectrum $S(\omega)$, the corresponding rate-distortion and dispersion functions, and required blocklengths for five Gaussian autoregressive sources with a common maximum distortion $d_{\text{flat}} = 1$: \ac{iid} Gaussian source, AR(1) with $\myVec{a}_2 = 0.5$, AR(2) with $\myVec{a}_3 = [-0.7, -1.2]$, AR(4) with $\myVec{a}_4 = [-0.4, 0.3, 0.2, 0.3]$, and AR(8) with $\myVec{a}_5 = [1, 0.4, 0.5, 0.6, -0.4, 0.8, -0.4, 0.7]$. The corresponding innovation variances are $\sigma_1^2 = 1$, $\sigma_2^2 = 0.75$, $\sigma_3^2 = 1.48$, $\sigma_4^2 = 0.33$, and $\sigma_5^2 = 1.75$. The sources with coefficients $\boldsymbol{a}_3$ and $\boldsymbol{a}_5$ are nonstationary, whereas the remaining sources are stationary.
    }
    \label{fig:four_plots}
\end{figure*}
\else
\fi

\subsection{Gaussian Autoregressive Sources}
Whenever the $n$th-order \ac{rdf} and source dispersion converge to their limiting counterparts at a suitable rate, Theorem~\ref{thm:gaussian_approximation} yields the corresponding limiting second-order characterization. We extend our preliminary results on stationary, finite-memory Gaussian autoregressive sources in~\cite{tasci2026_ISIT} to the more general setting of possibly nonstationary Gaussian autoregressive sources with infinite memory. In particular, we obtain reverse water-filling expressions for both the first- and second-order terms. Consequently, the following corollary identifies the operational dispersion in~\eqref{eq:disp_oper} with the limiting informational dispersion in~\eqref{eq:disp_info} for Gaussian autoregressive sources:
\begin{equation}
    V(d) = \mathbb{V}(d).
\end{equation}

\begin{corollary}
\label{corr:gaussAR}
Consider a scalar Gaussian autoregressive source
$\{X_i\}_{i=1}^\infty$, defined by
\begin{equation}
\label{eq:GaussARSource}
X_{i} =
\begin{cases}
    -\sum_{k=1}^{\infty} a_k X_{i-k} + Z_i, & \qquad i\ge 1, \\
    0, & \qquad i\leq 0,
\end{cases}
\end{equation}
where $\{a_k\}_{k=1}^\infty$ are real numbers such that there exists a constant $A_0$ satisfying
\begin{equation}
\label{eq:coef_bound}
    \sum_{k=1}^\infty k|a_k| < A_0,
\end{equation}
and $Z_i \sim \mathcal{N}(0,\sigma^2)$ are independent noises that form the innovation process. Let the distortion measure be mean-squared error and define
\begin{equation}
\label{eq:g}
g(\omega) \triangleq \left|
\sum_{k=0}^{\infty} a_k e^{-jk\omega}
\right|^2.
\end{equation}
with $a_0 = 1$. Then, for every $\epsilon\in(0,1)$ and $d\in(0, d_{\text{flat}})$ where 
\begin{equation}
\label{eq:dmax}
d_{\text{flat}} = \frac{1}{2\pi} \int_{-\pi}^\pi \frac{\sigma^2}{g(\omega)} d\omega,
\end{equation}
the minimum achievable rate satisfies
\begin{equation}
\label{eq:GAR_rnd}
R(n,d,\epsilon)=\mathbb{R}(d)+\sqrt{\frac{\mathbb{V}(d)}{n}}Q^{-1}(\epsilon)
+O\left(\frac{\log n}{n}\right),
\end{equation}
where
\begin{equation}
\label{eq:GAR_RD}
\mathbb{R}(d) = \frac{1}{2\pi} \int_{-\pi}^\pi \frac{1}{2}\log
\max\left[g(\omega), \frac{\sigma^2}{\theta^\ast} \right] d\omega,
\end{equation}
\begin{equation}
\label{eq:GAR_VD}
\mathbb{V}(d)
= \frac{1}{4\pi}\int_{-\pi}^{\pi}
\min\left[1,\left(\frac{\sigma^2}{\theta^\ast g(\omega)}\right)^{2}\right] d\omega,
\end{equation}
and the water level $\theta^\ast>0$ is defined as the one that satisfies
\begin{equation}
\label{eq:GAR_distortion_lim}
d = \frac{1}{2\pi} \int_{-\pi}^\pi \min\left[\theta^\ast, \frac{\sigma^2}{g(\omega)} \right] d\omega.
\end{equation}
\end{corollary}

The threshold $d_{\mathrm{flat}}$ in~\eqref{eq:dmax} is the distortion level beyond which the limiting rate-distortion function remains constant. It corresponds to the water level reaching the maximum value of $\sigma^2/g(\omega)$.

To prove Corollary~\ref{corr:gaussAR} we need refined asymptotic analysis of the eigenvalues of the covariance matrix. The autoregressive recursion in~\eqref{eq:GaussARSource} can be written in matrix form as $\myVec{Z} = \mathsf{A}_n\myVec{X}$, where $\mathsf{A}_n \in \mathbb{R}^{n\times n}$ is the lower-triangular Toeplitz matrix
\begin{equation}
\label{eq:def_A_n}
\mathsf{A}_n=
\begin{bmatrix}
1    & 0      & \cdots & \cdots & 0 \\
a_1    & 1    & \ddots &        & \vdots \\
\vdots & a_1    & \ddots & \ddots & \vdots \\
a_{n-2}     & \ddots & \ddots & \ddots & 0 \\
a_{n-1}     & a_{n-2}   & \cdots  & a_1    & 1
\end{bmatrix},
\end{equation}
and the covariance matrix of $\myVec{X}$ is given by
\begin{equation}
\label{eq:covariance}
     \mathsf{\Sigma}_n \triangleq \mathsf{A}_n^{-1} \mathbb{E}\left[\myVec{Z}\myVec{Z}^T\right] (\mathsf{A}_n^T)^{-1} = \sigma^2 (\mathsf{A}_n^T \mathsf{A}_n)^{-1}.
\end{equation}
By the decorrelation equivalence established in subsection~\ref{subsec:decor}, an orthogonal transformation that diagonalizes $\mathsf{\Sigma}_n$
converts the source block $\myVec{X}$ into independent Gaussian coordinates whose variances are the eigenvalues of $\mathsf{\Sigma}_n$. Letting $\{\mu_{n,i}\}_{i=1}^n$ denote the eigenvalues of $\mathsf{A}_n^\top \mathsf{A}_n$, we can express the eigenvalues of $\mathsf{\Sigma}_n$ as
\begin{equation}
\label{eq:eig_equivalence}
    \sigma_{n,i}^2  = \frac{\sigma^2}{\mu_{n,i}}, \qquad i=1,\ldots,n . 
\end{equation}

While the eigenvalues $\{\sigma_{n,i}^2\}_{i = 1}^n$ do not generally admit a closed-form expression, we can analyze the asymptotic distribution of $\{\mu_{n,i}\}_{i=1}^n$ because $\mathsf{A}_n^\top \mathsf{A}_n$ is an almost Toeplitz matrix. We use the following refined asymptotic result on eigenvalues to prove Corollary~\ref{corr:gaussAR}.

\begin{lemma}
\label{lemm:GARnon_asymptotic_eigenvalues}
Fix a sequence $\{a_k\}_{k=0}^\infty$ such that~\eqref{eq:coef_bound} is satisfied. Consider the eigenvalues $\{\mu_{n,i}\}_{i=1}^n$ of $\mathsf{A}_n^T\mathsf{A}_n$ with $\mathsf{A}_n$ in~\eqref{eq:def_A_n}, and let $F(t)$ be any $L_F$-Lipschitz function defined on the interval $t \in [\alpha, \beta]$, where
\begin{equation}
\label{eq:alpha_beta}
    \alpha \triangleq \inf_{n\in \mathbb{N}, \enspace i\in[n]} \mu_{n,i},
\qquad
\beta \triangleq \sup_{n\in \mathbb{N}, \enspace i\in[n]} \mu_{n,i}.
\end{equation}
Then, 
\begin{equation}
\label{eq:eig_concentration_rate}
\left|
\frac{1}{n}\sum_{i=1}^n F\left(\mu_{n,i}\right)
-\frac{1}{2\pi}\int_{-\pi}^{\pi} F\big(g(\omega)\big)\,d\omega
\right|
\le \frac{C_F}{n},
\end{equation}
where $g(\omega)$ is defined in~\eqref{eq:g}, and $C_F>0$ is a constant that
depends on $L_F$ and $A_0$ in~\eqref{eq:coef_bound}.
\end{lemma}
\begin{proof}[Proof Idea]
     We approximate the eigenvalues $\{\mu_{n,i}\}_{i=1}^n$ by those of a circulant matrix, whose eigenvalues are given by uniform samples of $g(\omega)$. Using Lidskii's Theorem~\ref{thm:Lidskii} in Appendix~\ref{sec:proof_GAR_eig}, we show that average eigenvalue perturbation is $O(1/n)$. The remainder of the argument follows from elementary Riemann-sum estimates, see Appendix~\ref{sec:proof_GAR_eig} for details.
\end{proof}

\begin{proof}[Proof of Corollary~\ref{corr:gaussAR}]
The proof consists of checking that the assumptions~\ref{ass:bdd_disp}--
\ref{ass:uni_curv}
hold for all $d\in(0,d_{\text{flat}})$, so that~\eqref{eq:corr_rnd} in Corollary~\ref{corr:gaussian} applies,  and showing that there exist constants $C_r, C_v > 0$ such that
\begin{equation}
\label{eq:concentration_rvd}
\bigl|\mathbb{R}_n(d)-\mathbb{R}(d)\bigr|\le \frac{C_r}{n},
\qquad
\left|\sqrt{\mathbb{V}_n(d)}-\sqrt{\mathbb{V}(d)}\right|\le \frac{C_v}{\sqrt{n}},
\end{equation}
where $\mathbb{R}(d)$ and $\mathbb{V}(d)$ are given by~\eqref{eq:GAR_RD}, \eqref{eq:GAR_VD}. We establish~\eqref{eq:concentration_rvd} by using Lemma~\ref{lemm:GARnon_asymptotic_eigenvalues} and by bounding the rate at which $\theta_n^\ast$ converges to $\theta^\ast$. Substituting~\eqref{eq:concentration_rvd} into~\eqref{eq:thm_rnd} yields~\eqref{eq:GAR_rnd}. See Appendix~\ref{app:proof_gaussAR} for full details. 
\end{proof}

\ifdoublecolumn
\else
\begin{figure*}[!t]
    \centering

    \subfloat[$S(\omega)$ in~\eqref{eq:spectrum}.]{%
        \includegraphics[width=0.47\textwidth]{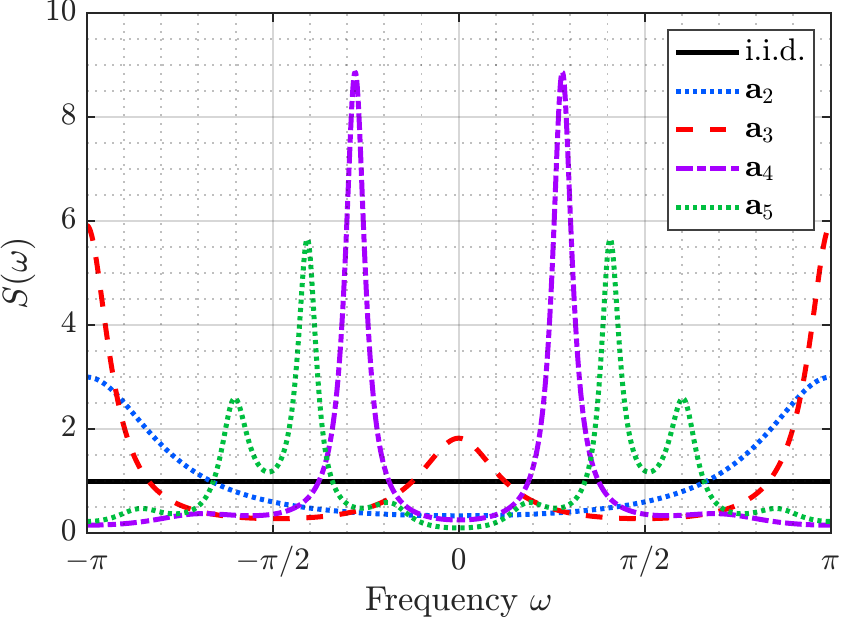}
        \label{fig:spectrum}
    }
    \hfill
    \subfloat[Rate-distortion $\mathbb{R}(d)$ in~\eqref{eq:GAR_RD}.]{%
        \includegraphics[width=0.47\textwidth]{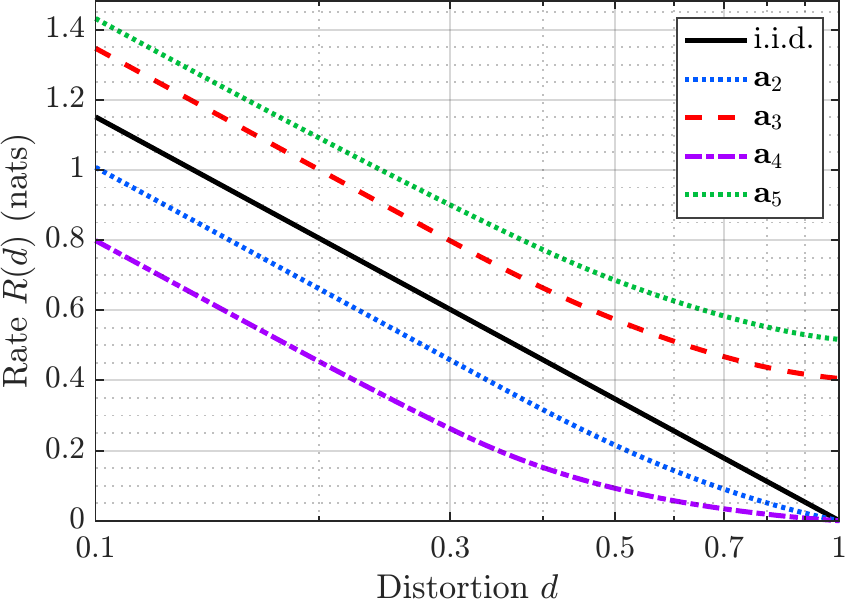}
        \label{fig:rate_distortion}
    }

    \par\medskip

    \subfloat[Dispersion-distortion $\mathbb{V}(d)$ in~\eqref{eq:GAR_VD}.]{%
        \includegraphics[width=0.47\textwidth]{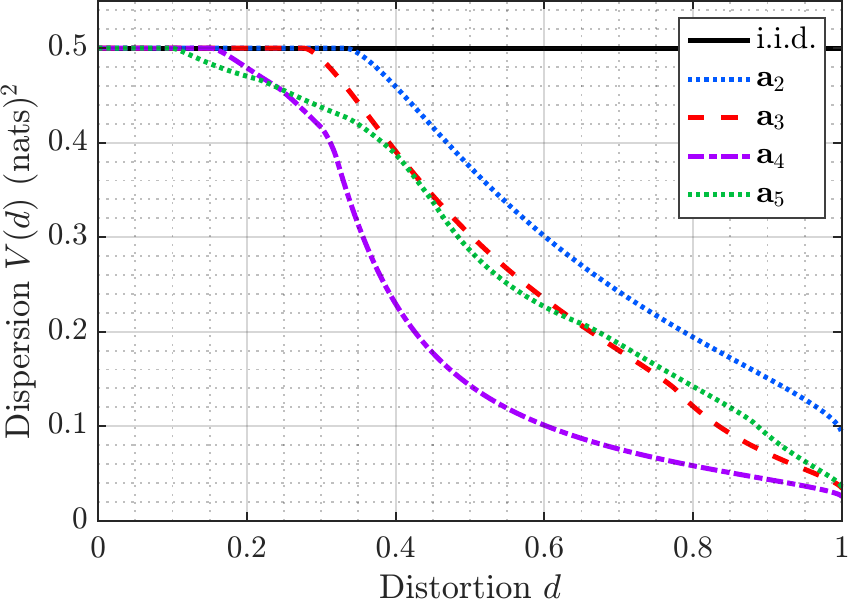}
        \label{fig:dispersion}
    }
    \hfill
    \subfloat[Required blocklength in~\eqref{eq:required_blocklength} to sustain $R = 0.75$ (nats), provided that excess-distortion probability is bounded by $\epsilon$. Color identifies the GAR source, while line style identifies $\epsilon$.]{%
        \includegraphics[width=0.47\textwidth]{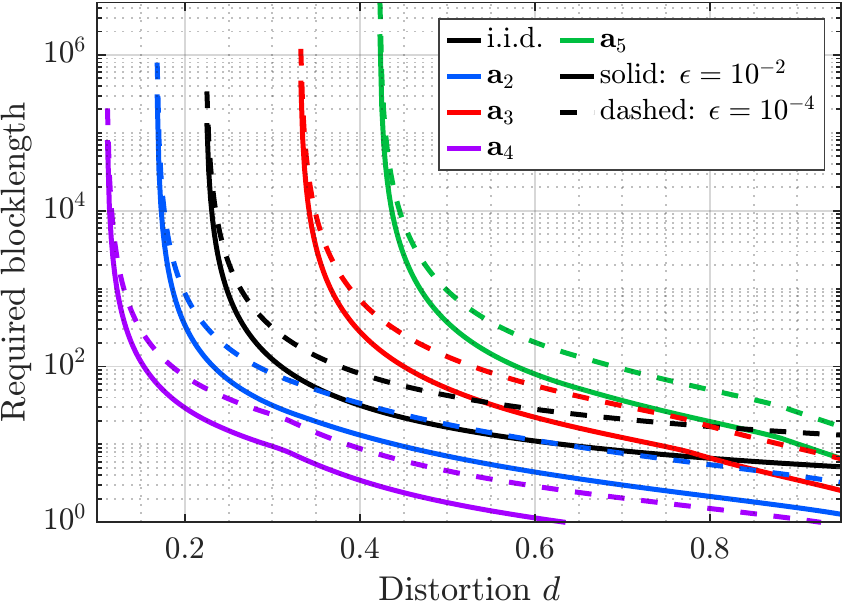}
        \label{fig:gar_required_blocklength}
    }

    \caption{ 
    Comparison of spectrum $S(\omega)$, the corresponding rate-distortion and dispersion functions, and required blocklengths for five Gaussian autoregressive sources with a common maximum distortion $d_{\text{flat}} = 1$: \ac{iid} Gaussian source, AR(1) with $\myVec{a}_2 = 0.5$, AR(2) with $\myVec{a}_3 = [-0.7, -1.2]$, AR(4) with $\myVec{a}_4 = [-0.4, 0.3, 0.2, 0.3]$, and AR(8) with $\myVec{a}_5 = [1, 0.4, 0.5, 0.6, -0.4, 0.8, -0.4, 0.7]$. The corresponding innovation variances are $\sigma_1^2 = 1$, $\sigma_2^2 = 0.75$, $\sigma_3^2 = 1.48$, $\sigma_4^2 = 0.33$, and $\sigma_5^2 = 1.75$. The sources with coefficients $\boldsymbol{a}_3$ and $\boldsymbol{a}_5$ are nonstationary, whereas the remaining sources are stationary.
    }
    \label{fig:four_plots}
\end{figure*}
\fi

Next, we discuss the differences between stationary and nonstationary Gaussian autoregressive sources in terms of their rate-distortion and dispersion functions. To this end, define the function
\begin{equation}
\label{eq:spectrum}
    S(\omega) = \frac{\sigma^2}{g(\omega)}.
\end{equation}
The function $S(\omega)$ is considered to be the spectrum of the process, even when the process is nonstationary and the spectrum does not truly exist. In Fig.~\ref{fig:four_plots}, we compare $S(\omega)$, \ac{rdf}, dispersion and required blocklength to sustain a fixed rate for Gaussian autoregressive sources of different orders. We tune the innovation variances so that the sources have the same $d_{\text{flat}}$. The sources with coefficients $\boldsymbol{a}_3$ and $\boldsymbol{a}_5$ are nonstationary, whereas the remaining sources are stationary. In Fig.~\ref{fig:rate_distortion}, we see that the minimum rate $R_{\min}$ required to achieve finite distortion is positive for the nonstationary sources and zero for stationary ones. To explain this distinction, recall that Kolmogorov's water-filling solution~\cite{kolmogorov56_RD} given by
\begin{equation}
\label{eq:RK}
    \mathbb{R}_{\text{K}}(d) = \frac{1}{2\pi} \int_{-\pi}^\pi \max\left(0, \frac{1}{2}\log \frac{S(\omega)}{\theta^\ast} \right) d\omega,
\end{equation}
which can be shown to be equal to $\mathbb{R}(d)$ in~\eqref{eq:GAR_RD} for stationary sources. However, for nonstationary sources, the two expressions need not be equal. The difference is due to components that have growing covariance eigenvalues, corresponding to the zeros of the characteristic polynomial
\begin{equation}
\label{eq:polynomial}
    A(z) = \sum_{k=0}^{\infty} a_k z^{-k}.
\end{equation}
that lie outside the unit circle. Denote the zeros of~\eqref{eq:polynomial} by $\{\rho_i\}_{i = 1}^m$, ordered such that $|\rho_1|\geq \cdots \geq |\rho_r|>1>|\rho_{r+1}|\geq \cdots \geq |\rho_m|$.  Gray and Hashimoto~\cite{Gray08_NoteOnAutoregressive} showed that the minimum rate $R_{\min}$ is given by
\begin{equation}
    R_{\min} = \sum_{i = 1}^r\log|\rho_i|
\end{equation}
and
\begin{equation}
    \mathbb{R}(d) = \mathbb{R}_{\text{K}}(d) + R_{\min}.
\end{equation}
For stationary sources, all zeros lie inside the unit circle, yielding $R_{\min}=0$. Thus, the positive minimum rate observed for the nonstationary sources is precisely the additional contribution $R_{\min}$. Moreover, the threshold $d_{\mathrm{flat}}$ in~\eqref{eq:dmax} marks the distortion level at which $\mathbb{R}_{\mathrm{K}}(d)$ becomes zero. Consequently, for $d\geq d_{\mathrm{flat}}$, the limiting \ac{rdf} remains constant at $\mathbb{R}(d)=R_{\min}$.

As shown in Fig.~\ref{fig:spectrum}, $S(\omega)$ distributes energy differently across frequencies. We refer to the frequency components satisfying $S(\omega) \leq \theta^\ast$ as \textit{inactive}, since these components do not contribute to the water-filling term in~\eqref{eq:RK}. However, they still contribute to the dispersion in~\eqref{eq:GAR_VD}, although less than the active components. Thus, while the rate is determined only by the active components, the dispersion is affected by the entire spectrum. This difference between the rate-distortion and dispersion functions is illustrated by comparing Fig.~\ref{fig:rate_distortion} and Fig.~\ref{fig:dispersion}. Sources that have similar rate-distortion behavior in Fig.~\ref{fig:rate_distortion} have noticeably different dispersion behavior in Fig.~\ref{fig:dispersion}. This illustrates that the dispersion captures spectral features that are not fully reflected by the \ac{rdf} alone. 

Fig.~\ref{fig:rate_distortion} also demonstrates that, under the same $d_{\text{flat}}$, the \ac{iid} Gaussian source is the hardest to compress among stationary sources. As $S(\omega)$ becomes more concentrated, as seen in Fig.~\ref{fig:spectrum}, the rate decreases because the energy is concentrated over fewer active components. For nonstationary sources, after subtracting the additional term $R_{\min}$, the resulting curves also lie below that of the \ac{iid} Gaussian source.
In contrast, no analogous additive penalty appears in the dispersion. Moreover, in Fig.~\ref{fig:dispersion} a stationary source has a larger dispersion than a nonstationary source. This illustrates that nonstationarity affects the first and second-order asymptotics differently.

We further observe in Fig.~\ref{fig:dispersion} a critical distortion level below which the sources with memory have the same dispersion as the \ac{iid} Gaussian source. This threshold is analogous to the critical threshold observed by Gray~\cite{gray70_Autoregressive} for the \ac{rdf} and is consistent with the dispersion threshold identified by Tian and Kostina~\cite{tian19_Stationary, tian21_NonStationary} for the Gauss-Markov source, which corresponds to the source defined in~\eqref{eq:GaussARSource} with $m=1$. This phenomenon occurs because in the regime of low distortion where the reverse water-filling level lies below the minimum of $\sigma^2/g(\omega)$, all components remain active and therefore contribute equally to the dispersion.

In Fig.~\ref{fig:gar_required_blocklength}, we show the blocklength required to operate at $R$, approximated using~\eqref{eq:GAR_rnd} as
\begin{equation}
\label{eq:required_blocklength}
     n(R, d, \epsilon) \approx \mathbb{V}(d)\left(\frac{Q^{-1}(\epsilon)}{R - \mathbb{R}(d)}\right)^2.
\end{equation}
The required blocklength decreases as the distortion level increases, whereas it increases as the target excess-distortion probability $\epsilon$ becomes smaller. For nearly all distortion levels, nonstationary sources require substantially larger blocklengths to achieve the same distortion under the same rate budget. However, at larger distortion levels, the nonstationary source with coefficients $\boldsymbol{a}_3$ requires smaller blocklengths than the \ac{iid} source, as its substantially lower dispersion outweighs the effect of its higher \ac{rdf}. Thus, the required blocklength depends on both \ac{rdf} and dispersion, and nonstationarity does not necessarily imply a larger required blocklength.

Corollary~\ref{corr:gaussAR} recovers the same second-order term for Gauss-Markov sources as~\cite{tian19_Stationary, tian21_NonStationary}, improving the remainder from $o\left(\frac{1}{\sqrt{n}}\right)$ to $O\left(\frac{\log n}{n}\right)$. 

\subsection{Discussion}
We begin by discussing the scope of our results and some of the technical challenges. First, assumption~\ref{ass:bdd_disp} is not satisfied by all Gaussian sources. To see this, consider a Gaussian source with a single nonzero, possibly growing, covariance eigenvalue, while all other eigenvalues are zero. In this case, the dispersion is $\mathbb{V}_n(d) = \frac{1}{2n}$, violating assumption~\ref{ass:bdd_disp}.

We next turn to Gaussian autoregressive sources, for which nonstationarity poses a significant challenge. In fact, even the characterization of the \ac{rdf} of nonstationary autoregressive sources required several works to be fully understood~\cite{berger70_InformationRatesWiener,gray70_Autoregressive,Hashimoto80_RDAutoregressive,Gray08_NoteOnAutoregressive}. The main difficulty introduced by nonstationarity lies in the behavior of the smallest, near zero eigenvalues $\{\mu_{n,i}\}_{i=1}^n$ of $\mathsf{A}_n^T\mathsf{A}_n$. More precisely, let
\begin{equation}
    \alpha' = \min_{\omega \in [-\pi, \pi]} g(\omega)
    \qquad
    \beta' = \max_{\omega \in [-\pi, \pi]} g(\omega).
\end{equation}
For stationary sources, the constants $\alpha$ and $\beta$ defined in~\eqref{eq:alpha_beta} satisfy
\begin{equation}
\label{eq:mu_i_bounds_stationary}
    0<\alpha' = \alpha \leq \mu_{n,i} \leq \beta = \beta'<\infty,
    \qquad i\in[n],\ n\in\mathbb{N}.
\end{equation}
Thus, all eigenvalues are uniformly bounded away from zero. In contrast, for nonstationary sources, $\alpha=0$, meaning that some eigenvalues become arbitrarily small as blocklength increases. Nevertheless, all eigenvalues remain strictly positive for every finite blocklength, and the bounds take the form
\begin{equation}
\label{eq:mu_i_bounds_non_stationary}
    0 = \alpha < \mu_{n,i} \leq \beta = \beta'<\infty,
    \qquad i\in[n],\ n\in\mathbb{N}.
\end{equation}

The small eigenvalues near zero create a difficulty in passing to the limit in~\eqref{eq:rd_info}, due to the discontinuity of the logarithm at zero. Gray~\cite[Eq.~19]{gray70_Autoregressive} establishes the limit by reformulating the $n$th-letter \ac{rdf} in terms of a function that is continuous at zero, and using the equivalence between the limiting eigenvalues distribution of the inverse covariance matrix and its Toeplitz approximation, thereby deriving the \ac{rdf} of possibly nonstationary Gaussian autoregressive processes. Hashimoto and Arimoto~\cite{Hashimoto80_RDAutoregressive} circumvented the same difficulty in a different way. For sources with finite memory, they observed that only finitely many eigenvalues decay exponentially. They treated these eigenvalues separately and derived an \ac{rdf} expression~\cite[Thm.~2]{Hashimoto80_RDAutoregressive} that is distinct in form but equivalent to Gray's. The equivalence between the two \ac{rdf} expressions was later established by Gray and Hashimoto~\cite{Gray08_NoteOnAutoregressive}.

Although these works characterize the limiting \ac{rdf}, they do not provide the convergence rate required for our second-order analysis. Lemma~\ref{lemm:GARnon_asymptotic_eigenvalues} addresses this gap by establishing an explicit $O(1/n)$ rate, generalizing the results of Tian and Kostina~\cite{tian19_Stationary,tian21_NonStationary} for stationary and nonstationary Gauss-Markov sources, as well as our preliminary results~\cite{tasci2026_ISIT} for a stationary Gaussian autoregressive sources with finite memory. The lemma thus provides a quantitative refinement of Gray's asymptotic result. 

More broadly, the present work extends our preliminary second-order analysis in two complementary directions:
\begin{enumerate}
    \item \textit{From stationary to nonstationary:} The preliminary version of Corollary~\ref{corr:gaussAR} was restricted to stationary sources because it relied on the stronger uncoupled distortion moment condition in~\eqref{eq:uncoupled_moment}. For nonstationary sources, the components with growing eigenvalues cause the moment bound in~\eqref{eq:uncoupled_moment} to fail. In this work, we relax the condition in~\eqref{eq:uncoupled_moment} to the assumption~\ref{ass:uni_moments}, which holds even when eigenvalues grow without bound. This relaxation allows us to apply Theorem~\ref{thm:gaussian_approximation} to nonstationary Gaussian autoregressive sources.

     \item \textit{From finite to infinite memory:}
    Although the first version of Lemma~\ref{lemm:GARnon_asymptotic_eigenvalues} already covered possibly nonstationary Gaussian autoregressive sources, it was restricted to sources with finite memory. The proof relied primarily on the memory order $m$   being fixed and finite. Specifically, we approximated the eigenvalues $\{\mu_{n,i}\}_{i=m}^{n-m}$ by those of a Toeplitz matrix and bounded the approximation error, while showing that the remaining $2m$ possibly unstable eigenvalues are negligible for the purposes of the analysis. In this work, we extend this result to include sources with infinite memory in Lemma~\ref{lemm:GARnon_asymptotic_eigenvalues}, by bounding average perturbation of the eigenvalues using Lidskii's theorem in Appendix~\ref{sec:proof_GAR_eig}, rather than approximating each eigenvalue individually.
\end{enumerate}
Together, these two extensions allow us to extend Corollary~\ref{corr:gaussAR} to possibly nonstationary Gaussian autoregressive sources with infinite memory, providing a unified treatment of stationary and nonstationary sources.
\color{black}

\section{Converse}
\label{sec:converse}
The proof is based on a general converse by Kostina and Verdú~\cite{kostina12_FixedLength}.
\begin{lemma}[Kostina and Verdú{~\cite[Thm.~7]{kostina12_FixedLength}}]
\label{lem:converse}
Fix $d > d_{\min}$. Any $(n, M, d, \epsilon)$ code must satisfy
\begin{equation}
\label{eq:converse}
\epsilon \ge \sup_{\gamma \ge 0} 
\mathbb{P} \left[ \jmath(\myVec{X}, d) \ge \log M + \gamma \right] - \exp(-\gamma).
\end{equation}
\end{lemma}
We also need the following regularity results.
\begin{lemma}
\label{lem:bdd_curv}
Fix $d > d_{\min}$. Then there exists a finite constant $C_{1}$ such that
\begin{equation}
\label{eq:C1}
    \limsup_{n \to \infty} \lambda_n^\ast \le C_{1} .
\end{equation}
\end{lemma}
\begin{proof}
For each $n$, the $n$th-order \ac{rdf} $\mathbb{R}_n(d)$ is convex. Then for all sufficiently large $n$, and for all $d_0 > d_{\min}$ we have 
\begin{equation}
     \mathbb{R}_n(d) - \mathbb{R}_n'(d) (d - d_0) \leq \mathbb{R}_n(d_0).
\end{equation}
Let $d_0 = \frac{d + d_{\min}}{2}$. Then by definition of $\lambda_n^\ast$ in~\eqref{eq:lambda}, we have
\begin{equation}
\lambda_n^\ast \le \frac{\mathbb{R}_n(d_0) - \mathbb{R}_n(d)}{d - d_0}
\le \frac{\mathbb{R}_n(d_0)}{d - d_0}.
\end{equation}
Since $d_0 > d_{\min}$, the sequence $\{\mathbb{R}_n(d_0)\}_{n\ge 1}$ is bounded for all large $n$. Taking the limit superior on both sides proves the claim.
\end{proof}

\begin{proof}[Proof of the converse part of Theorem~\ref{thm:gaussian_approximation}]
Let $W_i = \jmath\left(X_i, d_i\right)$ in equation~\eqref{eq:tilted_single_letter}.
By assumption~\ref{ass:bdd_disp}, the average variance of $W_i$ is bounded away from zero, while assumption~\ref{ass:up_bdd_disp} ensures  that their average third absolute central moment is bounded. By Hölder's inequality, the latter also implies that dispersion is bounded above. Therefore, the Berry--Esseen constant $B_n$ in Theorem~\ref{thm:berry_Esseen}, stated in Appendix~\ref{app:berry_esseen} below, is uniformly bounded for sufficiently large $n$.

Set $\gamma=\frac{1}{2}\log n$ in Lemma~\ref{lem:converse}. For fixed $\epsilon\in(0,1)$, define
\begin{equation}
\epsilon_n \triangleq \epsilon+\frac{B_n+1}{\sqrt{n}},
\end{equation}
so that $\epsilon_n\in(0,1)$ for all sufficiently large  $n$. Choose
\begin{equation}
\log M \triangleq n \mathbb{R}_n(d)
+ \sqrt{n \mathbb{V}_n(d)}\, Q^{-1}(\epsilon_n)
-\frac{\log n}{2}.
\end{equation}
Then, for any $(n,M,d,\epsilon')$ code, Lemma~\ref{lem:converse} yields
\begin{align}
\epsilon'
&\ge \mathbb{P} \left[\jmath(\myVec{X},d)\ge \log M + \frac{\log n}{2}\right]-\frac{1}{\sqrt{n}} \\
&= \mathbb{P} \left[\jmath(\myVec{X},d)\ge n\mathbb{R}_n(d)
+ \sqrt{n \mathbb{V}_n(d)}\, Q^{-1}(\epsilon_n)\right]-\frac{1}{\sqrt{n}} \\
&\ge \epsilon_n-\frac{B_n}{\sqrt{n}}-\frac{1}{\sqrt{n}} \label{prf:conv_eq:apply_BE} \\
&= \epsilon,
\end{align}
where~\eqref{prf:conv_eq:apply_BE} follows from Berry--Esseen Theorem~\ref{thm:berry_Esseen}. Therefore, any admissible code must satisfy
\begin{equation}
R(n,d,\epsilon)\ge \frac{\log M}{n}.
\end{equation}
Finally, the Taylor expansion of $Q^{-1}(\epsilon_n)$ around $\epsilon$, together with assumption~\ref{ass:up_bdd_disp} and Hölder's inequality, yields
\begin{align}
    R(n, d, \epsilon) \geq \mathbb{R}_n(d) + \sqrt{\frac{\mathbb{V}_n(d)}{n}}Q^{-1}(\epsilon) - \frac{\log n}{2n} + O\left (\frac{1}{n}\right)
\end{align}
\end{proof}

\textit{Remark:} The proof of the converse part does not rely on assumptions~\ref{ass:uni_moments}--
\ref{ass:uni_curv}. 

\section{Achievability}
\label{sec:achievability}
Our achievability proof builds on the nonasymptotic bounds of Kostina and Verdú~\cite{kostina12_FixedLength}. To extend their second-order analysis to our setting with non-identical source components, we introduce a point-mass product proxy measure, which is the key new device enabling our argument. 

For any $\myVec{x} \in \boldsymbol{\mySet{X}}$, we define the proxy random vector $\myVec{\hat X}(\myVec{x})$ to be distributed according to the point-mass product measure
\begin{equation}
\label{eq:proxy}
P_{\myVec{\hat X}(\myVec{x})}(\myVec{a})
\triangleq \prod_{i=1}^n \delta_{x_i}(a_i).
\end{equation}
This proxy measure ensures that the expected value of the distortion and information densities under the proxy distribution matches the empirical average induced by the sequence $\myVec{x}$, allowing us to invoke the Berry--Esseen Theorem and to construct a typical set.

We define the distortion $d$-ball $\mathcal{B}(\myVec{x},d)$ centered at $\myVec{x}\in\boldsymbol{\mySet{X}}$ as
\begin{equation}
\mathcal{B}(\myVec{x},d) \triangleq \left\{ \myVec{y}' \in \boldsymbol{\mySet{Y}} \colon \mathsf{d}(\myVec{x}, \myVec{y}') \le d \right\}.
\end{equation}
The random-coding bound, stated next, relates the excess-distortion probability to the probability that reproduction $\myVec{Y}$ falls in the distortion-$d$ ball for $\myVec{X}$.
\begin{lemma}[Kostina and Verdú{~\cite[Cor. 11]{kostina12_FixedLength}}]
\label{lem:rcb}
There exists an $(n,M,d,\epsilon)$ code with
\begin{equation}
\label{eq:achievability}
\epsilon \le \inf_{P_{\myVec{Y}}} 
\mathbb{E}_{\myVec{X}}\left[
e^{-M P_{\myVec{Y}}(\mathcal{B}(\myVec{X},d))}
\right],
\end{equation}
where the infimum is over all probability distributions $P_{\myVec{Y}}$ on $\boldsymbol{\mySet{Y}}$ with $\myVec{Y}$ independent of $\myVec{X}$.
\end{lemma}

The next lemma is one of the main contributions of this paper. It generalizes the nonasymptotic refinement of the lossy \ac{aep}~\cite{kontoyiannis00_PointwiseRedundancy, kostina12_FixedLength, yang99_RedundancyRD}. The lossy \ac{aep} relates the probability of distortion-$d$ balls to the $d$-tilted information and plays a central role in 
finite blocklength achievability bounds in rate-distortion theory~\cite{kostina12_FixedLength}. 
\begin{lemma}[Lossy AEP]
\label{lem:lossyAEP}
Suppose that assumptions~\ref{ass:bdd_disp}--
\ref{ass:uni_curv} hold. For any $\epsilon \in (0,1)$, there exist constants $ C_{2}, c, K_{3}>0$ such that for all sufficiently large $n$,
\begin{equation}
\Prob{\log\frac{1}{P_{\myVec{Y^\ast}}(\mathcal{B}(\myVec{X},d))}
\le
\jmath(\myVec{X},d) + C_{2} \log n + c}
\ge 1 - \frac{K_{3}}{\sqrt{n}},
\end{equation}
where $\myVec{Y^\ast}$ is the \ac{rdf}-achieving reproduction vector for $\mathbb{R}_n(d)$.
\end{lemma}
For Gauss-Markov sources, a weaker version of Lemma~\ref{lem:lossyAEP} was shown in~\cite{tian19_Stationary, tian21_NonStationary}, while for \ac{iid} sources it is equivalent to~\cite[Lem.~2]{kostina12_FixedLength}. To prove Lemma~\ref{lem:lossyAEP}, we need Lemmas~\ref{lem:lower_bound_distortion_balls}--\ref{lem:tilted_concentration}, stated next. 
\begin{lemma}[Kostina and Verdú{~\cite[Lem.~1]{kostina12_FixedLength}}]
\label{lem:lower_bound_distortion_balls}
Fix a distribution $P_{\myVec{Y}}$. Then for any $\myVec{x} \in \boldsymbol{\mySet{X}}$, it holds that
\ifdoublecolumn
\begin{align}
\label{eq:d_ball_prob}
P_{\myVec{Y}}(\mathcal{B}(\myVec{x}, d))
\geq &\sup_{P_{\myVec{\hat X}},  \gamma > 0} \bigg(
\exp \left(- \hat{\lambda}_n^\ast \gamma
 - J_\myVec{Y}(\myVec{x}, \hat{\lambda}_n^\ast) + \hat{\lambda}_n^\ast nd
\right) \notag \\
& \mathbb{P}\left[
d - \frac{\gamma}{n}
< \mathsf{d}(\myVec{x}, \myVec{\hat Z^\ast})
\leq d
\middle|
\myVec{\hat X} = \myVec{x}
\right]\bigg),
\end{align}
\else
\begin{align}
\label{eq:d_ball_prob}
P_{\myVec{Y}}(\mathcal{B}(\myVec{x}, d))
\geq &\sup_{P_{\myVec{\hat X}},  \gamma > 0} \bigg(
\exp \left(- \hat{\lambda}_n^\ast \gamma
 - J_\myVec{Y}(\myVec{x}, \hat{\lambda}_n^\ast) + \hat{\lambda}_n^\ast nd
\right) \mathbb{P}\left[
d - \frac{\gamma}{n}
< \mathsf{d}(\myVec{x}, \myVec{\hat Z^\ast})
\leq d
\middle|
\myVec{\hat X} = \myVec{x}
\right]\bigg),
\end{align}
\fi
where the supremum is over all probability distributions $P_{\myVec{\hat X}}$ on $\boldsymbol{\mySet{X}}$ such that $d_{\min|\myVec{\hat X}, \myVec{Y}} < d$, $\hat{\lambda}_n^\ast \triangleq - \mathbb{R}_n'(\myVec{\hat X}, \myVec{Y}, d)$, and the random vector $\myVec{\hat Z^\ast}$ is distributed according to the minimizer $P_{\myVec{\hat Z^\ast}|\myVec{\hat X}}$ of~\eqref{eq:crem} achieving $\mathbb{R}_n(\myVec{\hat X}, \myVec{Y}, d)$.
\end{lemma}
\begin{lemma}[Shell-probability lower bound]
\label{lem:shell_probability_lower_bound}
Suppose that assumptions~\ref{ass:bdd_disp}--
\ref{ass:uni_curv} hold. Fix $\epsilon \in (0,1)$. Then, there exists a constant $\delta_0> 0$ such that for all $0 <\delta \le \delta_0$, and sufficiently large $n$, there exists a typical set $\mathcal{T}_n$ and constants $\gamma, C_3, K_{4} > 0$ such that 
\begin{equation}
\label{eq:typical}
    \mathbb{P}\left[\myVec{X} \not \in \mathcal{T}_n
\right]
\leq \frac{K_{4}}{\sqrt{n}},
\end{equation}
and for all $\myVec{x} \in \mathcal{T}_n$,
\begin{equation}
\label{eq:shell}
\mathbb{P}\left[
d - \frac{\gamma}{n}
< \mathsf{d}(\myVec{x}, \myVec{\hat Z^\ast})\leq d
\middle|
\myVec{\hat X} = \myVec{x}
\right]
\geq \frac{C_3}{\sqrt{n}},
\end{equation}
\begin{equation}
\label{eq:lambda_concentration}
|\hat{\lambda}_n^\ast(\myVec{x})- \lambda_n^\ast | < \delta
\end{equation}
where $\hat{\lambda}_n^\ast(\myVec{x}) = - \mathbb{R}_n'\bigl(\myVec{\hat X}(\myVec{x}), \myVec{Y^\ast}, d\bigr)$ and the random vector $\myVec{\hat Z^\ast}$ is distributed according to the minimizer $P_{\myVec{\hat Z^\ast}|\myVec{\hat X(x)}}$ of~\eqref{eq:crem} achieving $\mathbb{R}_n\bigl(\myVec{\hat X}(\myVec{x}), \myVec{Y^\ast}, d\bigr)$.
\end{lemma}

\begin{proof}[Proof]
The proof follows steps similar to those in~\cite[Lem.~4]{kostina12_FixedLength}. The main difference is the use of the proxy measure in~\eqref{eq:proxy} instead of an empirical measure. Under the proxy measure, the expected distortion equals the empirical distortion of the sequence $\myVec{x}$. This decomposition allows us to apply the Berry--Esseen Theorem even though the components are independent but not identically distributed. See Appendix~\ref{app:proof_shell_probability} for full details. 
\end{proof}
\begin{lemma}\label{lem:tilted_concentration}
Suppose that assumptions~\ref{ass:bdd_disp}--
\ref{ass:uni_curv} hold. Fix $\epsilon \in (0,1)$. Then there exist constants $C_4, K_{5} > 0$ such that for all sufficiently large $n$,
\ifdoublecolumn
\begin{align}
\mathbb{P}\bigg[
J_\myVec{Y^\ast}(\myVec{X}, \hat{\lambda}_n^\ast(\myVec{X})) &- \hat{\lambda}_n^\ast(\myVec{X}) nd \leq  J_\myVec{Y^\ast}(\myVec{X}, \lambda_n^\ast)  \notag \\
&- \lambda_n^\ast nd + C_4\log n
\bigg]
> 1 - \frac{K_{5}}{\sqrt{n}}.
\end{align}
\else
\begin{align}
\mathbb{P}\bigg[
J_\myVec{Y^\ast}(\myVec{X}, \hat{\lambda}_n^\ast(\myVec{X})) &- \hat{\lambda}_n^\ast(\myVec{X}) nd \leq  J_\myVec{Y^\ast}(\myVec{X}, \lambda_n^\ast) - \lambda_n^\ast nd + C_4\log n
\bigg]
> 1 - \frac{K_{5}}{\sqrt{n}}.
\end{align}
\fi
\end{lemma}
\begin{proof}
Since the expectation of $J_{\myVec{Y^\ast}}(\myVec{\hat X},\hat{\lambda}_n^\ast(\myVec{X}))$ under the proxy distribution becomes additive across coordinates with independent terms, we follow steps similar to those in the proof of~\cite[Lem.~5]{kostina12_FixedLength}. See Appendix~\ref{app:proof_tilted_concentration} for full details.
\end{proof}

\begin{proof}[Proof of Lemma~\ref{lem:lossyAEP}]
We follow steps similar to those in the proof of Lemma~2 in~\cite{kostina12_FixedLength}. 
Consider $n$ large enough such that Lemmas~\ref{lem:shell_probability_lower_bound}-\ref{lem:tilted_concentration} hold. Let $C_{1}$ be the constant in Lemma~\ref{lem:bdd_curv}, and $\delta, \gamma, C_3, C_4$ be the constants in Lemmas~\ref{lem:shell_probability_lower_bound}-\ref{lem:tilted_concentration}. Let $c = (C_{1} + \delta)\gamma - \log C_3$, and $C_{2} = \frac{1}{2} + C_4$. We set $P_{\myVec{Y}} = P_{\myVec{Y^\ast}}$ in Lemma~\ref{lem:lower_bound_distortion_balls} and evaluate the supremum in~\eqref{eq:d_ball_prob} at the point-mass product measure defined in~\eqref{eq:proxy}, and at the $\gamma$ value specified in Lemma~\ref{lem:shell_probability_lower_bound}. The shell probability term in~\eqref{eq:d_ball_prob} is then lower bounded by Lemma~\ref{lem:shell_probability_lower_bound}, while the remaining exponential term is controlled by Lemma~\ref{lem:tilted_concentration}.

Using elementary probability, we have
\ifdoublecolumn
\begin{align}
\mathbb{P} \bigg[
    &\log\frac{1}{P_{\myVec{Y^\ast}}(\mySet{B}(\myVec{X},d))} >
    \jmath(\myVec{X}, d)  + C_{2} \log n + c
\bigg] \\
\label{eq:lem_lossyAEP_pf_rcb} &\le
\mathbb{P} \bigg[
    J_\myVec{Y^\ast}(\myVec{X}, \hat{\lambda}_n^\ast(\myVec{X})) + \hat{\lambda}_n^\ast(\myVec{X}) \gamma - \hat{\lambda}_n^\ast(\myVec{X}) n d \notag \\
    & \qquad -\log \mathbb{P}\left[nd - \gamma < n \mathsf{d}(\myVec{x},\myVec{\hat Z}^\ast) \leq nd \, \middle| \, \myVec{\hat X} = \myVec{x} \right] \notag \\
    & \quad \qquad > \jmath(\myVec{X}, d)  + C_{2} \log n + c 
\bigg]  \\
\label{eq:lem_lossyAEP_pf_shell} &\le
\mathbb{P} \bigg[
    \myVec{X} \in \mathcal{T}_n,\, 
    J_\myVec{Y^\ast}(\myVec{X}, \hat{\lambda}_n^\ast(\myVec{X})) - \hat{\lambda}_n^\ast(\myVec{X})nd \notag \\
    & \quad \qquad > \jmath(\myVec{X}, d) - \hat{\lambda}_n^\ast(\myVec{X}) \gamma + \log \frac{C_3}{\sqrt{n}} + C_{2} \log n + c \bigg] \notag \\ 
& \quad + \Prob{\myVec{X} \not \in \mathcal{T}_n} \\
\label{eq:lem_lossyAEP_pf_typical} &\le
\mathbb{P} \bigg[
\myVec{X} \in \mathcal{T}_n, \, 
J_\myVec{Y^\ast}(\myVec{X}, \hat{\lambda}_n^\ast(\myVec{X})) - \hat{\lambda}_n^\ast(\myVec{X})nd \notag \\
& \qquad >
\jmath(\myVec{X}, d)  + \left(C_{2} - \frac{1}{2}\right) \log n + c  \notag \\ 
& \quad \qquad + \log C_3 - (\lambda_n^\ast + \delta)\gamma \bigg] + \frac{K_{4}}{\sqrt{n}} \\
&\le \mathbb{P} \bigg[J_\myVec{Y^\ast}(\myVec{X}, \hat{\lambda}_n^\ast(\myVec{X})) - \hat{\lambda}_n^\ast(\myVec{X})nd > \jmath(\myVec{X}, d) \notag \\
& \quad \qquad   + C_4 \log n  \bigg] + \frac{K_{4}}{\sqrt{n}} \\
&\le \frac{K_{4} + K_{5}}{\sqrt{n}} \label{eq:lem_lossyAEP_pf_tilted_concentration},
\end{align}
\else
\begin{align}
\mathbb{P} \bigg[
    &\log\frac{1}{P_{\myVec{Y^\ast}}(\mySet{B}(\myVec{X},d))} >
    \jmath(\myVec{X}, d)  + C_{2} \log n + c
\bigg] \\
\label{eq:lem_lossyAEP_pf_rcb} &\le
\mathbb{P} \bigg[
    J_\myVec{Y^\ast}(\myVec{X}, \hat{\lambda}_n^\ast(\myVec{X})) + \hat{\lambda}_n^\ast(\myVec{X}) \gamma - \hat{\lambda}_n^\ast(\myVec{X}) n d  -\log \mathbb{P}\left[nd - \gamma < n \mathsf{d}(\myVec{x},\myVec{\hat Z}^\ast) \leq nd \, \middle| \, \myVec{\hat X} = \myVec{x} \right] > \jmath(\myVec{X}, d)  + C_{2} \log n + c 
\bigg]  \\
\label{eq:lem_lossyAEP_pf_shell} &\le
\mathbb{P} \bigg[
    \myVec{X} \in \mathcal{T}_n,\, 
    J_\myVec{Y^\ast}(\myVec{X}, \hat{\lambda}_n^\ast(\myVec{X})) - \hat{\lambda}_n^\ast(\myVec{X})nd > \jmath(\myVec{X}, d) - \hat{\lambda}_n^\ast(\myVec{X}) \gamma + \log \frac{C_3}{\sqrt{n}} + C_{2} \log n + c \bigg]  + \Prob{\myVec{X} \not \in \mathcal{T}_n} \\
\label{eq:lem_lossyAEP_pf_typical} &\le
\mathbb{P} \bigg[
\myVec{X} \in \mathcal{T}_n, \, 
J_\myVec{Y^\ast}(\myVec{X}, \hat{\lambda}_n^\ast(\myVec{X})) - \hat{\lambda}_n^\ast(\myVec{X})nd >
\jmath(\myVec{X}, d)  + \left(C_{2} - \frac{1}{2}\right) \log n + c  + \log C_3 - (\lambda_n^\ast + \delta)\gamma \bigg] + \frac{K_{4}}{\sqrt{n}} \\
&\le \mathbb{P} \bigg[J_\myVec{Y^\ast}(\myVec{X}, \hat{\lambda}_n^\ast(\myVec{X})) - \hat{\lambda}_n^\ast(\myVec{X})nd > \jmath(\myVec{X}, d)  + C_4 \log n  \bigg] + \frac{K_{4}}{\sqrt{n}} \\
&\le \frac{K_{4} + K_{5}}{\sqrt{n}} \label{eq:lem_lossyAEP_pf_tilted_concentration},
\end{align}
\fi
where \eqref{eq:lem_lossyAEP_pf_rcb} follows from Lemma~\ref{lem:lower_bound_distortion_balls}, \eqref{eq:lem_lossyAEP_pf_shell} follows from  \eqref{eq:shell} in Lemma~\ref{lem:shell_probability_lower_bound}, \eqref{eq:lem_lossyAEP_pf_typical} follows from \eqref{eq:typical} and \eqref{eq:lambda_concentration} in the same lemma, and \eqref{eq:lem_lossyAEP_pf_tilted_concentration} follows from Lemma~\ref{lem:tilted_concentration}.
\end{proof}

\begin{proof}[Proof of achievability part of Theorem 1]
Once we have Lemma~\ref{lem:lossyAEP}, the achievability proof of Theorem~\ref{thm:gaussian_approximation} is similar to that of~\cite[Thm.~12]{kostina12_FixedLength}, except the random variables we apply the Berry--Esseen Theorem~\ref{thm:berry_Esseen} to are no longer identical. 

Fix $\epsilon \in (0,1)$ and define $\epsilon_n$
\begin{equation}
\label{eq:def_epsilon}
\epsilon_n \triangleq \epsilon-\frac{B_n+1+K_{3}}{\sqrt{n}},
\end{equation}
where $B_n$ is the Berry--Esseen constant that appears in the converse proof, and $K_{3}$ is the constant in Lemma~\ref{lem:lossyAEP}. As discussed in the converse proof, assumptions~\ref{ass:bdd_disp}--
\ref{ass:up_bdd_disp}
imply that $B_n$ is uniformly bounded for sufficiently large $n$. Hence, we have $\epsilon_n \in (0,1)$ for all sufficiently large $n$.

Choose $M$ according to
\ifdoublecolumn
\begin{align}
\label{eq:logM}
\log M \triangleq\;& n \mathbb{R}_n(d)
+ \sqrt{n \mathbb{V}_n(d)} Q^{-1}(\epsilon_n)
+ \log\frac{\log n}{2} \notag\\
&\quad + C_{2} \log n + c,
\end{align}
\else
\begin{align}
\label{eq:logM}
\log M \triangleq\;& n \mathbb{R}_n(d)
+ \sqrt{n \mathbb{V}_n(d)} Q^{-1}(\epsilon_n)
+ \log\frac{\log n}{2} + C_{2} \log n + c,
\end{align}
\fi
and define
\begin{equation}
\label{eq:def_Gn}
G_n \triangleq \log M - \jmath(\myVec{X}, d) - C_{2} \log n - c.
\end{equation}
as well as the event
\begin{equation}
\label{eq:def_G}
\mathcal{G}\triangleq\left\{G_n<\log\frac{\log n}{2}\right\}.
\end{equation}
By the Berry--Esseen theorem applied to $\jmath(\myVec{X},d)$, we obtain
\begin{align}
\label{eq:prop_G}
\mathbb{P}[\mathcal{G}]
&\le \mathbb{P}\left[
\jmath(\myVec{X}, d)
\ge  n \mathbb{R}_n(d)
+ \sqrt{n \mathbb{V}_n(d)} Q^{-1}(\epsilon_n)
\right] \\
&\le \epsilon_n+\frac{B_n}{\sqrt{n}}.
\end{align}

Next, define
\begin{align}
\label{eq:def_L}
\mathcal{L}
&= \left\{
\log\frac{1}{P_{\myVec{Y^\ast}}(\mySet{B}(\myVec{X},d))}
\le \jmath(\myVec{X}, d) + C_{2} \log n + c
\right\}.
\end{align}
By Lemma~\ref{lem:lossyAEP},
\begin{equation}
\label{eq:prob_L}
\mathbb{P}[\mathcal{L}] \ge 1-\frac{K_{3}}{\sqrt{n}}.
\end{equation}
Applying Lemma~\ref{lem:rcb} with the optimal reproduction distribution $P_{\myVec{Y^\ast}}$, there exists an $(n,M,d,\epsilon')$ code such that
\ifdoublecolumn
\begin{align}
\epsilon'
& \le \mathbb{E}_{\myVec{X}}\left[e^{-M P_{\myVec{Y^\ast}}(\mathcal{B}(\myVec{X},d))}\right] \\
&= \mathbb{E}_{\myVec{X}}\left[e^{-M P_{\myVec{Y^\ast}}(\mathcal{B}({\myVec{X}},d))}\mathbbm{1}\{\mathcal{L}\}\right] \notag \\
& \quad \quad + \mathbb{E}_{\myVec{X}}\left[e^{-M P_{\myVec{Y^\ast}}(\mathcal{B}({\myVec{X}},d))}\mathbbm{1}\{\mathcal{L}^{\mathrm{c}}\}\right] \\
& \le \mathbb{E}_{\myVec{X}}\left[e^{-e^{G_n}}\right]
+ \frac{K_{3}}{\sqrt{n}} \label{eq:proof_G_n} \\
&= \mathbb{E}_{\myVec{X}}\left[e^{-e^{G_n}}\mathbbm{1}\{\mathcal{G}\}\right]
+ \mathbb{E}_{\myVec{X}}\left[e^{-e^{G_n}}\mathbbm{1}\{\mathcal{G}^{\mathrm{c}}\}\right]
+ \frac{K_{3}}{\sqrt{n}} \\
&\le \mathbb{P}(\mathcal{G})+\frac{1}{\sqrt{n}}\mathbb{P}(\mathcal{G}^{\mathrm{c}})
+ \frac{K_{3}}{\sqrt{n}} \label{eq:proof_GandGc} \\
&\le \epsilon_n+\frac{B_n+1}{\sqrt{n}}
+ \frac{K_{3}}{\sqrt{n}} \label{eq:proof_epsilon}\\
&=\epsilon.
\end{align}
\else
\begin{align}
\epsilon'
& \le \mathbb{E}_{\myVec{X}}\left[e^{-M P_{\myVec{Y^\ast}}(\mathcal{B}(\myVec{X},d))}\right] \\
&= \mathbb{E}_{\myVec{X}}\left[e^{-M P_{\myVec{Y^\ast}}(\mathcal{B}({\myVec{X}},d))}\mathbbm{1}\{\mathcal{L}\}\right]  + \mathbb{E}_{\myVec{X}}\left[e^{-M P_{\myVec{Y^\ast}}(\mathcal{B}({\myVec{X}},d))}\mathbbm{1}\{\mathcal{L}^{\mathrm{c}}\}\right] \\
& \le \mathbb{E}_{\myVec{X}}\left[e^{-e^{G_n}}\right]
+ \frac{K_{3}}{\sqrt{n}} \label{eq:proof_G_n} \\
&= \mathbb{E}_{\myVec{X}}\left[e^{-e^{G_n}}\mathbbm{1}\{\mathcal{G}\}\right]
+ \mathbb{E}_{\myVec{X}}\left[e^{-e^{G_n}}\mathbbm{1}\{\mathcal{G}^{\mathrm{c}}\}\right]
+ \frac{K_{3}}{\sqrt{n}} \\
&\le \mathbb{P}(\mathcal{G})+\frac{1}{\sqrt{n}}\mathbb{P}(\mathcal{G}^{\mathrm{c}})
+ \frac{K_{3}}{\sqrt{n}} \label{eq:proof_GandGc} \\
&\le \epsilon_n+\frac{B_n+1}{\sqrt{n}}
+ \frac{K_{3}}{\sqrt{n}} \label{eq:proof_epsilon}\\
&=\epsilon.
\end{align}
\fi
Here, \eqref{eq:proof_G_n} follows from \eqref{eq:prob_L}, together with the definitions in \eqref{eq:def_L} and \eqref{eq:def_Gn}; \eqref{eq:proof_GandGc} follows by using the trivial bound $e^{-e^{G_n}}\le 1$ on $\mathcal{G}$ and the bound $e^{-e^{G_n}}\le 1/\sqrt{n}$ on $\mathcal{G}^{\mathrm{c}}$ via \eqref{eq:def_G}; and \eqref{eq:proof_epsilon} follows from the definition in \eqref{eq:def_epsilon}. Therefore, there exists an $(n,M,d,\epsilon')$ code with $\epsilon'\le\epsilon$, which implies
\begin{equation}
R(n,d,\epsilon)\le \frac{\log M}{n}.
\end{equation}
By Taylor expansion, assumption~\ref{ass:up_bdd_disp} and Hölder's inequality, we obtain
\ifdoublecolumn
\begin{align}
    R(n, d, \epsilon) \leq \mathbb{R}_n(d) + \sqrt{\frac{\mathbb{V}_n(d)}{n}}&Q^{-1}(\epsilon) + C_{2}\frac{\log n}{n} \notag \\
     &+ \frac{\log\log n}{n} + O\left(\frac{1}{n}\right)
\end{align}
\else
\begin{align}
    R(n, d, \epsilon) \leq \mathbb{R}_n(d) + \sqrt{\frac{\mathbb{V}_n(d)}{n}}&Q^{-1}(\epsilon) + C_{2}\frac{\log n}{n} + \frac{\log\log n}{n} + O\left(\frac{1}{n}\right)
\end{align}
\fi
\end{proof}

\ifdoublecolumn
\begin{figure}[!t]
    \centering

    \subfloat{%
        \includegraphics[width=0.95\columnwidth]{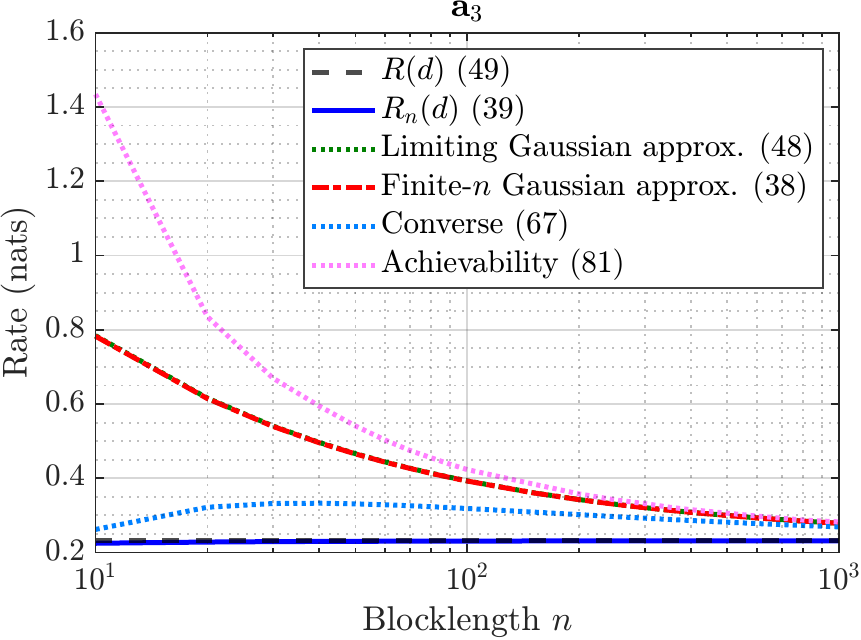}
    }

    \caption{Rate versus blocklength with parameters:  $\myVec{a}_3 = \lbrack -0.5 ,0.4\rbrack$, $\sigma^2 = 0.73$, $d = 0.5$, $\epsilon = 10^{-2}.$ The source is stationary.}
    \label{fig:gar_rate_blocklength}
\end{figure}
\else 
\begin{figure}[!htbp]
    \centering

    \subfloat{%
        \includegraphics[width=0.67\textwidth]{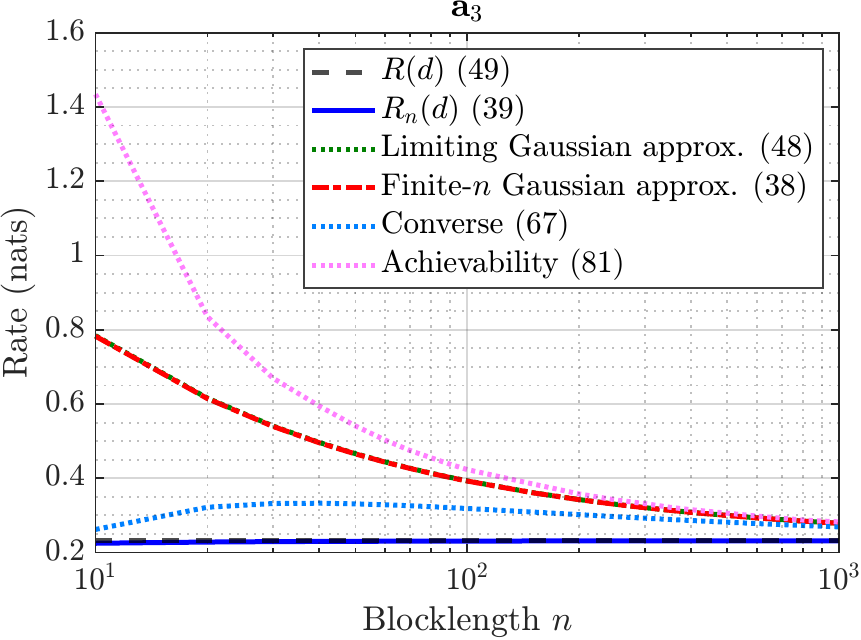}
    }
    
    \caption{Rate versus blocklength with parameters:  $\myVec{a}_3 = \lbrack -0.5 ,0.4\rbrack$, $\sigma^2 = 0.73$, $d = 0.5$, $\epsilon = 10^{-2}.$ The source is stationary.}
    \label{fig:gar_rate_blocklength}
\end{figure}

\FloatBarrier
\fi

Fig.~\ref{fig:gar_rate_blocklength} presents finite-blocklength numerical results for stationary Gaussian autoregressive source. We compare the limiting \ac{rdf} $\mathbb{R}(d)$, the $n$th-order \ac{rdf} $\mathbb{R}_n(d)$, the Gaussian approximation with limiting characteristics in~\eqref{eq:GAR_rnd} and Gaussian approximation with finite $n$th-order characteristics in~\eqref{eq:corr_rnd}, as well as Kostina and Verdú's converse~\eqref{eq:converse} and achievability~\eqref{eq:achievability} bounds. For the Gaussian approximations, we set $O\left(\frac{\log n}{n}\right) = \frac{\log n}{2n}$. 
For achievability, we relax the bound by evaluating the infimum at the optimal output marginal $P_{\myVec{Y^\ast}}$, and compute the exponential through Monte Carlo averaging.
The figure illustrates the finite-blocklength penalty in lossy compression of Gaussian autoregressive sources. In Fig.~\ref{fig:gar_rate_blocklength}, the required rate decreases toward the limiting \ac{rdf} $\mathbb{R}(d)$ as the blocklength $n$ increases. The $n$th-order Gaussian approximation closely matches the limiting Gaussian approximation, while the achievability and converse bounds sandwich the Gaussian approximation and become tighter for larger $n$.

\color{black}

\section{Conclusion} 
\label{sec:conclusion}
In this paper, we extend the second-order rate-distortion analysis of Kostina and Verdú~\cite{kostina12_FixedLength} from \ac{iid} sources to sources with independent but non-identically distributed components (Theorem~\ref{thm:gaussian_approximation}). Independence of the source components and separability of the distortion enable single-letterization of the $d$-tilted information, which allows us to invoke the Berry--Esseen theorem. For Gaussian sources with memory satisfying regularity conditions, Corollary~\ref{corr:gaussian} shows that the second-order term can be computed explicitly from the finite-dimensional eigenspectrum, recovering the familiar reverse water-filling expressions. In Corollary~\ref{corr:gaussAR}, we derive a second-order approximation for Gaussian autoregressive sources in terms of the limiting \ac{rdf} and the source dispersion. To the best of our knowledge, this provides the first unified second-order characterization of both stationary and nonstationary Gaussian autoregressive sources. Our main technical novelties are threefold. First, we introduce a point-mass product proxy measure in the achievability proof to facilitate the typicality arguments for sources with non-identical components. Second, we analyze nonstationary sources by relaxing the conventional uncoupled distortion moment assumption to a weaker coupled one. Third, we use Lidskii's theorem in Lemma~\ref{lemm:GARnon_asymptotic_eigenvalues} to establish convergence bounds for averages of functions of the eigenvalues of the inverse covariance matrix of Gaussian autoregressive sources. These bounds enable us to show that $\mathbb{R}_n(d)$ and $\mathbb{V}_n(d)$ converge to their limiting values $\mathbb{R}(d)$ and $\mathbb{V}(d)$ at rate $O(1/n)$. This lemma refines Gray's result~\cite[Eq.~19]{gray70_Autoregressive} on the asymptotic eigenvalue distribution. Our analysis extends the second-order analysis of the Gauss-Markov source in~\cite{tian19_Stationary} to Gaussian autoregressive sources, both generalizing those results and refining the remainder term to a smaller-order term.

\appendices
\section{Berry--Esseen Theorem}
\label{app:berry_esseen}
\begin{theorem}
\label{thm:berry_Esseen}
Let $W_1, \dots, W_n$ be a collection of independent zero-mean random 
variables. Define the variance $V^2_i \triangleq \mathbb{E}[|W_i|^2]$ and third absolute moment $T_i \triangleq \mathbb{E}[|W_i|^3]$. Let the average variance $V_n^2$ and average third absolute moment $T_n$
\begin{equation}
\label{eq:moment_Esseen}
V_n^2 \triangleq \frac{1}{n} \sum_{i=1}^n V_i^2, 
\quad 
T_n \triangleq \frac{1}{n} \sum_{i=1}^n T_i.
\end{equation}
Then for $n \in \mathbb{N}$ and $B_n = \frac{T_n}{V_n^{3}}$, we have
\begin{equation}
\sup_{t \in \mathbb{R}} 
\left| 
\mathbb{P} \left[ \frac{1}{V_n \sqrt{n}} \sum_{i=1}^n W_i < t \right] - \Phi(t) 
\right| 
\le \frac{B_n}{\sqrt{n}},
\end{equation}
where $\Phi$ is the cumulative distribution function of the standard normal distribution $\mathcal{N}(0,1)$.
\end{theorem}

\begin{corollary}
\label{cor:berry_Esseen}
Let $W_1, \dots, W_n$ be a collection of independent zero-mean random 
variables. Suppose that
\begin{align}
0 < V_{\min} \le V_n \le V_{\max}, \\
T_n \le T_{\max}.
\end{align}
where $V_n, T_n$ are defined in~\eqref{eq:moment_Esseen}. Denote $B_{\max} = \frac{T_{\max}}{V_{\min}^{3}}$. Then, for arbitrary $b > 0$ and for all
\begin{align}
    \tau &> \tau_0 \triangleq (2B_{\max} + b)V_{\max}\sqrt{2\pi} \\
    n &> \frac{\tau^2}{2V^2_{\min} (\log \tau - \log \tau_0)}
\end{align}
it holds that 
\begin{equation}
    \mathbb{P}\left[0 <\sum_{i=1}^{n} W_i < \tau \right ] \geq \frac{b}{\sqrt{n}}
\end{equation}


\end{corollary}

\section{Proof of Corollary~\ref{corr:gaussAR}}
\label{app:proof_gaussAR}
\begin{proof}[Proof of Corollary~\ref{corr:gaussAR}]
By the decorrelation equivalence established in subsection~\ref{subsec:decor}, it suffices to work with the decorrelated process $\myVec{U}$. Let 
\begin{equation}
    \theta_{\max} =  \max_{\omega \in [-\pi, \pi]} S(\omega),
\end{equation}
denote the water level corresponding to the maximum distortion level $d_{\text{flat}}$ in~\eqref{eq:dmax}, where $g(\omega)$ is defined in~\eqref{eq:g}. 
For $d \in (0, d_{\text{flat}})$, the limiting water level satisfies $\theta^\ast \in (0,\theta_{\max})$. Since $\lim_n\theta_n^\ast = \theta^\ast$, there exists a closed interval $[\underline \theta, \bar \theta]\subset(0,\theta_{\max})$ such that for sufficiently large $n$
\begin{equation}
    \label{eq:closed_set}
    \theta_n^\ast,\theta^\ast \in [\underline \theta,\bar \theta].
\end{equation} 
We next give a concentration result that controls the $n$th-letter water level $\theta_n^\ast$ and its convergence to limiting water level $\theta^\ast$.

\begin{lemma}
\label{lem:theta_bound}
    Fix $d\in(0,d_{\text{flat}})$. There exists a constant $C_{5}>0$ such that for all sufficiently large $n$ the $n$-letter water level $\theta_n^\ast$ and the limiting water level $\theta^\ast$ satisfy  
    \begin{equation}
    \label{eq:theta_bound}
        |\theta_n^\ast - \theta^\ast|\le \frac{C_{5}}{n}.
    \end{equation}
\end{lemma}
\begin{proof}
    See Appendix~\ref{sec:proof_theta_bound}.
\end{proof}
\ifshortversion

Lemma~\ref{lemm:GARnon_asymptotic_eigenvalues} provides a convergence bound for a fixed Lipschitz function $F$. However, in the subsequent analysis, we need to consider functions $F_{\theta_n^\ast}$ that vary with $n$ through the water level $\theta_n^\ast$. We therefore extend the bound to this setting using Lemma~\ref{lem:theta_bound}. Suppose that $F_\theta(\mu)$ is uniformly Lipschitz in both arguments for $\theta \in [\underline \theta, \bar \theta]$ in~\eqref{eq:closed_set}, and $\mu \in [\alpha, \beta]$ in~\eqref{eq:mu_i_bounds_non_stationary}. Then we have
\ifdoublecolumn
\begin{align}
    &
    \left|\frac{1}{n} \sum_{i=1}^n F_{\theta_n^\ast}(\mu_{n,i}) - \frac{1}{2\pi}\int_{-\pi}^{\pi}F_{\theta^\ast}(g(\omega))\,d\omega\right|  \\
    &\le \left|\frac{1}{n} \sum_{i=1}^n F_{\theta_n^\ast}(\mu_{n,i})-\frac{1}{n} \sum_{i=1}^n F_{\theta^\ast}( \mu_{n,i})\right| \notag \\
    & \qquad + \left|\frac{1}{n} \sum_{i=1}^n F_{\theta^\ast}(\mu_{n,i})-\frac{1}{2\pi}\int_{-\pi}^{\pi}
    F_{\theta^\ast}(g(\omega))\,d\omega\right|  \\
    & \le \frac{\tilde C_F}{n} \label{eq:eig_concentration_F}
\end{align}
\else
\begin{align}
    \left|\frac{1}{n} \sum_{i=1}^n F_{\theta_n^\ast}(\mu_{n,i}) - \frac{1}{2\pi}\int_{-\pi}^{\pi}F_{\theta^\ast}(g(\omega))\,d\omega\right| 
    &\le \left|\frac{1}{n} \sum_{i=1}^n F_{\theta_n^\ast}(\mu_{n,i})-\frac{1}{n} \sum_{i=1}^n F_{\theta^\ast}( \mu_{n,i})\right| \notag \\
    & \qquad + \left|\frac{1}{n} \sum_{i=1}^n F_{\theta^\ast}(\mu_{n,i})-\frac{1}{2\pi}\int_{-\pi}^{\pi}
    F_{\theta^\ast}(g(\omega))\,d\omega\right|  \\
    & \le \frac{\tilde C_F}{n} \label{eq:eig_concentration_F}
\end{align}
\fi
where~\eqref{eq:eig_concentration_F} follows from Lipschitz continuity, ~\eqref{eq:eig_concentration_rate}, and~\eqref{eq:theta_bound}. 
\fi

We now verify the assumptions of Corollary~\ref{corr:gaussian}. \ifshortversion From equation~\eqref{eq:tilted_gauss}, and using the intervals $\theta \in [\underline \theta, \bar \theta]$ in~\eqref{eq:closed_set}, and $\mu \in [\alpha, \beta]$ in~\eqref{eq:mu_i_bounds_non_stationary} it follows directly that assumptions~\ref{ass:bdd_disp} and~\ref{ass:up_bdd_disp} are satisfied.
\else We first show that dispersion term $\mathbb{V}_n(d)$ is uniformly bounded away from zero. We have 
\begin{align}
\mathbb{V}_n(d) &= \frac{1}{2n} \sum_{i=1}^n \min \left( 1, \left( \frac{\sigma^2}{\theta^\ast_n \mu_{n,i}} \right)^2 \right) \label{eq:VD_lowerbound1} \\
&\ge \frac{1}{2} \min \left( 1, \left( \frac{\sigma^2}{\bar \theta \beta} \right)^2 \right) \label{eq:VD_lowerbound2} \\
&> 0
\end{align}
where~\eqref{eq:VD_lowerbound1} follows from~\eqref{eq:corr_disp} and~\eqref{eq:eig_equivalence}, while~\eqref{eq:VD_lowerbound2} is due to $\theta_n^\ast \le \bar \theta$ and $\mu_{n,i} < \beta$ for all sufficiently large $n$, by~\eqref{eq:closed_set} and~\eqref{eq:mu_i_bounds_non_stationary}, respectively. Thus, assumption~\ref{ass:bdd_disp} is satisfied.

Similarly, we show that the average third absolute central moment of the $d$-tilted information is uniformly bounded. From equation~\eqref{eq:tilted_gauss}, we have 
\begin{align}
\mathbb{T}_n(d)
&= \frac{K_{0}'}{8n}\sum_{i=1}^n \min\left(1,\left(\frac{\sigma_{n,i}^2}{\theta_n^\ast}\right)^3\right) 
\end{align}
where $K_{0}'\triangleq \mathbb{E}[|Z^2-1|^3]$ for $Z\sim\mathcal{N}(0,1)$. Thus, assumption~\ref{ass:up_bdd_disp} is satisfied.
\fi 

We verify assumption~\ref{ass:uni_moments} through the stronger condition~\ref{ass:uni_moments*}, which is easier to check for Gaussian sources. The condition~\ref{ass:uni_moments*} implies assumption~\ref{ass:uni_moments}, which is shown in Appendix~\ref{app:discussion_assumption}. From the optimal reproduction distribution~\cite[Thm. 10.3.3]{cover2006_book}, we know that $U_i - V_i^\ast \sim \mathcal{N}(0, d_i)$ where $d_i = \min(\theta^\ast_n, \sigma_{n,i}^2).$ Hence, under squared error distortion, we have $\mathsf{d}_i(U_i,V_i^\ast) \sim d_i\chi_1^2$, where $\chi_k^2$ denotes chi-squared distribution with $k$ degrees of freedom. The moment generating function of $\mathsf{d}_i(U_i,V_i^\ast) $ is given by
\begin{equation}
    \E{\exp{(t\mathsf{d}_i(U_i, V_i^\ast))}} = (1- 2d_it)^{-1/2} \quad \text{for} \enspace t <\frac{1}{2d_i},
\end{equation}
where the expectation is with respect to the joint optimal distribution $P_{U_i, V_i^\ast}$. 

By~\eqref{eq:closed_set}, we have $\theta_n^\ast < \bar \theta$ for sufficiently large $n$. Thus, we can choose
\begin{equation}
    0 < \kappa_{1} < \frac{1}{2\bar \theta},
\end{equation}
so that $\kappa_{1} < \lambda_n^\ast$, since $\lambda_n^\ast= 1/ (2\theta_n^\ast)$. Moreover, we have
\begin{align}
    \frac{1}{n}\sum_{i=1}^n \mathbb{E}\left[\exp(\kappa_{1} \mathsf{d}_i(U_i,V_i^\ast))\right] &= \frac{1}{n}\sum_{i=1}^n (1- 2d_i\kappa_{1})^{-1/2} \\
    &\leq (1- 2\bar \theta\kappa_{1})^{-1/2} \label{eq:exp_bound_last}
\end{align}
where~\eqref{eq:exp_bound_last} follows from $d_i \leq \theta^\ast_n \leq \bar \theta$ for sufficiently large $n$, and finite since $\kappa_{1} < 1/2\bar \theta$. Thus condition~\ref{ass:uni_moments*} is verified.  

We now compute the curvature in the assumption~\ref{ass:uni_curv}, and show a uniform upper bound. By properties~\ref{property:B} and~\ref{property:L} in Appendix~\ref{app:proof_shell_probability}, we have
\begin{equation}
\label{eq:inverse_curvature}
    \mathbb{R}''(\myVec{U}, \myVec{V}^\ast, d) = \frac{1}{\E{\frac{1}{n}\left|J''_\myVec{V^\ast}(\myVec{U}, \lambda_n^\ast)\right|}}
\end{equation}

The quantity $J_\myVec{V^\ast}(\myVec{U}, \lambda)$ single-letterizes:
\begin{equation}
\label{eq:single_letter_UV}
    J_\myVec{V^\ast}(\myVec{U}, \lambda) = \sum_{i=1}^nJ_{V_i^\ast}(U_i, \lambda)
\end{equation}
where, for $V_i^\ast \sim \mathcal{N}(0, \nu_{n,i})$ with $\nu_{n,i} = \max(0, \sigma_{n,i}^2 - \theta_n^\ast)$
\begin{align}
    J_{V_i^\ast}(u_i, \lambda) &= -\log\mathbb{E} \left[ \exp \left( -\lambda (u_i - V_i^\ast)^2 \right) \right] \\
    &= \frac{\lambda u_i^2}{1 + 2\lambda \nu_{n,i}} + \frac{1}{2}\log(1+2\lambda \nu_{n,i}) \label{eq:J_i}.
\end{align}
Differentiating~\eqref{eq:J_i} with respect to $\lambda$ gives
\begin{align}
J'_{V_i^\ast}(u_i, \lambda) &= \frac{\nu_{n,i}}{a_{n,i}(\lambda)} + \frac{u_i^2}{a_{n,i}^2(\lambda)} \label{eq:J'}\\
J''_{V_i^\ast}(u_i, \lambda) &= -\frac{2\nu_{n,i}^2}{a^2_{n,i}(\lambda)} - \frac{4u_i^2\nu_{n,i}}{a_{n,i}^3(\lambda)} \label{eq:J''}
\end{align}
where $a_{n,i}(\lambda) \triangleq 1+2\lambda \nu_{n,i}$. Define the active and inactive coordinate sets
\begin{equation}
\mathcal{A}_n(\theta)\triangleq\{i\le n:\sigma_{n,i}^2>\theta\},
\qquad
\mathcal{I}_n(\theta)\triangleq\{i\le n:\sigma_{n,i}^2\le\theta\}.
\end{equation}
At $\lambda_n^\ast= 1/ (2\theta_n^\ast)$, $a_{n,i}(\lambda_n^\ast) = \max(1,  \sigma_{n,i}^2/\theta_n^\ast)$. Therefore, evaluating the expectation over the coordinates gives
\begin{align}
\mathbb{E}\left[\frac{1}{n}\bigl|J_{\myVec{V^\ast}}''(\myVec{U},\lambda_n^\ast)\bigr|\right]
&= \frac{1}{n} \sum_{i\in\mathcal{A}_n(\theta_n^\ast)} 2(\theta_n^\ast)^2\left(1 - \frac{(\theta_n^\ast)^2}{\sigma_{n,i}^4}\right) \label{eq:exp_second_der_J},
\end{align}
where we used $\nu_{n,i} = 0$ for $i\in\mathcal{I}_n(\theta_n^\ast)$ and $\nu_{n,i} = \sigma_{n,i}^2-\theta_n^\ast$ for $i\in\mathcal{A}_n(\theta_n^\ast)$. Let
\begin{equation}
    F_{\theta}(\mu) = 2\theta^2\max\left(0, 1-\left(\frac{\theta\mu}{{\sigma^2}}\right)^2\right).
\end{equation}
\ifshortversion
$F_{\theta}(\mu)$ is Lipschitz in both arguments for $\theta \in [\underline \theta, \bar \theta]$ in~\eqref{eq:closed_set}, and $\mu \in [\alpha, \beta]$ in~\eqref{eq:mu_i_bounds_non_stationary}, thus from~\eqref{eq:eig_concentration_F} we have
\else
For every fixed $\theta \in [\underline \theta, \bar \theta]$, the function $F_\theta(\cdot)$ is Lipschitz on the interval $\mu \in [\alpha, \beta]$ in~\eqref{eq:mu_i_bounds_non_stationary} with constant $L_{F_\theta} =4\bar\theta^3/\sigma^2.$ Moreover, for every fixed $\mu \in [\alpha, \beta]$, the function $F_{(\cdot)}(\mu)$ is Lipschitz on $[\underline \theta, \bar \theta]$ with constant $L_r = 4 \bar \theta$. For all sufficiently large $n$, both $\theta^\ast$ and $\theta_n^\ast$ lie in $[\underline \theta, \bar \theta]$ by~\eqref{eq:closed_set}. Thus, we have  
\begin{align}
\label{eq:concentration_J''}
    &\left|\mathbb{E}\left[\frac{1}{n}\bigl|J''_{\myVec{V^\ast}}(\myVec{U},\lambda_n^\ast)\bigr|\right]- \frac{(\theta^\ast)^2}{\pi} \int_{\frac{\sigma^2}{g(\omega)}>\theta^\ast} 1-\left(\frac{\theta^\ast g(\omega)}{{\sigma^2}}\right)^2\,d\omega \right| \\
    &\le \left|\frac{1}{n} \sum_{i=1}^n F_{\theta_n^\ast}(\mu_{n,i})-\frac{1}{n} \sum_{i=1}^n F_{\theta^\ast}( \mu_{n,i})\right| \notag \\
    & \qquad + \left|\frac{1}{n} \sum_{i=1}^n F_{\theta^\ast}(\mu_{n,i})-\frac{1}{2\pi}\int_{-\pi}^{\pi}
    F_{\theta^\ast}(g(\omega))\,d\omega\right|  \\
    & \leq \frac{L_r C_{5}}{n} + \frac{C_L}{n},
\end{align}
where the first term follows from the Lipschitz continuity of $F_{(\cdot)}(\mu)$ and Lemma~\ref{lem:theta_bound}, while the second term follows from Lemma~\ref{lemm:GARnon_asymptotic_eigenvalues} applied to $F_{\theta^\ast}$, with $C_L$ denoting the corresponding constant in~\eqref{eq:eig_concentration_rate}.  Consequently, from~\eqref{eq:concentration_J''}, we obtain
\fi
\ifdoublecolumn
\begin{align}
\label{eq:limiting_curv}
\lim_{n\to\infty}
&\mathbb{E}\left[
\frac{1}{n}
\bigl|J''_{\myVec{V^\ast}}(\myVec{U},\lambda_n^\ast)\bigr|
\right] \notag\\
&=
\frac{(\theta^\ast)^2}{\pi}
\int_{\frac{\sigma^2}{g(\omega)}>\theta^\ast}
\left(1-\left(\frac{\theta^\ast g(\omega)}{\sigma^2}\right)^2\right)d\omega.
\end{align}
\else
\begin{align}
\label{eq:limiting_curv}
\lim_{n\to\infty}
\mathbb{E}\left[
\frac{1}{n}
\bigl|J''_{\myVec{V^\ast}}(\myVec{U},\lambda_n^\ast)\bigr|
\right] =
\frac{(\theta^\ast)^2}{\pi}
\int_{\frac{\sigma^2}{g(\omega)}>\theta^\ast}
\left(1-\left(\frac{\theta^\ast g(\omega)}{\sigma^2}\right)^2\right)d\omega.
\end{align}
\fi
This limit is positive since the set $\{\omega \colon \frac{\sigma^2}{g(\omega)}>\theta^\ast\}$ has positive measure for $\theta^\ast\in(0,\theta_{\max})$. Thus there exists a constant $\kappa_{2}'$ such that for all sufficiently large $n\in\mathbb{N}$
\begin{equation}
\label{eq:J''_lowerbound}
\mathbb{E}\left[\frac{1}{n}\bigl|J_{\myVec{V^\ast}}''(\myVec{U},\lambda_n^\ast)\bigr|\right] \ge \kappa_{2}'.
\end{equation}
Hence, by~\eqref{eq:inverse_curvature} and inequality~\eqref{eq:J''_lowerbound}, assumption~\ref{ass:uni_curv} is satisfied. Therefore, Corollary~\ref{corr:gaussian} holds.


It remains to show that the \ac{rdf} and the source dispersion converge to their limiting functions with error of order $O(1/n)$. From~\eqref{eq:corr_rd}, we have
\begin{align}
    \mathbb{R}_n(d) &= \frac{1}{2n} \sum_{i=1}^n \log \max \left( 1, \frac{\sigma_{n,i}^2}{\theta^\ast_n} \right) \\
    &=\frac{1}{2n} \sum_{i=1}^n \log \max \left( \mu_{n,i}, \frac{\sigma^2}{\theta^\ast_n} \right)
\end{align}
where the second equality follows from $\prod_{i=1}^n\mu_{n,i}=1$, since $\det(\mathsf{A}_n^\top\mathsf{A}_n)=1$. 
Let
\begin{equation}
    F^r_\theta(\mu)\triangleq \frac{1}{2} \log \max\left(\mu, \frac{\sigma^2}{\theta} \right).
\end{equation}
\ifshortversion
$F^r_{\theta}(\mu)$ is Lipschitz in both arguments, thus from~\eqref{eq:eig_concentration_F}, there exists a constant $C_r>0$ such that
\begin{align}
\label{eq:concentration_rd}
    \bigl|\mathbb{R}_n(d)-\mathbb{R}(d)\bigr| & \leq \frac{C_r}{n},
\end{align}
\else
For every fixed $\theta \in [\underline \theta, \bar \theta]$, the function $F^r_\theta(\cdot)$ is Lipschitz on the interval $\mu \in [\alpha, \beta]$ in~\eqref{eq:mu_i_bounds_non_stationary} with constant $L_{F^r_\theta} =\theta/(2\sigma^2).$ Applying Lemma~\ref{lemm:GARnon_asymptotic_eigenvalues} to $F^r_{\theta^\ast}$, we obtain
\begin{equation}
\label{eq:RD_concentration}
    \left|\frac{1}{n} \sum_{i=1}^n F^r_{\theta^\ast}(\mu_{n,i})-\mathbb{R}(d)\right| < \frac{C_L}{n}
\end{equation}
where $C_L$ is the constant in~\eqref{eq:eig_concentration_rate} corresponding to $F^r_{\theta^\ast}$.

Moreover, for every fixed $\mu \in [\alpha, \beta]$, the function $F^r_{(\cdot)}(\mu)$ is Lipschitz on $[\underline \theta, \bar \theta]$ with constant $L_r = 1/(2\underline \theta)$. For all sufficiently large $n$, both $\theta^\ast$ and $\theta_n^\ast$ lie in $[\underline \theta, \bar \theta]$ by~\eqref{eq:closed_set}. Thus, we have  
\begin{align}
\label{eq:concentration_rd}
    \bigl|\mathbb{R}_n(d)-\mathbb{R}(d)\bigr| &\le \left|\frac{1}{n} \sum_{i=1}^n F^r_{\theta_n^\ast}(\mu_{n,i})-\frac{1}{n} \sum_{i=1}^n F^r_{\theta^\ast}( \mu_{n,i})\right| \notag \\
    & \qquad + \left|\frac{1}{n} \sum_{i=1}^n F^r_{\theta^\ast}(\mu_{n,i})-\mathbb{R}(d)\right|  \\
    & \leq \frac{L_r C_{5}}{n} + \frac{C_L}{n},
\end{align}
where the first term follows from the Lipschitz continuity of $F_{(\cdot)}(\mu)$ and Lemma~\ref{lem:theta_bound}, while the second term follows from~\eqref{eq:RD_concentration}.

\fi

\ifshortversion
The same argument applies to the dispersion with the function
\begin{equation}
    F^v_\theta (\mu) \triangleq \frac{1}{2}\min \left( 1, \left( \frac{\sigma^2}{\theta \mu} \right)^2 \right).
\end{equation}
Thus, there exists a constant $C_v>0$ such that
\begin{align}
    \left|\sqrt{\mathbb{V}_n(d)} - \sqrt{\mathbb{V}(d)}\right|^2 
    &\leq \left|\mathbb{V}_n(d) - \mathbb{V}(d)\right| \\
    & \leq \frac{C_v}{n} \label{eq:concentration_vd}
\end{align}
where~\eqref{eq:concentration_vd} follows from~\eqref{eq:eig_concentration_F} with function $F^v_\theta(\mu)$. 
\else
The same argument applies to the dispersion. Let
\begin{equation}
    F_\theta (\mu) \triangleq \frac{1}{2}\min \left( 1, \left( \frac{\sigma^2}{\theta \mu} \right)^2 \right).
\end{equation}

For every fixed $\theta \in [\underline \theta, \bar \theta]$, the function $F_\theta(\cdot)$ is Lipschitz on the interval $\mu \in [\alpha, \beta]$ in~\eqref{eq:mu_i_bounds_non_stationary} with constant $L_{F_\theta} = \theta/\sigma^2.$ Applying Lemma~\ref{lemm:GARnon_asymptotic_eigenvalues} to $F_{\theta^\ast}$, we obtain
\begin{equation}
\label{eq:VD_concentration}
    \left|\frac{1}{n} \sum_{i=1}^n F_{\theta^\ast}(\mu_{n,i})-\mathbb{V}(d)\right| < \frac{C_L}{n}
\end{equation}
where $C_L$ is the constant in~\eqref{eq:eig_concentration_rate} corresponding to $F_{\theta^\ast}$.

Moreover, for every fixed $\mu \in [\alpha, \beta]$, the function $F_{(\cdot)}(\mu)$ is Lipschitz on $[\underline \theta, \bar \theta]$ with constant $L_v = 1/\underline \theta$. For all sufficiently large $n$, both $\theta^\ast$ and $\theta_n^\ast$ lie in $[\underline \theta, \bar \theta]$ by~\eqref{eq:closed_set}.  Thus, we have  
\begin{align}
    \bigl|\mathbb{V}_n(d)-\mathbb{V}(d)\bigr| &\le \left|\frac{1}{n} \sum_{i=1}^n F_{\theta_n^\ast}(\mu_{n,i})-\frac{1}{n} \sum_{i=1}^n F_{\theta^\ast}( \mu_{n,i})\right| \notag \\
    & \qquad + \left|\frac{1}{n} \sum_{i=1}^n F_{\theta^\ast}(\mu_{n,i})-\mathbb{V}(d)\right|  \\
    & \leq \frac{L_v C_{5}}{n} + \frac{C_L}{n},
\end{align}
where the first term follows from the Lipschitz continuity of $F_{(\cdot)}(\mu)$ and Lemma~\ref{lem:theta_bound}, while the second term follows from~\eqref{eq:VD_concentration}.

Thus, there exists a constant $C_v>0$ such that
\begin{align}
    \left|\sqrt{\mathbb{V}_n(d)} - \sqrt{\mathbb{V}(d)}\right|^2 &\leq   \left|\sqrt{\mathbb{V}_n(d)} - \sqrt{\mathbb{V}(d)}\right| \notag \\ &\qquad \times \left|\sqrt{\mathbb{V}_n(d)} + \sqrt{\mathbb{V}(d)}\right| \\
    & = \left|\mathbb{V}_n(d) - \mathbb{V}(d)\right| \\
    & \leq \frac{C_v}{n} \label{eq:concentration_vd}
\end{align} 
\fi Substituting~\eqref{eq:concentration_rd} and \eqref{eq:concentration_vd} into~\eqref{eq:corr_rnd} yields~\eqref{eq:GAR_rnd}.
\end{proof}

\section{Proof of Lemma~\ref{lem:theta_bound}}
\label{sec:proof_theta_bound}
\begin{proof}[Proof of Lemma~\ref{lem:theta_bound}]
Define
\begin{align}
    h_n(\theta) &\triangleq \frac{1}{n} \sum_{i=1}^n \min\left(\theta, \frac{\sigma^2}{\mu_{n,i}} \right),\\
    h(\theta) &\triangleq \frac{1}{2\pi}\int_{-\pi}^\pi \min\left(\theta, \frac{\sigma^2}{g(\omega)}\right)\, d\omega .
\end{align}
Let 
\begin{equation}
    F_\theta(\mu)\triangleq \min(\theta, \sigma^2/\mu)
\end{equation}
For every fixed $\theta \in [\underline \theta, \bar \theta]$, the function $F_\theta(\cdot)$ is Lipschitz on the interval $\mu \in [\alpha, \beta]$ in~\eqref{eq:mu_i_bounds_non_stationary}  with constant $L_{F_\theta} = \theta^2/\sigma^2$. 

For all sufficiently large $n$, both $\theta^\ast$ and $\theta_n^\ast$ lie in $[\underline \theta, \bar \theta]\subset (0,\theta_{\max})$ by~\eqref{eq:closed_set}. Applying Lemma~\ref{lemm:GARnon_asymptotic_eigenvalues} to $F_\theta$, we obtain
\begin{equation}
    \sup_{\theta \in [\underline \theta, \bar \theta]} |h_n(\theta)-h(\theta)| \le \frac{C_F}{n},
\end{equation}
where $C_{F}$ is the constant in~\eqref{eq:eig_concentration_rate}, chosen with the uniform Lipschitz bound $\bar \theta^2/\sigma^2$. 

By~\eqref{eq:corr_distortion} and~\eqref{eq:GAR_distortion_lim}, we have $d=h_n(\theta_n^\ast)=h(\theta^\ast)$. Thus we have
\begin{align}
    \bigl|h(\theta^\ast)-h(\theta_n^\ast)\bigr|
    &= \bigl|h_n(\theta_n^\ast)-h(\theta_n^\ast)\bigr|
    \le \frac{C_F}{n}.
    \label{eq:bound_on_g}
\end{align}
Since $\bar \theta < \theta_{\max}$, the set
\begin{equation}
    \left\{\omega \colon 
    \frac{\sigma^2}{g(\omega)} \geq \bar \theta \right\},
\end{equation}
has positive Lebesgue measure, call $\gamma>0$. Without loss of generality, assume $\theta_n^\ast < \theta^\ast$. Then
\ifdoublecolumn
\begin{align}
    h&(\theta^\ast)-h(\theta_n^\ast) \\
    &= \frac{1}{2\pi}\int_{-\pi}^\pi \min\left(\theta^\ast, \frac{\sigma^2}{g(\omega)}\right) - \min\left(\theta_n^\ast, \frac{\sigma^2}{g(\omega)}\right) d\omega \\
    &\geq \frac{1}{2\pi}\int_{\left\{\omega \colon 
    \sigma^2/g(\omega) \geq \bar \theta \right\} } \left(\theta^\ast - \theta_n^\ast\right) d\omega \\
    &= \frac{\gamma (\theta^\ast - \theta_n^\ast)}{2\pi} 
\end{align}
\else
\begin{align}
    h(\theta^\ast)-h(\theta_n^\ast)
    &= \frac{1}{2\pi}\int_{-\pi}^\pi \min\left(\theta^\ast, \frac{\sigma^2}{g(\omega)}\right) - \min\left(\theta_n^\ast, \frac{\sigma^2}{g(\omega)}\right) d\omega \\
    &\geq \frac{1}{2\pi}\int_{\left\{\omega \colon 
    \sigma^2/g(\omega) \geq \bar \theta \right\} } \left(\theta^\ast - \theta_n^\ast\right) d\omega \\
    &= \frac{\gamma (\theta^\ast - \theta_n^\ast)}{2\pi} 
\end{align}
\fi
Repeating the argument for $\theta_n^\ast \geq \theta^\ast$, yields
\begin{equation}
\label{eq:proof_theta_bound}
    |h(\theta^\ast)-h(\theta_n^\ast)| \geq \frac{\gamma |\theta^\ast - \theta_n^\ast|}{2\pi} 
\end{equation}
Combining~\eqref{eq:bound_on_g} and~\eqref{eq:proof_theta_bound}, we conclude
    \begin{equation}
        |\theta^\ast - \theta_n^\ast| \leq  \frac{2\pi |h(\theta^\ast) - h(\theta_n^\ast)|}{\gamma} \le \frac{2\pi C_F }{\gamma n}.
    \end{equation}
\end{proof}

\section{Proof of Lemma~\ref{lem:shell_probability_lower_bound}}
\label{app:proof_shell_probability}
\begin{proof}
Before proving the lemma, we introduce some notation and collect several properties analogous to those in~\cite{kostina12_FixedLength} for use in the proof. Fix $n\in\mathbb{N}$ and an arbitrary distribution $P_{\myVec{Y}}$ with independent components. Define
\begin{align}
    d_{\min | \myVec{X}, \myVec{Y}} &\triangleq \inf \{d \colon \mathbb{R}_n(\myVec{X}, \myVec{Y}, d) < \infty\}, \\
    d_{\max | \myVec{X}, \myVec{Y}} &\triangleq \frac{1}{n}\sum_{i=1}^n \mathbb{E}\left[\mathsf{d}_i(X_i,Y_i)\right]  \\
    &= \E{\mathsf{\bar d}_{\myVec{Y}, 1, 1}(\myVec{X}, {0})}
\end{align}

Denoting differentiation with respect to $\lambda$ by $(\cdot)'$, the following pointwise properties, which follow by direct calculation as in~\cite{yang99_RedundancyRD}, hold for every $\lambda>0$ and $i\leq n$:
\begin{enumerate}[label=(\Alph*)]
  \item $J'_{Y_i}(x_i,\lambda) = \mathsf{\bar d}_{Y_i,1}(x_i, \lambda)$, \label{property:A}

  \item $J''_{Y_i}(x_i,\lambda) = \left(\mathsf{\bar d}^2_{Y_i, 1}(x_i, \lambda) -  \mathsf{\bar d}_{Y_i, 2}(x_i, \lambda)\right) \leq 0$, \label{property:B}
  
  \item $J'''_{Y_i}(x_i,\lambda) = \mathsf{\bar d}_{Y_i, 3}(x_i, \lambda) -  3\mathsf{\bar d}_{Y_i, 2}(x_i, \lambda)\mathsf{\bar d}_{Y_i, 1}(x_i, \lambda) + 2 \mathsf{\bar d}^3_{Y_i, 1}(x_i, \lambda)$, \label{property:C}
  
  \item $\mathsf{\bar d}'_{Y_i, j}(x_i, \lambda) \leq 0$. \label{property:D}
\end{enumerate}

The above properties, together with Hölder's inequality, yield the following bounds for every $s\in[0,\lambda]$ and $i\leq n$:
\begin{enumerate}[label=(\Alph*), start=5]
  \item $0 \le J'_{Y_i}(x_i, \lambda) \le \mathsf{\bar d}_{Y_i, 1}(x_i, s),$ \label{property:E}

  \item $0 \leq -J''_{Y_i}(x_i,\lambda) \leq \mathsf{\bar d}_{Y_i, 2}(x_i, s)$  \label{property:F}
  
  \item $|J'''_{Y_i}(x_i,\lambda)| \leq 6\mathsf{\bar d}_{Y_i, 3}(x_i, s)$  \label{property:G}
\end{enumerate}

While the preceding properties hold pointwise, establishing further properties involving expectations over the source distribution requires justifying the interchange of differentiation and expectation using integrable bounds. 
We establish integrability using coupled moments. Specifically, suppose that for some $s'\geq 0$, we have
\begin{align}
d_u &\triangleq \mathbb{E}\left[
    \mathsf{\bar d}_{\myVec{Y},1,1}(\myVec{X},s') \right] <\infty.
    \label{eq:integrable_function}
\end{align}
Observe that when $s'=0$, we recover $d_u=d_{\max\mid\myVec{X},\myVec{Y}}$. Consequently, adapting the arguments in~\cite{yang99_RedundancyRD} using integrable bounds at $s'$, we obtain the following properties whenever $d_{\min|\myVec{X},\myVec{Y}}<d_u$:
\begin{enumerate}[label=(\Alph*), start = 8]
    \item $\mathbb{E}\left[\mathsf{\bar d}_{\myVec{Y}, j, 1}(\myVec{X}, \lambda)\right]$ is strictly decreasing for $\lambda \geq s'$ whenever $\mathbb{E}\left[\mathsf{\bar d}_{\myVec{Y}, j, 1}(\myVec{X}, s')\right] <\infty$, \label{property:H}
    
    
    \item $d_{\min | \myVec{X}, \myVec{Y}} = \E{\alpha_\myVec{Y} (\myVec{X})}$ where $\alpha_\myVec{Y} (\myVec{X}) = \essinf \mathsf{d}(\myVec{x}, \myVec{Y})$ \label{property:I} 

\end{enumerate}
Finally, recalling that $\lambda^\ast_{\myVec{X},\myVec{Y}} =-\mathbb{R}'_n(\myVec{X},\myVec{Y},d)$, the following hold for every $\lambda^\ast_{\myVec{X},\myVec{Y}} > s'$ or for every $d \in (d_{\min|\myVec{X},\myVec{Y}}, d_u)$
\begin{enumerate}[label=(\Alph*), start = 10]
  \item $\mathbb{E}\left[\mathsf{\bar d}_{\myVec{Y},1, 1}(\myVec{X}, \lambda^\ast_{\myVec{X}, \myVec{Y}})\right] = d$ \label{property:J}

  \item $\mathbb{E}[J_{\myVec{Y}}(\myVec{X},\lambda^\ast_{\myVec{X}, \myVec{Y}})] - \lambda^\ast_{\myVec{X}, \myVec{Y}}nd= \sup_{\lambda >0}\{\mathbb{E}[J_{\myVec{Y}}(\myVec{X},\lambda)] - \lambda nd\}$ \label{property:K}

  \item 
     $\mathbb{R}''_n\left(\myVec{X}, \myVec{Y}, d \right) = \frac{1}{\mathbb{E}\left[\mathsf{\bar d}_{\myVec{Y}, 2, 1}(\myVec{X}, \lambda^\ast_{\myVec{X}, \myVec{Y}}) - \mathsf{\bar d}_{\myVec{Y}, 1, 2}(\myVec{X}, \lambda^\ast_{\myVec{X}, \myVec{Y}})  \right]}$  whenever $\mathbb{E}\left[ \mathsf{\bar d}_{\myVec{Y},2,1}(\myVec{X},s') \right] < \infty.$ 
     \label{property:L}
\end{enumerate}

Now, we return to the proof of Lemma~\ref{lem:shell_probability_lower_bound}. By assumption~\ref{ass:uni_moments}, we may take $s' = \lambda_n^\ast - \kappa_{1}$ in the definition of $d_u$ in~\eqref{eq:integrable_function}. We show that $d$ is uniformly bounded away from both endpoints. To this end, observe that by properties~\ref{property:B} and~\ref{property:L}, we have
\begin{equation}
\label{eq:uni_curv2}
    \mathbb{R}''(\myVec{X}, \myVec{Y}^\ast, d) = \frac{1}{\E{\frac{1}{n}\left|J''_\myVec{Y^\ast}(\myVec{X}, \lambda_n^\ast)\right|}}
\end{equation}
and from properties~\ref{property:D} and ~\ref{property:G}, assumption~\ref{ass:uni_moments} with Hölder's inequality, we have for all $\lambda \geq \lambda^\ast_n - \kappa_1$
\begin{equation}
    \E{\frac{1}{n}|J'''_{\myVec{Y^\ast}} (\myVec{X},\lambda)|} < 6K_{1}^{1/2}.
\end{equation}
Then by Taylor's theorem for all $\lambda \geq \lambda^\ast_n - \kappa_1$, we have 
\begin{align}
    \E{\frac{1}{n}|J''_{\myVec{Y^\ast}} (\myVec{X},\lambda)|} &\geq \E{\frac{1}{n}|J''_{\myVec{Y^\ast}} (\myVec{X},\lambda_n^\ast)|} - 6K_{1}^{1/2}|\lambda - \lambda_n^\ast| \\
    &\geq \frac{1}{K_2} - 6K_{1}^{1/2}|\lambda - \lambda_n^\ast| \label{eq:uniform_bound_onJ''}
\end{align}
where~\eqref{eq:uniform_bound_onJ''} follows from assumption~\ref{ass:uni_curv} and~\eqref{eq:uni_curv2}. Then, choose $0<\eta<\kappa_{1}/2$ sufficiently small such that
\begin{align}
\label{eq:uni_curv}
    \liminf_{n \rightarrow\infty} \inf_{|\lambda - \lambda_n^\ast| < \eta} \mathbb{E}\left[\frac{1}{n}\bigl|J_{\myVec{Y^\ast}}''(\myVec{X},\lambda)\bigr|\right] \ge \kappa_{2}
\end{align}
for some constant $\kappa_{2} > 0$. Thus, we have
\ifdoublecolumn
\begin{align}
d - &\E{\mathsf{\bar d}_{\myVec{Y}^\ast,1,1}
(\myVec{X},\lambda_n^\ast+\eta)}
\\
&= \E{\mathsf{\bar d}_{\myVec{Y}^\ast,1,1}
(\myVec{X},\lambda_n^\ast)}
- \E{\mathsf{\bar d}_{\myVec{Y}^\ast,1,1}
(\myVec{X},\lambda_n^\ast+\eta)}
\label{eq:gap_identity}\\
&= -\E{\int_0^\eta
\mathsf{\bar d}'_{\myVec{Y}^\ast,1,1}
(\myVec{X},\lambda_n^\ast+t)\,dt}
\label{eq:gap_integral}\\
&= \E{\int_0^\eta \frac{1}{n}
\bigl|J_{\myVec{Y^\ast}}''
(\myVec{X},\lambda_n^\ast+t)\bigr|\,dt}
\label{eq:gap_curvature}\\
&\ge \frac{\kappa_{2}\eta}{2},
\label{eq:gap_uniform}
\end{align}
\else
\begin{align}
d - \E{\mathsf{\bar d}_{\myVec{Y}^\ast,1,1}
(\myVec{X},\lambda_n^\ast+\eta)} &= \E{\mathsf{\bar d}_{\myVec{Y}^\ast,1,1}
(\myVec{X},\lambda_n^\ast)}
- \E{\mathsf{\bar d}_{\myVec{Y}^\ast,1,1}
(\myVec{X},\lambda_n^\ast+\eta)}
\label{eq:gap_identity}\\
&= -\E{\int_0^\eta
\mathsf{\bar d}'_{\myVec{Y}^\ast,1,1}
(\myVec{X},\lambda_n^\ast+t)\,dt}
\label{eq:gap_integral}\\
&= \E{\int_0^\eta \frac{1}{n}
\bigl|J_{\myVec{Y^\ast}}''
(\myVec{X},\lambda_n^\ast+t)\bigr|\,dt}
\label{eq:gap_curvature}\\
&\ge \frac{\kappa_{2}\eta}{2},
\label{eq:gap_uniform}
\end{align}
\fi
for all sufficiently large $n$, where \eqref{eq:gap_identity} follows from property~\ref{property:J}, \eqref{eq:gap_curvature} follows from properties \ref{property:A} and \ref{property:B}, and \eqref{eq:gap_uniform} follows from \eqref{eq:uni_curv} and Fubini's theorem. An analogous argument yields the same bound for the other side. Thus fixing, 
\begin{equation}
\label{eq:Delta}
    0 < \Delta < \frac{\kappa_{2} \eta}{3},
\end{equation}
ensures that $d \in (d_{\min|\myVec{X},\myVec{Y}^\ast} + \frac{3\Delta}{2}, d_u - \frac{3\Delta}{2})$ for all sufficiently large $n$. Define
\begin{align}
\underline{\lambda}_n &\triangleq - \mathbb{R}'_n\left(\myVec{X}, \myVec{Y}^\ast, d + \frac{3\Delta}{2}\right), \\
\bar{\lambda}_n &\triangleq - \mathbb{R}'_n\left(\myVec{X}, \myVec{Y}^\ast, d - \frac{3\Delta}{2}\right), \\
\delta_n &\triangleq  \frac{3\Delta}{2}\sup_{|\alpha| < 3\Delta/2}\mathbb{R}_n''\left(\myVec{X}, \myVec{Y}^\ast, d + \alpha \right), \label{eq:delta} \\
\mu_n'' &\triangleq \frac{1}{n} \E{\bigl|J_{\myVec{Y}^\ast}''(\myVec{X}, \lambda_n^\ast)\bigr|} 
\end{align}
which are well defined due to equation~\eqref{eq:Delta}. By~\eqref{eq:uni_curv}, and properties~\ref{property:B} and \ref{property:L}, $\delta_n$ in \eqref{eq:delta} is uniformly bounded by $\delta \triangleq \frac{3\Delta}{2\kappa_{2}}$. Further define
\begin{align}
\overline{V_n}(\myVec{x})
&\triangleq \frac{1}{n} \sum_{i=1}^n
     \sup_{|\alpha|<\delta}
     \bigl| J_{Y_i^\ast}''(x_i, \lambda_n^\ast + \alpha) \bigr|\\
 \underline {V_n}(\myVec{x}) 
&\triangleq \frac{1}{n} \sum_{i=1}^n
     \inf_{|\alpha|<\delta}
     \bigl| J_{Y_i^\ast}''(x_i, \lambda_n^\ast + \alpha) \bigr|
\end{align}

Next, construct $\mathcal{T}_n$ as the set of all sequences $\myVec{x}$ that satisfy
\begin{align}
\mathsf{\bar d}_{\myVec{Y}^\ast,1,1}(\myVec{x}, {\underline \lambda_n}) &> d + \Delta 
\label{eq:const1} \\
\mathsf{\bar d}_{\myVec{Y}^\ast,1,1}(\myVec{x}, {\bar \lambda_n}) &< d - \Delta
\label{eq:const2} \\
\mathsf{\bar d}_{\myVec{Y}^\ast,3,1}(\myVec{x}, {\underline \lambda_n}) &
\le \E{\mathsf{\bar d}_{\myVec{Y}^\ast,3,1}(\myVec{X}, {\lambda_n^\ast - \kappa_{1}})} + \Delta
\label{eq:const3}\\
\overline{V_n}(\myVec{x})
&\le
\frac{3\mu_n''}{2}
\label{eq:const4} \\
\underline{V_n}(\myVec{x})
&\ge
\frac{\mu_n''}{2}.
\label{eq:const5}
\end{align}
For $\myVec{x} \in \mathcal{T}_n$, from~\eqref{eq:const1}--\eqref{eq:const2}, we have
\begin{equation}
\mathsf{\bar d}_{\myVec{Y}^\ast,1,1}(\myVec{x}, {\bar \lambda_n})
< d <
\mathsf{\bar d}_{\myVec{Y}^\ast,1,1}(\myVec{x}, {\underline \lambda_n}) .
\end{equation}
Moreover,
\begin{align}
d
&=
\E{\mathsf{\bar d}_{\myVec{Y}^\ast,1,1}(\myVec{\hat X}, {\hat \lambda_n^\ast(\myVec{X})})}
\label{eq:distortion_expec} \\
&=
\mathsf{\bar d}_{\myVec{Y}^\ast,1,1}(\myVec{x}, {\hat \lambda_n^\ast(\myVec{x})}),
\label{eq:distortion_single_letter}
\end{align}
where~\eqref{eq:distortion_expec} follows from property~\ref{property:J}, and~\eqref{eq:distortion_single_letter} follows from the choice of proxy measure in~\eqref{eq:proxy}. By property~\ref{property:H}, it follows that
\begin{equation}
    \underline \lambda_n < \hat \lambda_n^\ast(\myVec{x}) < \bar{\lambda}_n.
\end{equation}
Applying Taylor's theorem, we obtain
\begin{equation}
    -\frac{3\Delta}{2}\,\mathbb{R}''_n\left(\myVec{X}, \myVec{Y}^\ast, \bar d \right)
    + \lambda_n^\ast
    < \hat \lambda_n^\ast(\myVec{x})
    < \lambda_n^\ast
    + \frac{3\Delta}{2}\,\mathbb{R}''_n\left(\myVec{X}, \myVec{Y}^\ast, \underline d \right),
\end{equation}
for some $\bar d \in \bigl[d,\, d+\tfrac{3\Delta}{2}\bigr]$ and $\underline d \in \bigl[d-\tfrac{3\Delta}{2},\, d\bigr]$. By~\eqref{eq:delta},
\begin{equation}
\label{eq:lambda_concentration2}
    |\hat \lambda_n^\ast(\myVec{x}) - \lambda_n^\ast| < \delta.
\end{equation}

By the choice of $\Delta$, the expectations of the left-hand sides of~\eqref{eq:const1}--\eqref{eq:const3} satisfy the corresponding inequalities. By assumption~\ref{ass:uni_moments} and Hölder's and Chebyshev's inequalities, the probability of violating any of~\eqref{eq:const1}--\eqref{eq:const3} is $O\left(\frac{1}{n}\right)$.

By Taylor's theorem
\ifdoublecolumn
\begin{align}
    \mathbb{E}&\left[\frac{1}{n}\sum_{i=1}^n
    \sup_{|\alpha|<\delta}
    \bigl| J_{Y_i^\ast}''(X_i, \lambda_n^\ast + \alpha) - J_{Y_i^\ast}''(X_i, \lambda_n^\ast ) \bigr|\right] \\
    &\quad \leq \delta \E{\frac{1}{n}\sum_{i=1}^n\sup_{|\alpha|<\delta} \bigl| J_{Y_i^\ast}'''(X_i, \lambda_n^\ast + \alpha)\bigr|}  \\
    &\quad\leq 6\delta \mathbb{E}\left[
    \mathsf{\bar d}_{\myVec{Y},3,1}(\myVec{X},\lambda_n^\ast - \kappa_{1}) \right] \label{eq:thrid_der} \\
    &\quad\leq 6\delta K_{1}^\frac{1}{2} \label{eq:third_der_upper_bound}
\end{align}
\else
\begin{align}
    \mathbb{E}\left[\frac{1}{n}\sum_{i=1}^n
    \sup_{|\alpha|<\delta}
    \bigl| J_{Y_i^\ast}''(X_i, \lambda_n^\ast + \alpha) - J_{Y_i^\ast}''(X_i, \lambda_n^\ast ) \bigr|\right] & \leq \delta \E{\frac{1}{n}\sum_{i=1}^n\sup_{|\alpha|<\delta} \bigl| J_{Y_i^\ast}'''(X_i, \lambda_n^\ast + \alpha)\bigr|}  \\
    &\leq 6\delta \mathbb{E}\left[
    \mathsf{\bar d}_{\myVec{Y},3,1}(\myVec{X},\lambda_n^\ast - \kappa_{1}) \right] \label{eq:thrid_der} \\
    &\leq 6\delta K_{1}^\frac{1}{2} \label{eq:third_der_upper_bound}
\end{align}
\fi
where~\eqref{eq:thrid_der} follows from property~\ref{property:G} and~\ref{property:D}, while~\eqref{eq:third_der_upper_bound} follows from assumption~\ref{ass:uni_moments}. Consequently, for sufficiently small $\Delta$ (therefore, small enough $\delta$),
\begin{equation}
    \frac{3\mu_n'' }{4}
    \le
    \mathbb{E}\left[\underline{V_n}(\myVec{X})\right]
    \le
    \mathbb{E}\left[\overline{V_n}(\myVec{X})\right]
    \le
    \frac{5\mu_n''}{4}.
\end{equation}
Hence, by property~\ref{property:F}, assumption~\ref{ass:uni_moments}, and Chebyshev's inequality, the probability of violating~\eqref{eq:const4}--\eqref{eq:const5} is also $O\left(\frac{1}{n}\right)$. Therefore,~\eqref{eq:typical} follows from the union bound. 

Next, we bound the moments appearing in the Berry-Esseen Corollary~\ref{cor:berry_Esseen}, stated in Appendix~\ref{app:berry_esseen}, for random variables $W_i = \mathsf{d}_i(x_i, \hat Z_i^\ast) - \E{\mathsf{d}_i(x_i, \hat Z_i^\ast) | \hat X_i = x_i}$. Define
\ifdoublecolumn
\begin{align}
  \mu(\myVec{x}) &\triangleq \frac{1}{n} \sum_{i=1}^n
     \mathbb{E}\left[ \mathsf{d}_i(x_i, \hat Z_i^\ast)\,\big|\, \hat{X}_i = x_i \right]
\\ 
  &= \frac{1}{n} \sum_{i=1}^n
     \mathsf{\bar d}_{Y_i^\ast,1}\left(x_i, \hat \lambda^\ast_n(\myVec{x})\right) \\
 &= d
\\
  V(\myVec{x})
  &\triangleq \frac{1}{n} \sum_{i=1}^n
     \left[ 
       \mathsf{\bar d}_{Y_i^\ast,2}\left(x_i, \hat \lambda^\ast_n(\myVec{x})\right)
       -  \mathsf{\bar d}^2_{Y_i^\ast, 1}\left(x_i, \hat \lambda^\ast_n(\myVec{x})\right)
     \right]
\\
  &= -\,\frac{1}{n} \sum_{i=1}^n
      J_{Y_i^\ast}''\left(x_i, \hat \lambda^\ast_n(\myVec{x})\right)
\\
  T(\myVec{x})
  &\triangleq \frac{1}{n} \sum_{i=1}^n
     \mathbb{E}\Bigg[
       \bigg| \mathsf{d}_i(x_i,\hat Z_i^\ast)
       \notag \\
       & \qquad - \mathbb{E}\left[\mathsf{d}_i(x_i,\hat Z_i^\ast)\mid \hat{X}_i = x_i\right] 
       \bigg|^{3}
       \,\Big|\, \hat{X}_i = x_i
     \Bigg]
\\
  &\le \frac{8}{n} \sum_{i=1}^n
     \mathbb{E}\left[
       \mathsf{d}_i^3(x_i,\hat Z_i^\ast)\,\big|\, \hat X_i = x_i
     \right]
\\
  &= \frac{8}{n} \sum_{i=1}^n
     \mathsf{\bar d}_{Y_i^\ast, 3}(x_i, \hat \lambda^\ast_n(\myVec{x}))
\end{align}
\else
\begin{align}
  \mu(\myVec{x}) &\triangleq \frac{1}{n} \sum_{i=1}^n
     \mathbb{E}\left[ \mathsf{d}_i(x_i, \hat Z_i^\ast)\,\big|\, \hat{X}_i = x_i \right]
\\ 
  &= \frac{1}{n} \sum_{i=1}^n
     \mathsf{\bar d}_{Y_i^\ast,1}\left(x_i, \hat \lambda^\ast_n(\myVec{x})\right) \\
 &= d
\\
  V(\myVec{x})
  &\triangleq \frac{1}{n} \sum_{i=1}^n
     \left[ 
       \mathsf{\bar d}_{Y_i^\ast,2}\left(x_i, \hat \lambda^\ast_n(\myVec{x})\right)
       -  \mathsf{\bar d}^2_{Y_i^\ast, 1}\left(x_i, \hat \lambda^\ast_n(\myVec{x})\right)
     \right]
\\
  &= -\,\frac{1}{n} \sum_{i=1}^n
      J_{Y_i^\ast}''\left(x_i, \hat \lambda^\ast_n(\myVec{x})\right)
\\
  T(\myVec{x})
  &\triangleq \frac{1}{n} \sum_{i=1}^n
     \mathbb{E}\Bigg[
       \bigg| \mathsf{d}_i(x_i,\hat Z_i^\ast) - \mathbb{E}\left[\mathsf{d}_i(x_i,\hat Z_i^\ast)\mid \hat{X}_i = x_i\right] 
       \bigg|^{3}
       \,\Big|\, \hat{X}_i = x_i
     \Bigg]
\\
  &\le \frac{8}{n} \sum_{i=1}^n
     \mathbb{E}\left[
       \mathsf{d}_i^3(x_i,\hat Z_i^\ast)\,\big|\, \hat X_i = x_i
     \right]
\\
  &= \frac{8}{n} \sum_{i=1}^n
     \mathsf{\bar d}_{Y_i^\ast, 3}(x_i, \hat \lambda^\ast_n(\myVec{x}))
\end{align}
\fi
Thus, when $\myVec{x} \in \mathcal{T}_n$, we have
\begin{align}
  \frac{\mu_n''}{2} &\leq V(\myVec{x}) \leq \frac{3\mu_n''}{2}
  \label{eq:vx} \\
  T(\myVec{x})
  &\le
   8\E{\mathsf{\bar d}_{\myVec{Y}^\ast,3,1}(\myVec{X}, {\lambda_n^\ast - \kappa_{1}})} + 8\Delta \label{eq:tx}
\end{align}
where~\eqref{eq:vx} follows from \eqref{eq:const4}, \eqref{eq:const5}, and~\eqref{eq:lambda_concentration2}, while~\eqref{eq:tx} follows from~\eqref{eq:const3}. By inequality~\eqref{eq:uni_curv}, we have $\mu_n'' > 0$, and assumption~\ref{ass:uni_moments} implies that $T(\myVec{x})$ is uniformly bounded for all sufficiently large $n$. Therefore, by the consequence of Berry--Esseen theorem, stated in Corollary~\ref{cor:berry_Esseen}, Appendix~\ref{app:berry_esseen}, the shell probability~\eqref{eq:shell} is on the order of $\Omega\left(\frac{1}{\sqrt{n}}\right)$.
\end{proof}

\section{Proof of Lemma~\ref{lem:tilted_concentration}}
\label{app:proof_tilted_concentration}
\begin{proof}

For all $\myVec{x} \in \mathcal{T}_n$, we have
\ifdoublecolumn
\begin{align}
    \sum_{i=1}^n &\bigl[ J_{Y^\ast_i}(x_i,\hat \lambda_n^\ast (\myVec{x})) 
      - J_{Y^\ast_i}(x_i,\lambda^\ast_n) - \hat \lambda_n^\ast (\myVec{x})d + \lambda^\ast_n d \bigr]
    \\ 
    &= \sup_{|\beta|<\delta} \sum_{i=1}^n
       \bigl[ J_{Y^\ast_i}(x_i,\lambda^\ast_n + \beta) - J_{Y^\ast_i}(x_i,\lambda^\ast_n) - \beta d \bigr] \label{eq:sup_J} \\
    &= \sup_{|\beta|<\delta} \Bigg\{
         \beta \sum_{i=1}^n \bigl(J_{Y^\ast_i}'(x_i,\lambda^\ast_n) - d \bigr) \notag \\
         &\qquad \qquad \qquad + \frac{\beta^2}{2} \sum_{i=1}^n J_{Y^\ast_i}''(x_i,\lambda^\ast_n + \xi_n)
       \Bigg\} \label{eq:taylor_J} \\
    &\le \sup_{|\beta|<\delta} \left\{
         \beta \Sigma'(\myVec{x}) - \frac{\beta^2}{2} \Sigma''(\myVec{x})
       \right\} \label{eq:Sigma_J}  \\
    &\le \frac{(\Sigma'(\myVec{x}))^2}{2 \Sigma''(\myVec{x})} \label{eq:SigmaMax_J}.
\end{align}
\else
\begin{align}
    \sum_{i=1}^n \bigl[ J_{Y^\ast_i}(x_i,\hat \lambda_n^\ast (\myVec{x})) 
      - J_{Y^\ast_i}(x_i,\lambda^\ast_n) - \hat \lambda_n^\ast (\myVec{x})d + \lambda^\ast_n d \bigr]    &= \sup_{|\beta|<\delta} \sum_{i=1}^n
       \bigl[ J_{Y^\ast_i}(x_i,\lambda^\ast_n + \beta) - J_{Y^\ast_i}(x_i,\lambda^\ast_n) - \beta d \bigr] \label{eq:sup_J} \\
    &= \sup_{|\beta|<\delta} \Bigg\{
         \beta \sum_{i=1}^n \bigl(J_{Y^\ast_i}'(x_i,\lambda^\ast_n) - d \bigr) + \frac{\beta^2}{2} \sum_{i=1}^n J_{Y^\ast_i}''(x_i,\lambda^\ast_n + \xi_n)
       \Bigg\} \label{eq:taylor_J} \\
    &\le \sup_{|\beta|<\delta} \left\{
         \beta \Sigma'(\myVec{x}) - \frac{\beta^2}{2} \Sigma''(\myVec{x})
       \right\} \label{eq:Sigma_J}  \\
    &\le \frac{(\Sigma'(\myVec{x}))^2}{2 \Sigma''(\myVec{x})} \label{eq:SigmaMax_J}.
\end{align}
\fi
Here:
\begin{itemize}
    \item \eqref{eq:sup_J} follows from~\eqref{eq:lambda_concentration}, property~\ref{property:K}, and by the identity 
    \begin{equation}
        \mathbb{E}\bigl[J_{\myVec{Y^\ast}}(\myVec{\hat X},\hat \lambda_n^\ast(\myVec{X}))\bigr]
        = \sum_{i=1}^n J_{Y^\ast_i}\bigl(x_i,\hat \lambda_n^\ast(\myVec{x})\bigr),
    \end{equation}
    due to choice of proxy measure in~\eqref{eq:proxy}.
    \item \eqref{eq:taylor_J} follows from Taylor's theorem for some $\xi_n$ satisfying $\lvert \xi_n \rvert \le \delta$.
    \item In~\eqref{eq:Sigma_J}, we defined
    \begin{align}
    \Sigma'(\myVec{x})
        &\triangleq \sum_{i=1}^n \Bigl(J'_{Y^\ast_i}(x_i, \lambda^\ast_n) - d\Bigr), \\
    \Sigma''(\myVec{x})
        &\triangleq \sum_{i=1}^n \inf_{\lvert \alpha \rvert < \delta}
            \left\lvert J''_{Y^\ast_i}(x_i, \lambda^\ast_n + \alpha) \right\rvert .
    \end{align}
    \item Finally,~\eqref{eq:SigmaMax_J} is obtained by maximizing the quadratic expression over $\beta$.
\end{itemize}

From properties~\ref{property:A} and~\ref{property:J},
\begin{equation}
    \mathbb{E}\left[\frac{1}{n}\Sigma'(\myVec{X})\right]
    = \mathbb{E}\left[\frac{1}{n}J'_{\myVec{Y^\ast}}(\myVec{X}, \lambda^\ast_n)\right] - d = 0,
\end{equation}
hence $\Sigma'(\myVec{X})$ is zero mean. Denote
\begin{align}
W_i &\triangleq J'_{Y_i^\ast}(X_i,\lambda_n^\ast)-d_i, \\
V_n' &\triangleq \frac{1}{n} \sum_{i=1}^n \mathrm{Var}\bigl[W_i\bigr].
\end{align}

Let $\epsilon'>0$, and set
\begin{equation}
    b_n \triangleq n^{1/4},
    \qquad
    A_n \triangleq V_n'+\epsilon' .
\end{equation}
By assumption~\ref{ass:uni_moments} and property~\ref{property:A}, $V_n'$ is bounded for all sufficiently large $n$. Hence, there exists a constant
$M'<\infty$ such that $A_n\le M'$ 
\ifdoublecolumn
\begin{equation}
\label{proof:epsbound}
    2b_n\sqrt{A_n\frac{\log n}{n}}
    \le 2\sqrt{M'}\, n^{-1/4}\sqrt{\log n}
    \le \epsilon'
\end{equation}
for all sufficiently large $n$. Therefore, 
\begin{align}
\mathbb P&\left[
    \sum_{i=1}^n W_i
    > \sqrt{A_n n\log n}
\right]
\\
&\le
\mathbb P\left[
    \sum_{i=1}^n W_i
    > \sqrt{A_n n\log n},
    \ \max_{1\le i\le n}|W_i|<b_n
\right] \nonumber \\
&\quad
+
\mathbb P\left[
    \max_{1\le i\le n}|W_i|\ge b_n
\right]  \\
&\le
\exp\left(
    -\frac{3A_n n\log n}
    {6nV_n' + 2b_n\sqrt{A_n n\log n}}
\right)
+
\sum_{i=1}^n \mathbb P\left[|W_i|\ge b_n\right] \label{proof:Bernstein_union} \\
&=
\exp\left(
    -\frac{3A_n \log n}
    {6V_n' + 2b_n\sqrt{A_n\frac{\log n}{n}}}
\right)
+
\sum_{i=1}^n \mathbb P\left[|W_i|^6\ge b_n^6\right]  \\
&\le
\exp\left(
    -\frac{3(V_n'+\epsilon')\log n}
    {6V_n' + 2\epsilon'}
\right)
+ \sum_{i=1}^n \frac{\mathbb E[|W_i|^6]}{b_n^6} \label{proof:markov_epsbound} \\
&\le
\frac{1}{\sqrt n}
+
\sum_{i=1}^n
\frac{\mathbb{E}\left[
    \mathsf{\bar d}_{Y_i, 6}(X_i, \lambda_n^\ast - \kappa_{1}) \right]}{n^{3/2}} \label{proof:propertyD}\\
&\le
\frac{1+K_{1}}{\sqrt n}, \label{proof:ass_dist}
\end{align}
\else
\begin{equation}
\label{proof:epsbound}
    2b_n\sqrt{A_n\frac{\log n}{n}}
    \le 2\sqrt{M'}\, n^{-1/4}\sqrt{\log n}
    \le \epsilon'
\end{equation}
for all sufficiently large $n$. Therefore, 
\begin{align}
\mathbb P\left[
    \sum_{i=1}^n W_i
    > \sqrt{A_n n\log n}
\right] 
&\le
\mathbb P\left[
    \sum_{i=1}^n W_i
    > \sqrt{A_n n\log n},
    \ \max_{1\le i\le n}|W_i|<b_n
\right] +
\mathbb P\left[
    \max_{1\le i\le n}|W_i|\ge b_n
\right]  \\
&\le
\exp\left(
    -\frac{3A_n n\log n}
    {6nV_n' + 2b_n\sqrt{A_n n\log n}}
\right)
+
\sum_{i=1}^n \mathbb P\left[|W_i|\ge b_n\right] \label{proof:Bernstein_union} \\
&=
\exp\left(
    -\frac{3A_n \log n}
    {6V_n' + 2b_n\sqrt{A_n\frac{\log n}{n}}}
\right)
+
\sum_{i=1}^n \mathbb P\left[|W_i|^6\ge b_n^6\right]  \\
&\le
\exp\left(
    -\frac{3(V_n'+\epsilon')\log n}
    {6V_n' + 2\epsilon'}
\right)
+ \sum_{i=1}^n \frac{\mathbb E[|W_i|^6]}{b_n^6} \label{proof:markov_epsbound} \\
&\le
\frac{1}{\sqrt n}
+
\sum_{i=1}^n
\frac{\mathbb{E}\left[
    \mathsf{\bar d}_{Y_i, 6}(X_i, \lambda_n^\ast - \kappa_{1}) \right]}{n^{3/2}} \label{proof:propertyD}\\
&\le
\frac{1+K_{1}}{\sqrt n}, \label{proof:ass_dist}
\end{align}
\fi
for all sufficiently large $n$. Here, \eqref{proof:Bernstein_union} follows from Bernstein's inequality and the union bound, \eqref{proof:markov_epsbound} follows from Markov's inequality and~\eqref{proof:epsbound}, \eqref{proof:propertyD} is due to property~\ref{property:E}, and \eqref{proof:ass_dist} is due to assumption~\ref{ass:uni_moments}. Consequently,
\begin{equation}
\label{eq:Sigma'}
    \mathbb P\left[
        \left|\Sigma'(\mathbf X)\right|
        >
        \sqrt{A_n n\log n}
    \right]
    \le \frac{2(1+K_{1})}{\sqrt n}
\end{equation}
for all sufficiently large $n$. The arguments leading up to~\eqref{eq:Sigma'} deviate from Kostina and Verdú's proof of~\cite[Lemma~5]{kostina12_FixedLength}, in which they apply the Berry--Esseen theorem to show~\eqref{eq:Sigma'}. In our setting, the variance is not guaranteed to be uniformly bounded away from zero. Hence, we choose $A_n$ to be separated from $V_n'$ by a small positive constant. If $V_n'$ is bounded away from zero for all sufficiently large $n$, then Kostina and Verdú's arguments apply and we recover the same remainder term. To handle the case where $V_n'$ may converge to zero, we instead use Bernstein's inequality, at the expense of an additional $\epsilon'$ overhead in the remainder $O\left(\frac{\log}{n}\right)$.

By property~\ref{property:F} and assumption~\ref{ass:uni_moments}, together with the positivity of $\mu_n'' \ge \kappa_{2}$ ensured by inequality~\eqref{eq:uni_curv},
Chebyshev's inequality gives
\begin{align}
\mathbb{P}\Bigl[ \Sigma''(\myVec{X}) < n \tfrac{\mu_n''}{2} \Bigr]
&\le \mathbb{P}\Bigl[ \mathbb{E}[\Sigma''(\myVec{X})] - \Sigma''(\myVec{X})
    > n \tfrac{\mu_n''}{4} \Bigr] \\
&\le \frac{K_{5}''}{n}
\end{align}
for some constant $K_{5}'' >0$. 

Let $G_n$ be the set of $\myVec{x} \in \mathcal{T}^{n}$ satisfying both
\begin{align}
\big( \Sigma'(\myVec{x}) \big)^{2} &\le A_n n \log n \label{eq:vnlogn} \\
\Sigma''(\myVec{x}) &\ge n \frac{\mu_n''}{2}.\label{eq:nmun}
\end{align}
Since $A_n$ is bounded above and $\mu_n'' > \kappa_{2}$ for sufficiently large $n$, there exists a constant $C_4 >0$ such that
\begin{equation}
    \frac{2A_n}{\mu_n''} \le C_4. \label{eq:c4}
\end{equation}
Finally, denoting
\begin{equation}
g(\myVec{x}) = \sum_{i=1}^n J_{Y^\ast_i}(x_{i}, \hat \lambda_n^\ast (\myVec{x})) 
          - \sum_{i=1}^n J_{Y^\ast_i}(x_{i}, \lambda_n^{\ast}) 
          - \big( \hat \lambda_n^\ast (\myVec{x}) - \lambda^\ast_n \big) n d 
\end{equation}
we have
\ifdoublecolumn
\begin{align}
\mathbb{P}&\big[g(\myVec{X}) > C_{4}\log n\big] \\
&= \mathbb{P}\left[
    g(\myVec{X}) > C_{4}\log n,\ 
    g(\myVec{X}) \le \frac{(\Sigma'(\myVec{X}))^{2}}{2\Sigma''(\myVec{X})}
\right] \nonumber\\
&\qquad
+ \mathbb{P}\left[
    g(\myVec{X}) > C_{4}\log n,\ 
    g(\myVec{X}) > \frac{(\Sigma'(\myVec{X}))^{2}}{2\Sigma''(\myVec{X})}
\right] \\
&\le \mathbb{P}\left[
    \frac{(\Sigma'(\myVec{X}))^{2}}{2\Sigma''(\myVec{X})} > C_{4}\log n
\right]
+ \frac{K_{4}}{\sqrt{n}} \label{eq:typical1}\\
&= \mathbb{P}\left[
    \frac{(\Sigma'(\myVec{X}))^{2}}{2\Sigma''(\myVec{X})} > C_{4}\log n,\ 
    \myVec{X}\in G_n
\right] \notag \\
&\qquad
+ \mathbb{P}\left[
    \frac{(\Sigma'(\myVec{X}))^{2}}{2\Sigma''(\myVec{X})} > C_{4}\log n,\ 
    \myVec{X}\notin G_n
\right] + \frac{K_{4}}{\sqrt{n}} \\
&< 0 + \frac{K_{5}'}{\sqrt{n}} + \frac{K_{5}''}{n} + \frac{K_{4}} {\sqrt{n}}\label{eq:final}
\end{align}
\else
\begin{align}
\mathbb{P}\big[g(\myVec{X}) > C_{4}\log n\big]
&= \mathbb{P}\left[
    g(\myVec{X}) > C_{4}\log n,\ 
    g(\myVec{X}) \le \frac{(\Sigma'(\myVec{X}))^{2}}{2\Sigma''(\myVec{X})}
\right]
+ \mathbb{P}\left[
    g(\myVec{X}) > C_{4}\log n,\ 
    g(\myVec{X}) > \frac{(\Sigma'(\myVec{X}))^{2}}{2\Sigma''(\myVec{X})}
\right] \\
&\le \mathbb{P}\left[
    \frac{(\Sigma'(\myVec{X}))^{2}}{2\Sigma''(\myVec{X})} > C_{4}\log n
\right]
+ \frac{K_{4}}{\sqrt{n}} \label{eq:typical1}\\
&= \mathbb{P}\left[
    \frac{(\Sigma'(\myVec{X}))^{2}}{2\Sigma''(\myVec{X})} > C_{4}\log n,\ 
    \myVec{X}\in G_n
\right] 
+ \mathbb{P}\left[
    \frac{(\Sigma'(\myVec{X}))^{2}}{2\Sigma''(\myVec{X})} > C_{4}\log n,\ 
    \myVec{X}\notin G_n
\right] + \frac{K_{4}}{\sqrt{n}} \\
&< 0 + \frac{K_{5}'}{\sqrt{n}} + \frac{K_{5}''}{n} + \frac{K_{4}} {\sqrt{n}}\label{eq:final}
\end{align}
\fi
where \eqref{eq:typical1} follows from~\eqref{eq:typical} and \eqref{eq:final} follows from~\eqref{eq:vnlogn}-\eqref{eq:c4}.   
\end{proof}

\section{Proof of 
Lemma~\ref{lemm:GARnon_asymptotic_eigenvalues}}
\label{sec:proof_GAR_eig}
\begin{proof}
Define
\begin{equation}
p(\omega) \triangleq \sum_{k=0}^\infty a_k e^{-jk\omega},
\end{equation}
so that $g(\omega)=|p(\omega)|^2$. Then
\begin{align}
|g'(\omega)|
&= \left|p'(\omega)\overline{p(\omega)}
+p(\omega)\overline{p'(\omega)}\right| \\
&\le 2|p(\omega)|\,|p'(\omega)|
\label{eq:g_deriv_1} \\
&\le 2\left(\sum_{k=0}^\infty |a_k|\right)
\left(\sum_{k=0}^\infty k|a_k|\right)
\label{eq:g_deriv_2} \\
&\le 2(1+A_0)A_0.
\label{eq:g_deriv_3}
\end{align}
where
\eqref{eq:g_deriv_3} follows from assumption~\eqref{eq:coef_bound}. Hence $g$ is Lipschitz on $[-\pi,\pi]$. 
Since $g$ is periodic,
\begin{align}
I
&\triangleq \frac{1}{2\pi}\int_{-\pi}^{\pi} F\bigl(g(\omega)\bigr)\,d\omega \\
&= \frac{1}{2\pi}\int_{0}^{2\pi} F\bigl(g(\omega)\bigr)\,d\omega.
\end{align}

Let $g_{n,1}\le \cdots \le g_{n,n}$ be the nondecreasing rearrangement of
\begin{equation}
g\left(\frac{2\pi(i-1)}{n}\right), \qquad i=1,\dots,n,
\end{equation}
and define
\begin{equation}
S_n \triangleq \frac{1}{n}\sum_{i=1}^n F(g_{n,i}),
\qquad
Q_n \triangleq \frac{1}{n}\sum_{i=1}^n F(\mu_{n,i}).
\end{equation}
Since $F\circ g$ is Lipschitz, the standard Riemann-sum estimate gives
\begin{align}
|S_n-I|
&\le \frac{2\pi L_F(1+A_0)A_0}{n}.
\label{eq:integral_est}
\end{align}

Although the eigenvalues $\{\mu_{n,i}\}_{i=1}^n$ of $\mathsf{A}_n^T\mathsf{A}_n$ do not generally admit a closed form, we approximate them using 
$\{g_{n,i}\}_{i=1}^n$. To that end, define the cyclic and non-cyclic $k$-shift matrices by $[\mathsf{U}_n^k]_{(i,j)}=\delta_{(i-k)\mod n,j}$ and $[\mathsf{S}_n^k]_{(i,j)}=\delta_{i-k,j}$, respectively, and write
\begin{equation}
\mathsf{A}_n = \sum_{k=0}^{\infty} a_k \mathsf{S}_n^k,
\qquad
\mathsf{C}_n \triangleq \sum_{k=0}^{\infty} a_k \mathsf{U}_n^k.
\end{equation}
The eigenvalues of the circulant matrix $\mathsf{C}_n^T\mathsf{C}_n$ are exactly given by $\{g_{n,i}\}_{i=1}^n$. Thus, we can approximate the eigenvalues $\{\mu_{n,i}\}_{i=1}^n$ using $\{g_{n,i}\}_{i=1}^n$ via Lidskii's theorem stated below.

\begin{theorem}[Lidskii's theorem, {\cite[Thm. 3.4.4]{bhatia1997matrix}}]
\label{thm:Lidskii}
Let $\|\cdot\|_p$ denote the Schatten-$p$ norm and $\mathsf{X},\mathsf{Y}\in\mathbb{C}^{n\times n}$ be Hermitian matrices with eigenvalues
\begin{equation}
\lambda_1(\mathsf{X})\le \cdots \le \lambda_n(\mathsf{X}),
\qquad
\lambda_1(\mathsf{Y})\le \cdots \le \lambda_n(\mathsf{Y}).
\end{equation}
Then
\begin{equation}
\sum_{i=1}^n
\left|
\lambda_i(\mathsf{X})-\lambda_i(\mathsf{Y})
\right|
\le
\|\mathsf{X}-\mathsf{Y}\|_1.
\label{eq:lidskii_general}
\end{equation}
\end{theorem}

Observe that
\begin{align}
\|\mathsf{A}_n-\mathsf{C}_n\|_1
&\le \sum_{k=0}^{\infty}|a_k|
\|\mathsf{U}_n^k-\mathsf{S}_n^k\|_1
\label{eq:AC_1} \\
&\le \sum_{k=0}^{\infty}k|a_k|
\label{eq:AC_2}
\end{align}
where \eqref{eq:AC_2} follows from
$\|\mathsf{U}_n^k-\mathsf{S}_n^k\|_1\le k$. Moreover,
\begin{align}
\|\mathsf{A}_n^T\mathsf{A}_n-\mathsf{C}_n^T\mathsf{C}_n\|_1
&\le
\|\mathsf{A}_n^T(\mathsf{A}_n-\mathsf{C}_n)\|_1
+\|(\mathsf{A}_n^T-\mathsf{C}_n^T)\mathsf{C}_n\|_1
\label{eq:gram_1} \\
&\le
\left(\|\mathsf{A}_n\|_\infty+\|\mathsf{C}_n\|_\infty\right)
\|\mathsf{A}_n-\mathsf{C}_n\|_1
\label{eq:gram_2} \\
&\le
2\left(\sum_{k=0}^{\infty}|a_k|\right)
\left(\sum_{k=0}^{\infty}k|a_k|\right)
\label{eq:gram_3} \\
&\le 2(1+A_0)A_0,
\label{eq:matrix_est}
\end{align}
where \eqref{eq:gram_1} follows from the triangle inequality,
\eqref{eq:gram_2} follows from Hölder's inequality for Schatten norms, and \eqref{eq:matrix_est} follows from assumption~\eqref{eq:coef_bound} again. Hence, 
\begin{align}
|Q_n-S_n|
&\le \frac{L_F}{n}\sum_{i=1}^n |\mu_{n,i}-g_{n,i}|
\label{eq:sum_est_1} \\
&\le \frac{L_F}{n}
\|\mathsf{A}_n^T\mathsf{A}_n-\mathsf{C}_n^T\mathsf{C}_n\|_1
\label{eq:sum_est_2} \\
&\le \frac{2L_F(1+A_0)A_0}{n}.
\label{eq:sum_est}
\end{align}
where \eqref{eq:sum_est_1} follows from the Lipschitz continuity of $F$,
\eqref{eq:sum_est_2} from Lidskii's theorem, and
\eqref{eq:sum_est} from~\eqref{eq:matrix_est}. Combining~\eqref{eq:sum_est} and~\eqref{eq:integral_est},
\begin{align}
|Q_n-I|
&\le |Q_n-S_n|+|S_n-I| \\
&\le \frac{2(1+\pi)L_F(1+A_0)A_0}{n}
\end{align}
proves the lemma. 
\end{proof}

\section{Discussion on Regularity Conditions}
\label{app:discussion_assumption}
In this appendix, we establish several implications among the regularity conditions. We first show that condition~\eqref{eq:uncoupled_moment} implies assumption~\ref{ass:up_bdd_disp} and, for $d < d_{\max}$, also implies assumptions~\ref{ass:uni_moments}-\ref{ass:uni_curv}, even for sources with independent but not necessarily identically distributed components. We then show that condition~\ref{ass:uni_moments*} implies assumption~\ref{ass:uni_moments}.

\subsection{Implication for assumption~\ref{ass:up_bdd_disp}}
By concavity of the logarithm, we have~\cite[Eq. 101]{kostina12_FixedLength}
\begin{equation}
\label{proof:lem:upperbdd_disp_Concavity}
    0 \le \jmath(x_i,d_i) + \lambda_n^\ast d_i \le \lambda_n^\ast \mathbb{E}\left[\mathsf{d}_i(x_i, Y_i^\ast)\right].
\end{equation}
The average third absolute moment of centered variables is
\begin{align}
    T_n
    &= \frac{1}{n}\sum_{i=1}^{n}
    \mathbb{E}\left[\left|\jmath(X_i,d_i)-\mathbb{E}\left[\jmath(X_i,d_i)\right]\right|^3\right] \\
    &\le \frac{8}{n}\sum_{i=1}^{n}\mathbb{E}\left[\left|\jmath(X_i,d_i)\right|^3\right] 
    \label{eq:Tn_65}\\
    \label{eq:Tn_shift_di}
    &\le \frac{32}{n}\sum_{i=1}^{n}\left(\mathbb{E}\left[\left|\jmath(X_i,d_i)+\lambda_n^\ast d_i\right|^3\right]
        + (\lambda_n^\ast)^3 d_i^3 \right) \\
    \label{proof:converse_eq:Tn_final}
    &\le \frac{64(\lambda_n^\ast)^3}{n}\sum_{i=1}^{n}\mathbb{E}\left[\mathsf{d}_i^3(X_i,Y_i^\ast)\right]
\end{align}
where~\eqref{eq:Tn_65} and~\eqref{eq:Tn_shift_di} follow from $|a-b|^3 \leq 4|a|^3 + 4|b|^3$, and~\eqref{proof:converse_eq:Tn_final} is from~\eqref{proof:lem:upperbdd_disp_Concavity} and Jensen's inequality. By condition~\eqref{eq:uncoupled_moment} and Lemma~\ref{lem:bdd_curv}, \eqref{proof:converse_eq:Tn_final} is uniformly bounded for sufficiently large $n$. Hence, condition~\eqref{eq:uncoupled_moment} implies the assumption~\ref{ass:up_bdd_disp}. 
\subsection{Implication for assumption~\ref{ass:uni_moments}}
We have
\begin{align}
    d_{\max|\myVec{X}, \myVec{Y}^\ast} - d &=  \E{\mathsf{\bar d}_{\myVec{Y}^\ast,1,1}(\myVec{X},0)}  - \E{\mathsf{\bar d}_{\myVec{Y}^\ast,1,1}(\myVec{X},\lambda_n^\ast)} \\ 
    &= -  \E{\int_{0}^{\lambda_n^\ast}\mathsf{\bar d}'_{\myVec{Y}^\ast,1,1}(\myVec{X}, \lambda) d\lambda}  \label{eq:dmax_difference2} \\
    &\leq \lambda_n^\ast \E{\mathsf{\bar d}_{\myVec{Y}^\ast,2,1}(\myVec{X}, 0)} \label{eq:dmax_difference3} \\
    &\leq \lambda_n^\ast (K_{1}')^{1/3} \label{eq:dmax_difference4}
\end{align}
where~\eqref{eq:dmax_difference3} follows from properties~\ref{property:A}-\ref{property:B} and~\ref{property:F}, and~\eqref{eq:dmax_difference4} follows from Jensen's inequality and condition~\eqref{eq:uncoupled_moment}. 

By $d<d_{\max}$, there exists $\epsilon'>0$ small enough such that $d + \epsilon' \leq d_{\max|\myVec{X}, \myVec{Y}^\ast}$ for all sufficiently large $n$. Then, choosing $\kappa_{1} = \epsilon'/(K_{1}')^{1/3}$, we have $\lambda_n^\ast \geq \kappa_{1}$ and
\begin{align}
    \mathbb{E}\left[
    \mathsf{\bar d}_{\myVec{Y}^\ast,6,1}(\myVec{X},\lambda_n^\ast - \kappa_{1}) \right] &\leq \mathbb{E}\left[
    \mathsf{\bar d}_{\myVec{Y}^\ast,6,1}(\myVec{X},0) \right] \label{eq:dbar_upperbound_V2}\\
    &\leq K_{1}' \label{eq:dbar_upperbound_V2_2}
\end{align}
where~\eqref{eq:dbar_upperbound_V2} follows from monotonicity in property~\ref{property:H}, while~\eqref{eq:dbar_upperbound_V2_2} follows from condition~\eqref{eq:uncoupled_moment}. Thus condition~\eqref{eq:uncoupled_moment} implies assumption~\ref{ass:uni_moments}.

\subsection{Implication for assumption~\ref{ass:uni_curv}}

Since $d < d_{\max} - \epsilon' $, choose $M_1$ sufficiently large so that $\bar \delta = d_{\max} - \epsilon' - K_{1}^{1/3}/M_1 -d $ is positive, where $K_{1}$ is the constant in assumption~\ref{ass:uni_moments}. Define the event
\begin{equation}
    \mathcal{M}_i = \{\mathsf{d}_i(X_i, Y^\ast_i) < M_1 \} 
\end{equation}
Then we have
\ifdoublecolumn
\begin{align}
    d_{\max} - &\epsilon' - d \\
    &\leq \aver{1}{\E{\mathsf{\bar d}_{Y^\ast_i,1}(X_i, 0) - \mathsf{\bar d}_{Y^\ast_i,1}(X_i, {\lambda_n^\ast})}}  \label{eq:proof_dbar_def} \\ 
    &= \aver{1}{\E{\mathsf{d}_i(X_i, Y^\ast_i) - \mathsf{\bar d}_{Y^\ast_i,1}(X_i, {\lambda_n^\ast})}} \\
    &\leq \aver{1}{\E{\left( \mathsf{d}_i(X_i, Y^\ast_i) - \mathsf{\bar d}_{Y^\ast_i,1}(X_i, {\lambda_n^\ast})\right)\mathbbm{1}[\mathcal{M}_i]}} \notag \\
    & \quad + \aver{1}{\E{\mathsf{d}_i(X_i, Y^\ast_i) \mathbbm{1}[\mathcal{M}_i^c] }} \label{eq:proof_dbar_separate} \\
    &\leq \aver{1}{\E{\left( \mathsf{d}_i(X_i, Y^\ast_i) - \mathsf{\bar d}_{Y^\ast_i,1}(X_i, {\lambda_n^\ast})\right)\mathbbm{1}[\mathcal{M}_i]}} \notag \\
    & \quad + \aver{1}{\left(\E{\mathsf{d}^2_i(X_i, Y^\ast_i)}\right)^{\frac{1}{2}}  \mathbb{P} \left(\left[  \mathsf{d}_i(X_i, Y^\ast_i) \geq M_1 \right] \right)^{\frac{1}{2}} } \label{eq:proof_dbar_hölder} \\
    &\leq \aver{1}{\E{\left( \mathsf{d}_i(X_i, Y^\ast_i) - \mathsf{\bar d}_{Y^\ast_i,1}(X_i, {\lambda_n^\ast})\right)\mathbbm{1}[\mathcal{M}_i]}} \notag \\
    & \quad + \aver{1}{\left(\E{\mathsf{d}^2_i(X_i, Y^\ast_i)}\right)^{\frac{1}{2}} \frac{\left(\E{ \mathsf{d}^2_i(X_i, Y^\ast_i)}\right)^{\frac{1}{2}}}{M_1^2} } \label{eq:proof_dbar_markov} \\
    &\leq \aver{1}{\E{\left( \mathsf{d}_i(X_i, Y^\ast_i) - \mathsf{\bar d}_{Y^\ast_i,1}(X_i, {\lambda_n^\ast})\right) \mathbbm{1}[\mathcal{M}_i]}} \notag \\
    & \quad +\frac{K_{1}^{1/3}}{M_1} \label{eq:proof_dbar_ass}
\end{align}
\else
\begin{align}
    d_{\max} - \epsilon' - d
    &\leq \aver{1}{\E{\mathsf{\bar d}_{Y^\ast_i,1}(X_i, 0) - \mathsf{\bar d}_{Y^\ast_i,1}(X_i, {\lambda_n^\ast})}}  \label{eq:proof_dbar_def} \\ 
    &= \aver{1}{\E{\mathsf{d}_i(X_i, Y^\ast_i) - \mathsf{\bar d}_{Y^\ast_i,1}(X_i, {\lambda_n^\ast})}} \\
    &\leq \aver{1}{\E{\left( \mathsf{d}_i(X_i, Y^\ast_i) - \mathsf{\bar d}_{Y^\ast_i,1}(X_i, {\lambda_n^\ast})\right)\mathbbm{1}[\mathcal{M}_i]}} + \aver{1}{\E{\mathsf{d}_i(X_i, Y^\ast_i) \mathbbm{1}[\mathcal{M}_i^c] }} \label{eq:proof_dbar_separate} \\
    &\leq \aver{1}{\E{\left( \mathsf{d}_i(X_i, Y^\ast_i) - \mathsf{\bar d}_{Y^\ast_i,1}(X_i, {\lambda_n^\ast})\right)\mathbbm{1}[\mathcal{M}_i]}}  + \aver{1}{\left(\E{\mathsf{d}^2_i(X_i, Y^\ast_i)}\right)^{\frac{1}{2}}  \mathbb{P} \left(\left[  \mathsf{d}_i(X_i, Y^\ast_i) \geq M_1 \right] \right)^{\frac{1}{2}} } \label{eq:proof_dbar_hölder} \\
    &\leq \aver{1}{\E{\left( \mathsf{d}_i(X_i, Y^\ast_i) - \mathsf{\bar d}_{Y^\ast_i,1}(X_i, {\lambda_n^\ast})\right)\mathbbm{1}[\mathcal{M}_i]}}  + \aver{1}{\left(\E{\mathsf{d}^2_i(X_i, Y^\ast_i)}\right)^{\frac{1}{2}} \frac{\left(\E{ \mathsf{d}^2_i(X_i, Y^\ast_i)}\right)^{\frac{1}{2}}}{M_1^2} } \label{eq:proof_dbar_markov} \\
    &\leq \aver{1}{\E{\left( \mathsf{d}_i(X_i, Y^\ast_i) - \mathsf{\bar d}_{Y^\ast_i,1}(X_i, {\lambda_n^\ast})\right) \mathbbm{1}[\mathcal{M}_i]}} +\frac{K_{1}^{1/3}}{M_1} \label{eq:proof_dbar_ass}
\end{align}
\fi
where all the expectations are taken with respect to $P_{X_i}\times P_{Y_i^\ast}$. The inequality~\eqref{eq:proof_dbar_def} follows from properties~\ref{property:H} and~\ref{property:J}, \eqref{eq:proof_dbar_separate} follows from the nonnegativity of $\mathsf{\bar d}_{Y^\ast_i,1}(x_i, {\lambda})$ via properties~\ref{property:A} and \ref{property:E}, \eqref{eq:proof_dbar_hölder} and \eqref{eq:proof_dbar_markov} follow from Hölder's and Markov's inequality, respectively, and \eqref{eq:proof_dbar_ass} follows from condition~\eqref{eq:uncoupled_moment}. On the other hand, we have
\ifdoublecolumn
\begin{align}
    |&J''_{Y^\ast_i}(x_i,\lambda_n^\ast)| \notag \\
    &= \mathsf{\bar d}_{Y^\ast_i, 2}(x_i, \lambda_n^\ast) -\mathsf{\bar d}^2_{Y^\ast_i, 1}(x_i, \lambda_n^\ast) \label{eq:proof_J''_def}\\
    &= \frac{\mathbb{E}\left[\left(\mathsf{d}_i(x_i, Y^\ast_i) - \mathsf{\bar d}_{Y^\ast_i, 1}(x_i, \lambda_n^\ast) \right)^2\exp(-\lambda_n^\ast \mathsf{d}_i(x_i, Y^\ast_i) )\right]}{\mathbb{E}\left[\exp(-\lambda_n^\ast \mathsf{d}_i(x_i, Y^\ast_i) )\right]} \\
    &\geq \exp(-C_{1}M_1) \mathbb{E}\left[\left(\mathsf{d}_i(x_i, Y^\ast_i) - \mathsf{\bar d}_{Y^\ast_i, 1}(x_i, \lambda_n^\ast) \right)^2\mathbbm{1}[ \mathcal{M}_i]\right] \label{eq:proof_J''_trunc} \\
    &\geq \exp(-C_{1}M_1) \mathbb{E}\left[\left(\mathsf{d}_i(x_i, Y^\ast_i) - \mathsf{\bar d}_{Y^\ast_i, 1}(x_i, \lambda_n^\ast) \right)\mathbbm{1}[\mathcal{M}_i]\right]^2 \label{eq:proof_J''_Jensen}
\end{align}
\else
\begin{align}
    |J''_{Y^\ast_i}(x_i,\lambda_n^\ast)|
    &= \mathsf{\bar d}_{Y^\ast_i, 2}(x_i, \lambda_n^\ast) -\mathsf{\bar d}^2_{Y^\ast_i, 1}(x_i, \lambda_n^\ast) \label{eq:proof_J''_def}\\
    &= \frac{\mathbb{E}\left[\left(\mathsf{d}_i(x_i, Y^\ast_i) - \mathsf{\bar d}_{Y^\ast_i, 1}(x_i, \lambda_n^\ast) \right)^2\exp(-\lambda_n^\ast \mathsf{d}_i(x_i, Y^\ast_i) )\right]}{\mathbb{E}\left[\exp(-\lambda_n^\ast \mathsf{d}_i(x_i, Y^\ast_i) )\right]} \\
    &\geq \exp(-C_{1}M_1) \mathbb{E}\left[\left(\mathsf{d}_i(x_i, Y^\ast_i) - \mathsf{\bar d}_{Y^\ast_i, 1}(x_i, \lambda_n^\ast) \right)^2\mathbbm{1}[ \mathcal{M}_i]\right] \label{eq:proof_J''_trunc} \\
    &\geq \exp(-C_{1}M_1) \mathbb{E}\left[\left(\mathsf{d}_i(x_i, Y^\ast_i) - \mathsf{\bar d}_{Y^\ast_i, 1}(x_i, \lambda_n^\ast) \right)\mathbbm{1}[\mathcal{M}_i]\right]^2 \label{eq:proof_J''_Jensen}
\end{align}
\fi
where~\eqref{eq:proof_J''_def} follows from property~\ref{property:B}, ~\eqref{eq:proof_J''_trunc} follows from the bounds $0 \leq \lambda_n^\ast \leq C_{1}$ from Lemma~\ref{lem:bdd_curv}, and \eqref{eq:proof_J''_Jensen} is due to Jensen's inequality. Taking expectation of \eqref{eq:proof_J''_Jensen} with respect to $P_{X_i}$, averaging over the components, and using Jensen's inequality together with \eqref{eq:proof_dbar_markov} gives
\begin{align}
\label{eq:check_assumption_v}
   \mathbb{E}\Bigg[&{\frac{1}{n}|J''_{\myVec{Y^\ast}} (\myVec{X},\lambda_n^\ast)|}\Bigg] \geq \exp(-C_{1}M_1) \bar \delta ^2.
\end{align}
By properties~\ref{property:B} and~\ref{property:L}, inequality~\eqref{eq:check_assumption_v} implies assumption~\ref{ass:uni_curv}.

\subsection{Implication for assumption~\ref{ass:uni_moments}}
Let $\kappa_{1} = \kappa_{1}'/2$. For every $\myVec{x} \in \mathcal{\boldsymbol{X}}$, we have
\ifdoublecolumn
\begin{align}
    &\mathsf{\bar d}_{Y_i^\ast,j}(x_i,\lambda_n^\ast - \kappa_{1}) \\
    &=\frac{\mathbb{E}\left[\mathsf{d}^j_i(x_i, Y_i^\ast)\exp(-(\lambda_n^\ast - \kappa_{1}) \mathsf{d}_i(x_i, Y_i^\ast) )\right]}{\mathbb{E}\left[\exp(-(\lambda_n^\ast - \kappa_{1}) \mathsf{d}_i(x_i, Y_i^\ast) )\right]} \label{eq:dbar_upperbound1} \\
    & \leq \frac{\mathbb{E}\left[\mathsf{d}^j_i(x_i, Y_i^\ast)\exp(\kappa_{1} \mathsf{d}_i(x_i, Y_i^\ast))\exp(-\lambda_n^\ast \mathsf{d}_i(x_i, Y_i^\ast) )\right]}{\mathbb{E}\left[\exp(-\lambda_n^\ast \mathsf{d}_i(x_i, Y_i^\ast) )\right]} \label{eq:dbar_upperbound2}  \\
    & \leq C'_j\frac{\mathbb{E}\left[\exp(\kappa_{1}' \mathsf{d}_i(x_i, Y_i^\ast))\exp(-\lambda_n^\ast \mathsf{d}_i(x_i, Y_i^\ast) )\right]}{\mathbb{E}\left[\exp(-\lambda_n^\ast \mathsf{d}_i(x_i, Y_i^\ast) )\right]} \label{eq:dbar_upperbound3} 
\end{align}
\else
\begin{align}
    \mathsf{\bar d}_{Y_i^\ast,j}(x_i,\lambda_n^\ast - \kappa_{1})
    &=\frac{\mathbb{E}\left[\mathsf{d}^j_i(x_i, Y_i^\ast)\exp(-(\lambda_n^\ast - \kappa_{1}) \mathsf{d}_i(x_i, Y_i^\ast) )\right]}{\mathbb{E}\left[\exp(-(\lambda_n^\ast - \kappa_{1}) \mathsf{d}_i(x_i, Y_i^\ast) )\right]} \label{eq:dbar_upperbound1} \\
    & \leq \frac{\mathbb{E}\left[\mathsf{d}^j_i(x_i, Y_i^\ast)\exp(\kappa_{1} \mathsf{d}_i(x_i, Y_i^\ast))\exp(-\lambda_n^\ast \mathsf{d}_i(x_i, Y_i^\ast) )\right]}{\mathbb{E}\left[\exp(-\lambda_n^\ast \mathsf{d}_i(x_i, Y_i^\ast) )\right]} \label{eq:dbar_upperbound2}  \\
    & \leq C'_j\frac{\mathbb{E}\left[\exp(\kappa_{1}' \mathsf{d}_i(x_i, Y_i^\ast))\exp(-\lambda_n^\ast \mathsf{d}_i(x_i, Y_i^\ast) )\right]}{\mathbb{E}\left[\exp(-\lambda_n^\ast \mathsf{d}_i(x_i, Y_i^\ast) )\right]} \label{eq:dbar_upperbound3} 
\end{align}
\fi
where the expectations are with respect to the optimal reproduction marginal $P_{Y_i^\ast}$. Here, \eqref{eq:dbar_upperbound2} follows from the nonnegativity of distortion and $\kappa_{1}'$, while \eqref{eq:dbar_upperbound3} follows from Taylor expansion of the exponential, where $C'_j$ is a finite constant.

Observe that the tilted distribution in~\eqref{eq:dbar_upperbound3} is precisely the optimal conditional distribution $P_{Y_i^\ast | X_i}$ given in~\eqref{eq:optimum_tilt}. Thus, taking expectations with respect to $P_{X_i}$, and averaging over the components, we get
\begin{align}
\mathbb{E}\left[
    \mathsf{\bar d}_{\myVec{Y}^\ast,6,1}(\myVec{X},\lambda_n^\ast - \kappa_{1})\right]
    &\le \frac{C'_6}{n}\sum_{i=1}^n \mathbb{E}\left[\exp(\kappa_{1}' \mathsf{d}_i(X_i,Y_i^\ast))\right] \\
    &\le C'_6K_{1}' \label{eq:dbar_upperbound4}
\end{align}
where~\eqref{eq:dbar_upperbound4} follows from condition~\ref{ass:uni_moments*}. Hence, condition~\ref{ass:uni_moments*} implies the assumption~\ref{ass:uni_moments}.

\section*{Acknowledgement}
\textit{Generative-AI use disclosure:} 
The authors used ChatGPT-6 to assist with refining code for finite-precision arithmetic, exploring and stress-testing proof strategies and intermediate lemmas, checking algebraic and probabilistic calculations, identifying possible logical gaps, and improving exposition. The central mathematical tool, the point-mass product proxy measure, was developed independently by the authors. All proofs in this paper were written and independently verified by the authors, who assume full responsibility for the correctness of all analytical and numerical results.

\bibliographystyle{IEEEtran}
\bibliography{ref.bib}

\end{document}